\RequirePackage{fix-cm}

\documentclass[11pt]{article}

\usepackage[utf8]{inputenc}
\usepackage{authblk}
\usepackage{amsmath}
\usepackage{amssymb}
\usepackage{amsthm}
\usepackage{amsfonts}
\usepackage{graphicx,epstopdf}
\usepackage{xcolor, soul}
\usepackage{enumitem}
\usepackage{setspace}
\usepackage{tcolorbox}
\usepackage{mathtools}
\usepackage{multirow}
\usepackage{color, colortbl}
\usepackage{array}
\usepackage{adjustbox}
\usepackage{mathrsfs}
\usepackage{soul}
\usepackage{hyperref}
\usepackage{array}
\usepackage{subcaption}
\usepackage{flafter}

\usepackage{caption}
\usepackage{rotating}
\usepackage{tikz}
\usepackage{ctable}

\definecolor{Gray}{gray}{0.9}
\newcolumntype{P}{>{\centering\arraybackslash}}

\newtheorem{remark}{Remark}
\newtheorem{definition}{Definition}
\newtheorem{proposition}{Proposition}
\newtheorem{theorem}{Theorem}
\newtheorem{lemma}{Lemma}
\newcommand{\legendxpos}{330}
\newcommand{\legendypos}{200}

\definecolor{srcjsqcolor}{RGB}{255,165,0}

\colorlet{brjsqcolor}{black!50!red}
\colorlet{brjsqstarcolor}{black!50!red}
\colorlet{genjsqcolor}{black!50!red}
\colorlet{genjsqstarcolor}{black!50!red}

\colorlet{brsedcolor}{black!60!green}
\colorlet{brsedstarcolor}{black!60!green}
\colorlet{gensedcolor}{black!60!green}
\colorlet{gensedstarcolor}{black!60!green}

\colorlet{brsewcolor}{black!10!blue}
\colorlet{brsewstarcolor}{black!10!blue}
\colorlet{gensewcolor}{black!10!blue}
\colorlet{gensewstarcolor}{black!10!blue}
\newcommand{\lthickness}{0.9}
\newcommand{\labscal}{1.0}
\tikzstyle{srcjsqstyle}=[draw=srcjsqcolor,dash pattern=,line width= \lthickness pt,line join=round]

\tikzstyle{brjsqstyle}=[draw=brjsqcolor,dash pattern=on 2pt off 2pt on 6pt off 2pt, line width= \lthickness pt,line join=round]
\tikzstyle{brjsqstarstyle}=[draw=brjsqstarcolor,densely dotted,line width= \lthickness pt,line join=round]
\tikzstyle{genjsqstyle}=[draw=genjsqcolor,dash pattern=,line width= \lthickness pt,line join=round]
\tikzstyle{genjsqstarstyle}=[draw=genjsqstarcolor,dashed,line width= \lthickness pt,line join=round]

\tikzstyle{brsedstyle}=[draw=brsedcolor,dash pattern=on 2pt off 2pt on 6pt off 2pt, line width= \lthickness pt,line join=round]
\tikzstyle{brsedstarstyle}=[draw=brsedstarcolor,densely dotted,line width= \lthickness pt,line join=round]
\tikzstyle{gensedstyle}=[draw=gensedcolor,dash pattern=,line width= \lthickness pt,line join=round]
\tikzstyle{gensedstarstyle}=[draw=gensedstarcolor,dashed,line width= \lthickness pt,line join=round]

\tikzstyle{brsewstyle}=[draw=brsewcolor,dash pattern=on 2pt off 2pt on 6pt off 2pt, line width= \lthickness pt,line join=round]
\tikzstyle{brsewstarstyle}=[draw=brsewstarcolor,densely dotted,line width= \lthickness pt,line join=round]
\tikzstyle{gensewstyle}=[draw=gensewcolor,dash pattern=,line width= \lthickness pt,line join=round]
\tikzstyle{gensewstarstyle}=[draw=gensewstarcolor,dashed,line width= \lthickness pt,line join=round]

\definecolor{Gray}{gray}{0.9}

\mathchardef\mhyphen="2D
\def\multiset#1#2{\ensuremath{\left(\kern-.3em\left(\genfrac{}{}{0pt}{}{#1}{#2}\right)\kern-.3em\right)}}

\newcommand{\Scal}{\mathcal{S}}
\newcommand{\Scalb}{\bar{\mathcal S}}
\newcommand{\Dcal}{\mathcal{D}}
\newcommand{\Dvec}{\boldsymbol{D}}
\newcommand{\dvec}{\boldsymbol{d}}
\newcommand{\Acal}{\mathcal{A}}
\newcommand{\Avec}{\boldsymbol{A}}
\newcommand{\avec}{\boldsymbol{a}}
\newcommand{\Acalv}{\vec{\Acal}}
\newcommand{\Avecv}{\vec{\Avec}}
\newcommand{\avecv}{\vec{\avec}}
\newcommand{\Avecn}[1]{\Avec^{\left(#1\right)}}
\newcommand{\avecn}[1]{\avec^{\left(#1\right)}}
\newcommand{\ain}[2]{a_{#1}^{\left(#2\right)}}
\newcommand{\Ain}[2]{A_{#1}^{\left(#2\right)}}

\newcommand{\Tcal}{\mathcal{T}}
\newcommand{\Pcal}{\mathcal{P}}
\newcommand{\Jcal}{\mathcal{J}}

\newcommand{\Qcal}{\mathcal{Q}}
\newcommand{\Bcal}{\mathcal{B}}

\newcommand{\zerovec}{\boldsymbol{0}}
\newcommand{\evec}{\boldsymbol{e}}

\newcommand{\bi}{b_i(\dvec)}
\newcommand{\bj}{b_j(\dvec)}
\newcommand{\bsp}{b_{s+1}(\dvec)}
\newcommand{\rsd}{\rho_{s+1}^{d_{s+1}}}

\newcommand{\li}{\lambda_i^{\mathrm{I}}}
\newcommand{\lb}{\lambda_i^{\mathrm{B}}}
\newcommand{\lil}{\lambda_{\ell}^{\mathrm{I}}}
\newcommand{\lbl}{\lambda_{\ell}^{\mathrm{B}}}

\newcommand{\ri}{r_i^{\mathrm I}(\dvec)}
\newcommand{\rb}{r_i^{\mathrm B}(\dvec)}
\newcommand{\rbj}[1]{r_i^{\mathrm B}\left(\dvec |J=#1\right)}

\newcommand{\genq}{\mathbf{GEN}}
\newcommand{\detq}{\mathbf{DET}}
\newcommand{\indq}{\mathbf{IND}}
\newcommand{\iidq}{\mathbf{IID}}
\newcommand{\srcq}{\mathbf{SRC}}
\newcommand{\sfcq}{\mathbf{SFC}}
\newcommand{\brq}{\mathsf{BR}}
\newcommand{\uniq}{\mathsf{UNI}}

\newcommand{\cida}{\mathbf{CID}}
\newcommand{\clda}{\mathbf{CLD}}
\newcommand{\cda}{\mathbf{CD}}
\newcommand{\nda}{\mathsf{ND}}
\newcommand{\lda}{\mathbf{LD}}
\newcommand{\ida}{\mathbf{ID}}
\newcommand{\qrf}{\mathbf{QRF}}
\newcommand{\qr}{\mathsf{QR}}

\newcommand{\arf}{\mathbf{ARF}}
\newcommand{\ar}{\mathsf{AR}}

\newcommand{\sed}{\mathsf{SED}}
\newcommand{\sew}{\mathsf{SEW}}
\newcommand{\jsq}{\mathsf{JSQ}}
\newcommand{\jiq}{\mathsf{JIQ}}
\newcommand{\fcfs}{\mathsf{FCFS}}
\newcommand{\plcfs}{\mathsf{PLCFS}}
\newcommand{\ps}{\mathsf{PS}}
\newcommand{\fb}{\mathsf{FB}}
\newcommand{\las}{\mathsf{LAS}}
\newcommand{\jfiq}{\mathsf{JFIQ}}
\newcommand{\jfsq}{\mathsf{JFSQ}}
\newcommand{\lsq}{\mathbf{LSQ}}
\newcommand{\genseed}{\mathsf{IPOptD}_{\genq}^{\mathsf{SEED}}}
\newcommand{\genseedq}{\mathsf{IPOptQ}_{\genq}^{\mathsf{SEED}}}
\newcommand{\genseeda}{\mathsf{IPOptA}_{\genq}^{\mathsf{SEED}}}
\newcommand{\ipoptd}[1]{\mathsf{IPOptD}_{#1}}
\newcommand{\ipopta}[1]{\mathsf{IPOptA}_{#1}}
\newcommand{\ipoptq}[1]{\mathsf{IPOptQ}_{#1}}

\newcommand{\indp}{\widetilde{p}}
\newcommand{\srcp}{\widehat{p}}

\newcommand{\pr}{\mathbb{P}}
\newcommand{\ep}{\mathbb{E}}
\newcommand{\DP}[2]{\left\langle #1, #2 \right\rangle}

\usetikzlibrary{decorations.pathreplacing,decorations.markings,shapes,arrows,patterns}

\begin{document}

\title{A General ``Power-of-$d$'' Dispatching Framework for Heterogeneous Systems
}

\author{Jazeem Abdul Jaleel \and
        Sherwin Doroudi \and
        Kristen Gardner \and
        Alexander Wickeham
}

\date{December 2021}

\maketitle

\begin{abstract}
Intelligent dispatching is crucial to obtaining low response times in large-scale systems.
One common scalable dispatching paradigm is the ``power-of-$d$,'' in which the dispatcher queries $d$ servers at random and assigns the job to a server based only on the state of the queried servers.
The bulk of power-of-$d$ policies studied in the literature assume that the system is homogeneous, meaning that all servers have the same speed; meanwhile real-world systems often exhibit server speed heterogeneity.

This paper introduces a general framework for describing and analyzing heterogeneity-aware power-of-$d$ policies.
The key idea behind our framework is that
dispatching policies can make use of server speed information at two decision points: when choosing which $d$ servers to query, and when assigning a job to one of those servers.
Our framework explicitly separates the dispatching policy into a querying rule and an assignment rule; we consider general families of both rule types.

While the strongest assignment rules incorporate both detailed queue-length information and server speed information, these rules typically are difficult to analyze.
We overcome this difficulty by focusing on heterogeneity-aware assignment rules that ignore queue length information beyond idleness status.  In this setting, we analyze mean response time and formulate novel optimization problems for the joint optimization of querying and assignment.  We build upon our optimized policies to develop heuristic queue length-aware dispatching policies.  Our heuristic policies perform well in simulation, relative to policies that have appeared in the literature.

\end{abstract}

\section{Introduction}
\label{sec:Intro}
Large-scale systems are everywhere, and deciding how to dispatch an arriving job to one of the many available servers is crucial to obtaining low response time.
One common scalable dispatching paradigm is the ``power-of-$d$,'' in which the dispatcher queries $d$ servers at random and assigns the job to a server based only on the state of the queried servers.
Such policies incur a much lower communication cost than querying all servers while sacrificing little in the way of performance.
However, many power-of-$d$ policies, such as \textsf{Join the Shortest Queue}-$d$ ($\jsq$-$d$)\footnote{Throughout, names and abbreviations of \emph{individual} rules and policies are rendered in \textsf{sans-serif font}; see Appendix~\ref{app:notation} for a list of individual rules and policies proposed, studied, and/or referenced in this paper.}~\cite{mitzenmacher2001power}, share a notable weakness: they do not account for the fact that, in many modern systems, the servers' speeds are \emph{heterogeneous}. Unfortunately, such heterogeneity-unaware dispatching policies can perform quite poorly in the presence of server heterogeneity~\cite{gardner2020scalable}.
Indeed, it is not straightforward to determine how to dispatch in heterogeneous systems to achieve low mean response times.
For example, it may sometimes be desirable to exclude the slowest classes of servers entirely, yet at other times even the slow servers are needed to maintain the system's stability.

Motivated by the need for dispatching policies that perform well in heterogeneous systems, researchers have designed new policies for this setting.
For example, under the \textsf{Shortest Expected Delay}-$d$ ($\sed$-$d$) policy the dispatcher queries $d$ servers uniformly at random and assigns the arriving job to the queried server at which the job's expected delay (the number of jobs in the queue, scaled by the server's speed) is the smallest~\cite{selen2016steady}.
Under the \textsf{Balanced Routing} ($\brq$) policy, the dispatcher queries $d$ servers with probabilities proportional to the servers' speeds and assigns the arriving job to the queried server with the fewest jobs in the queue~\cite{chen2012asymptotic}.
While both of these policies generally lead to better performance than the fully heterogeneity-unaware $\jsq$-$d$ policy, there is still substantial room for improvement.
Together, $\sed$-$d$ and $\brq$ illustrate a key observation about how to design heterogeneity-aware power-of-$d$ dispatching policies.
There are two decision points at which such policies can use server speed information: when choosing which $d$ servers to query (exploited by $\brq$), and when assigning a job to one of those servers (exploited by $\sed$-$d$).

One of the primary contributions of this paper is the introduction of a general framework to describe and analyze heterogeneity-aware power-of-$d$ policies; we discuss our framework in detail in Section~\ref{sec:model}.
Our framework explicitly separates the dispatching policy into a \emph{querying rule} that determines how to select $d$ servers upon a job's arrival, and an \emph{assignment rule} that determines where among the $d$ queried servers to send the job.
Both $\sed$-$d$ and $\brq$ fit within our framework, as do many other policies that have been proposed and studied in the literature.
For example, recent work has proposed two families of policies that leverage heterogeneity at both decision points by querying fixed numbers of ``fast'' and ``slow'' servers, then probabilistically choosing whether to assign the job to a fast or a slow server based on the idle/busy statuses of the queried servers~\cite{gardner2020scalable}.
One can also imagine designing new policies within our framework; for example, a policy could query $d$ servers probabilistically in proportion to their speeds---as in $\brq$---and then assign the job to the queried server at which its expected delay is smallest---as in $\sed$-$d$ (for more details on how such policies fit into our framework, see Sections~\ref{sec:model} and~\ref{sec:construction}).


Our framework is quite general in the space of querying rules it permits: we allow for any querying rule that is \emph{static} (i.e., ignores past querying and assignment decisions) and \emph{symmetric} (i.e., treats servers of the same speed class identically).
The $\brq$ \emph{querying rule}---viewed separately from the fact that the $\brq$ \emph{dispatching policy} from \cite{chen2012asymptotic} uses $\jsq$ assignment---for example, clearly satisfies these properties.
The $\brq$ querying rule is a member of what we call the \textbf{Independent and Identically Distributed Querying} ($\iidq$) \emph{family} of querying rules.\footnote{Throughout, names and abbreviations of parameterized \emph{families} of rules and policies are rendered in \textbf{bold serif font}; see Appendix~\ref{app:notation} for a list of families of rules and policies proposed, studied, and/or referenced in this paper.} Each specific policy within this family selects each of the $d$ servers independently according to the same distribution over the server speed classes.  That is, the $\iidq$ family of querying rules is parameterized by a set of probabilities that determine the rates at which each server class is queried.

We consider several families of querying rules that satisfy the static and symmetric properties; as is the case for the $\iidq$ family, each family is characterized by its own set of probabilistic parameters that determine how to select the $d$ servers, and different settings for these parameters specify different policies within the family (e.g., one parameter setting of the $\iidq$ querying rule family yields the $\brq$ querying rule, as alluded to above). 
Other examples of querying rule families in the literature include \textbf{Single Random Class} ($\srcq$) ~\cite{mukhopadhyay2016analysis}, under which a single server class is selected probabilistically for each arriving job and all $d$ queried servers are chosen from that class, and \textbf{Deterministic Class Mix} ($\detq$)~\cite{gardner2020scalable}, under which the $d$ queried servers always contain a fixed number of servers of each class.
We also introduce several new families of querying rules that generalize those in the literature in various ways.


Our framework also permits a wide range of \emph{assignment rules}.
For example, included in our framework are the assignment rules such as \textsf{Shortest Expected Delay} ($\sed$) and \textsf{Join the Shortest Queue} ($\jsq$), which when paired with \textsf{Uniform Querying} ($\uniq$)---the querying rule defined in Section \ref{sec:querying} that queries each server with equal probability, regardless of its class---constitute the $\sed$-$d$ and $\jsq$-$d$ \emph{dispatching policies} as they are typically defined in the literature.  The $\sed$ assignment rule is especially attractive as it simultaneously incorporates both detailed queue-length information and server class information when making an assignment decision among the queried servers.  The potentially powerful rules that make use of both class and queue-length information, such as $\sed$, fall within what we call the \textbf{Class and Length Differentiated} ($\clda$) family of assignment rules.  Unfortunately, general $\clda$ assignment rules (including $\sed$) preclude tractable exact performance analysis.  In light of this tractability barrier, we introduce the \textbf{Class and Idleness Differentiated} ($\cida$) family of assignment rules, a subfamily of $\clda$.  The assignment rules in the $\cida$ family eschew detailed queue length information and make assignment decisions based only on the idle/busy statuses and classes (speeds) of the queried servers.  
Even with the information limitations imposed by the $\cida$ family, there is a rich space of reasonable ways to assign jobs among queried servers of different speeds and idle/busy statuses.  While it is natural to favor an idle fast server over a slower server---whether busy or idle---it is less obvious whether a busy fast server or an idle slow server is preferable; it can even be beneficial to occasionally assign jobs to a busy slow server over a busy fast server.  Following our earlier work in \cite{gardner2020scalable}, we make decisions of this sort probabilistically.  As a result, policies within the $\cida$ family of assignment rules are parameterized by the probabilities with which each queried server class is assigned the arriving job.  Specifically, each set of parameters that specifies an assignment rule within $\cida$ encodes a distribution over the classes  for each type of ``scenario'' the dispatcher may confront---in terms of the speed classes of servers queried and their idle/busy statuses. As we show, unlike the dispatching policies driven by $\clda$ assignment, dispatching policies constructed from any static and symmetric querying rule and a $\cida$ assignment rule are amenable to exact analysis (Section~\ref{sec:analysis}).  

In light of the fact that we can---and do---analyze $\cida$-driven dispatching polices, the bulk of this paper (Sections~\ref{sec:analysis}--\ref{sec:numerical}) is devoted to the study of families of \emph{dispatching policies} that are formed by combining one of several families of querying rules (e.g., $\iidq$, $\srcq$) with the $\cida$ family of assignment rules.
Each resulting family constitutes (often infinitely) many possible individual dispatching policies, each of which is specified by a different choice of the probabilistic parameters governing the chosen querying and assignment rule families.
In Section~\ref{sec:optimization} we formulate optimization problems for \emph{jointly} determining the querying and assignment rule parameterizations that yield the lowest mean response time for a given set of system parameters (e.g., arrival rate, server classes, etc.).
To the best of our knowledge, this paper is the first to feature a joint-optimization of the querying and assignment decisions across continuous parameter spaces for both rule types; while our earlier work~\cite{gardner2020scalable} features a joint optimization, that paper considers only the $\detq$ querying family with only two server classes, which yields at most $|\detq|=d+1$ possible querying rules.
In addition to our allowance for continuous spaces of querying and assignment rules, in this paper we allow for any number of server classes, yielding substantially larger and more complicated optimization problems; for details on the sizes of our optimization problems see Appendix~D.
Nonetheless, the problem of selecting an optimal policy from many of the families introduced in this paper is significantly less computationally intensive than the corresponding problem associated with $\detq$-based policies, such as those in~\cite{gardner2020scalable}, because the continuous space of our querying rules allows for purely continuous optimization, obviating the need for combinatorial optimization.
We discuss practical considerations and present a numerical study of the performance of $\cida$-driven dispatching policies in Section~\ref{sec:numerical}.

Understandably, restricting ourselves to the $\cida$ assignment rule family leads to sacrifices in performance; one would expect $\clda$ assignment rules such as $\sed$ to yield lower mean response times when paired with a judiciously chosen querying rule.
At the same time, because of the difficulty of finding exacts mean response times for $\clda$ assignment rules---which make use of both server speed and detailed queue length information---it is also challenging to systematically identify querying rules that perform well in tandem with the $\clda$ assignment rules.
In Section~\ref{sec:construction}, we offer the following heuristic remedy to the problem of finding suitable querying rules to be paired with those assignment rules that are not amenable to tractable analysis: we pair various assignment rules in $\clda$ (e.g., $\sed$) with a querying rule that was jointly-optimized with a $\cida$ assignment rule.
Simulation results demonstrate that these heuristic dispatching policies tend to perform favorably to other policies---both those existing in the literature and the $\cida$-based policies we study in this paper.
Furthermore, our results yield insights about the relative importance of the querying and assignment decisions at different system loads: we observe that at light load the querying decision drives the dispatching policy's performance, whereas at heavy load the assignment decision plays the larger role.

While throughout the paper, we operate under the assumption that job sizes are exponentially distributed, many of our results hold for generally distributed job sizes (see Appendix F for details).
The work presented in this paper is a starting point for the further study of the policies within our framework; to this end, we discuss ample opportunities for future work in Section~\ref{sec:conclusion}.

\section{Literature review}
\label{sec:literature}

In large-scale systems, the power-of-$d$ is the dominant dispatching paradigm; power-of-$d$ policies operate by querying $d$ servers uniformly at random and dispatching an arriving job to one of the queried servers.
The best-known policy within this paradigm is \textsf{Join the Shortest Queue}-$d$ ($\jsq$-$d$), under which a job is dispatched to the server with the shortest queue among the $d$ queried servers.
Response time under $\jsq$-$d$ has been analyzed, under the assumption of homogeneous servers and exponential service times~\cite{mitzenmacher2001power,vvedenskaya1996queueing}.
$\jsq$-$2$ has also been studied in heterogeneous systems with general service times under both the $\fcfs$~\cite{izagirre2014light,zhou2017designing} and \textsf{Processor Sharing} ($\ps$) scheduling rules~\cite{mukhopadhyay2016analysis}.
Variants of $\jsq$-$d$ include $\jsq(d,T)$, under which a job is dispatched to a queried server with workload less than a threshold $T$, and \textsf{Join the Idle Queue}-$d$ ($\jiq$-$d$), which is a special case of $\jsq(d,T)$ with $T=0$~\cite{Hellemans2019WD}.
While power-of-$d$ policies typically are designed for homogeneous systems, several heterogeneity-aware policies akin to $\jsq$-$d$ also have been proposed.
These include \textsf{Shortest Expected Delay}-$d$ ($\sed$-$d$), which uses server speed information to assign a job to a queried server based on the expected waiting time rather than the number of jobs in the queue, and \textsf{Balanced Routing} ($\brq$), which queries $d$ servers with probabilities proportional to their speeds and then uses $\jsq$ assignment~\cite{chen2012asymptotic}. 
Other power-of-$d$-like familes of policies that make use of server speed information include $\mathbf{JIQ}$-($d_F,d_S$) and $\mathbf{JSQ}$-($d_F,d_S$)~\cite{gardner2020scalable}, as well as the \textbf{Hybrid SQ}(2) \textbf{Scheme}, which has been studied under the \textsf{Processor Sharing} ($\ps$) scheduling discipline~\cite{mukhopadhyay2016analysis}.
All of these policies fit within our framework; we will discuss many of them in more detail, in the context of our framework, in the sections that follow.

A different stream of related literature focuses on policies that use information about \emph{all} servers' states when making dispatching decisions; because these policies do not involve querying a subset of the servers, they fall outside of our framework.
The most well-known policy in this category is \textsf{Join the Shortest Queue} ($\jsq$), which is known to minimize mean response time in homogeneous systems with $\fcfs$ scheduling, assuming that service times are independent and identically distributed and have non-decreasing hazard rate~\cite{weber1978optimal,winston1977optimality}.
Mean response time under $\jsq$ has been analyzed approximately under both $\fcfs$ scheduling, assuming exponential service times~\cite{nelson1989approximation}, and $\ps$ scheduling, assuming general service times~\cite{gupta2007analysis}.
\textsf{Join the Idle Queue} ($\jiq$) was proposed as a low-communication alternative to $\jsq$~\cite{lu2011join,wang2018distributed}; again, the analysis assumes homogeneous servers.
More recently, several heterogeneity-aware variants on $\jsq$ and $\jiq$ have been proposed and studied~\cite{stolyar2015pull,zhou2017designing}.
While some of these policies have been shown to stochastically minimize the queue length distribution in heterogeneous systems~\cite{stolyar2015pull}, this does not imply optimality with respect to mean response time.
Indeed, policies within our framework can outperform these heterogeneity-aware policies that use state information from all servers (see, e.g.,~\cite{gardner2020scalable}).

Still other scalable heterogeneity-aware policies have been designed for systems with slightly different modeling assumptions than those we consider in this work.
For example, the $\jfiq$ and $\jfsq$ policies were designed for systems in which jobs have locality constraints (i.e., each job is capable of running on only a subset of the servers)~\cite{weng2020optimal}.
While the \emph{assignment} rules used in these policies are similar to some of the assignment rules that fit within our framework, the $\jfiq$ and $\jfsq$ \emph{dispatching policies} would not be considered part of our framework because they do not involve querying a subset of servers; instead, the dispatcher considers all compatible servers for each arriving job.
Similarly, the \textbf{Local Shortest Queue} ($\lsq$)\label{lsq} family of policies~\cite{vargaftik2020lsq} is orthogonal to our work; these policies assume multiple dispatchers, each of which store a local---possibly out of date---view of server states.
While some of the policies in the $\lsq$ family are quite similar to policies in our framework, the analytical approach and key insights of~\cite{vargaftik2020lsq} are fundamentally different from our work because of the use of out of date information.

Another category of heterogeneity-aware dispatching policies that fall outside our framework includes those policies that are designed specifically for small-scale systems.
Policies in this category use information about all servers' queue lengths---and sometimes more detailed information---when making dispatching decisions~\cite{banawan1989load,bonomi-tc-1990,FENG2005,esa2013,Sethuraman:1999,tantawi1985optimal}.
These policies typically would not be considered scalable, and hence are less applicable to the setting we consider in this paper.
Some policies, such as \textsf{Shortest Expected Delay} and \textsf{Generalized Join the Shortest Queue}, have well-defined power-of-$d$ variants appropriate for large-scale systems.
Thus far, analysis of these policies has focused on systems with only a small number of servers~\cite{banawan1992comparative,selen2016approximate,selen2016steady,whitt1986deciding}; we consider the power-of-$d$ versions of these policies, which do fall within our framework, in later sections.
Further away from our setting is work focusing on the ``slow server problem,'' which asks whether a slow server should be used at all~\cite{koole-scl-1995,larsen81,lin84,luh2002,rubinovitch85,rubinovitch85_stall,rykov09}.
These models consider systems with a central queue, and thus, the policies proposed do not apply to our setting.

\section{Model and framework}
\label{sec:model}
The framework introduced in this paper necessitates a large volume of notation.  Throughout the paper, notation is defined when introduced. Additionally, most of the notation in the paper is summarized in Appendix~\ref{app:notation}

\subsection{Preliminaries}

We consider a system with $k$\label{k} servers.
There are $s$ classes of server speeds,\label{s}
\begin{align}
    \Scal&\equiv\{1,\ldots,s\},\label{eq:scal}
\end{align}
where the number of class-$i$ servers is $k_i$\label{k_i}; let $q_i \equiv k_i/k$\label{q_i} be the fraction of servers belonging to class $i$.  
In the interest of both clarity and tractability, we assume that the size (i.e., service requirement in terms of time) of a job running on a class-$i$ server is exponentially distributed with rate $\mu_i$ (for a discussion of generally distributed service times, see Appendix F).
Classes are indexed in decreasing order of speed, i.e., $\mu_1 > \cdots > \mu_s$. \label{mu_i}We assume that $\displaystyle{\sum_{i=1}^s} \mu_i q_i = 1$.
Jobs arrive to the system as a Poisson process with rate $\lambda k$.\label{lambda}  Except where stated otherwise, we carry out our analysis in the regime where $k\to\infty$ under the assumption of asymptotic independence (see Section~\ref{sec:analysis} for details).

The goal is to minimize the mean response time $\ep[T]$\label{T}, i.e., the end-to-end duration of time from when a job first arrives to the dispatcher until it completes service at one of the servers.  Upon a job's arrival, the dispatcher (i) queries a given number ($d\ll k$)\label{d} of servers according to a \emph{querying rule}, then (ii) sends the job to one of the queried servers according to an  \emph{assignment rule}, at which (iii) the job is queued and/or served according to a work-conserving  \emph{scheduling rule}.
 In this paper, we are primarily interested in elaborating on and analyzing the consequences of the first two rule types---querying and assignment; together these two rules constitute the totality of the \emph{dispatching policy}.  We denote the dispatching policy that uses querying rule $\qr$\label{QR} and assignment rule $\ar$\label{AR} by $\DP{\qr}{\ar}$\label{DP}\label{qrar}.  Our goal is to find dispatching policies (i.e., jointly determine how to query servers and how to assign jobs) that result in low mean response times.  While explicitly determining and evaluating the performance of the optimal policy will be prohibitively difficult, we propose some families of rules that are simple to implement and understand alongside techniques for identifying optimal rules within these families given a particular problem instance.
 
The details of how individual rules function can depend on the parameters of a particular system (i.e., on the number of server classes $s$, the server speeds $\mu_1,\ldots,\mu_s$, the fraction of the total server count constituting each class, the arrival rate, $\lambda$, etc.) and the query count $d$ (which we can take as given). 
A \emph{family} of (querying or assignment) rules, is a collection of individual rules parameterized by a shared set of additional decision variables (e.g., probabilistic parameters indicating which server classes should be queried or which server should be assigned a job given the state of the queried servers).   We are interested in rule families insofar as they allow us to optimize over their parameter spaces in order to find the specific rule that minimizes the mean response time $\ep[T]$ within that family for a given system parameterization. We note that this optimization is performed once for a given system; the same querying rule and assignment rule are then applied throughout the system's lifetime. Even where optimization is prohibitively intractable, we can still set parameter values heuristically in the hope of finding strong policies among those available within a family.

Throughout the paper we use the following convention: the abbreviated names of \emph{individual} (querying, assignment, and scheduling) rules and dispatching policies are rendered in a sans-serif font (e.g., a querying rule $\qr$, an assignment rule $\ar$, and a dispatching policy $\mathsf{DP}$), while those of entire \emph{families} of rules and policies are rendered in a bold serif font (e.g., a querying rule family $\qrf$\label{qrf}, an assignment rule family $\arf$\label{arf}, and a dispatching policy family $\mathbf{DPF}$)\label{dpf}.  Often, we will also denote families of dispatching policies by extending our notation for individual dispatching rules $\DP{\qr}{\ar}$ as follows:  for an individual querying rule $\qr$ and a family of assignment rules $\arf$, let $\DP{\qr}{\arf}\equiv\{\DP{\qr}{\ar}\colon \ar\in\arf\}$ \label{qrarf} be the family of dispatching policies constructed from the individual querying rule $\qr$ in combination with any individual assignment rule $\ar$ belonging to the family $\arf$.  By analogy, for a querying rule family $\qrf$ and individual assignment rule $\ar$, let $\DP{\qrf}{\ar}\equiv\{\DP{\qr}{\ar}\colon \qr\in\qrf\}$. \label{qrfar}  When discussing a family of dispatching policies where neither querying nor assignment is restricted to an individual rule, we write $\DP{\qrf}{\arf}\equiv\{\DP{\qr}{\ar}\colon \qr\in\qrf,\, \ar\in\arf\}$. \label{qrfarf}

 We assume throughout that the sizes of specific jobs are unknown until they are completed, and hence we restrict attention to querying, assignment, and scheduling rules that cannot make use of (i.e., are ``blind'' to) job size information.  We further assume that querying and assignment decisions are made and carried out instantaneously without any overheads; consequently, jobs may not be held at the server for dispatching at some later time.  Under the assumption of exponentially distributed job sizes, our analysis and results hold under all work conserving size-blind scheduling rules.
 Under general service time distributions this is no longer the case; for a discussion of the interaction between service time distributions and scheduling rules, see Appendix F).

\subsection{Overview of querying rules}
\label{sec:querying}

When a job arrives, the dispatcher queries $d$ servers at random according to a \emph{querying rule}.  Throughout this paper, in the interest of tractability, brevity, and simplicity, we restrict attention to those querying rules that are \emph{static} and \emph{symmetric} (properties, which we define below).
\begin{definition}
A querying rule is \emph{static} if each querying decision is made without reference to any kind of state information, i.e., the set of servers queried upon a job's arrival is chosen independently of all past and future querying and assignment decisions.
\end{definition}


Insisting that our querying rules be static is motivated by simplicity, and may preclude some superior querying rules: it is conceivable that there would be some benefit in weighting the likelihood that a server is queried in terms of how recently it was queried (or better yet, in terms of how recently it was assigned a job), which is not possible under static querying rules.  We note in particular that restricting attention to static querying rules precludes round-robin querying (i.e., the rule where all servers would be put into an ordered list, and one would query by going down the list and querying the next $d$ servers at each arrival, cycling back to the beginning of the list after querying the server at the end of the list).   Nevertheless, this restriction comes with an important advantage: static querying rules can be uniquely and unambiguously described in terms of a probability distribution over the set of all $d$-tuples of servers.  By further imposing that our static querying rules also be symmetric (according to the definition that follows), we can simplify these distributions even further.
\begin{definition}
A static querying rule is \emph{symmetric} if it is equally likely to query a set of $d$ servers $U_1$ or $U_2$ whenever $U_1$ and $U_2$ contain the same number of class-$i$ servers for all $i\in\Scal$.
\end{definition}
Essentially, a static symmetric querying rule is one where each query is carried out independently of all others (as with all static querying rules), while no server (respectively combination of servers) is \emph{ex ante} treated any differently than any other server (respectively combination of servers) of the same class (respectively class composition).  As with the restriction to symmetric querying rules, requiring that a querying rule be symmetric may preclude superior dispatching policies.

These restrictions motivate the introduction of some additional notation and terminology.
Let $D_i$\label{D_i} denote the number of class-$i$ servers in a given query, let $\Dvec\equiv(D_1,\ldots,D_s)$\label{Dvec} denote the \emph{class mix}, let $d_i$ and $\dvec\equiv(d_1,\ldots,d_s)$\label{dvec} denote the realizations of the random variable $D_i$ and the random vector $\Dvec$, respectively, and finally let
\begin{align}
    \Dcal&\equiv\{\dvec\colon d_1+\cdots+d_s=d\}\label{eq:dcal}
\end{align}\label{Dcal}
be the set of all possible class mixes $\dvec$ (involving exactly $d$ servers).  Observe that any static symmetric querying rule can be uniquely and unambiguously defined in terms of a distribution over the set of all possible query mixes, $\Dcal$. 
Formally, a querying rule is given by a function $p\colon\Dcal\to[0,1]$ satisfying $\sum_{\dvec\in\Dcal} p(\dvec)=1$. The querying rule selects servers so that $\pr(\Dvec=\dvec)=p(\dvec)$.\label{p}

 We conclude this subsection by introducing the main families of querying rules---in addition to two individual rules---studied in this paper, taking the query count $d$ as given:

\begin{itemize}
    \item The \textbf{General Class Mix} ($\genq$)\label{genq} family consists of all (and only those) querying rules that are static and symmetric.  Note that such querying rules are equally likely to query any combination of $d$ servers that constitute the same query mix $\dvec\in\Dcal$.  The following families are all subsets of $\genq$.
    \item The \textbf{Independent Querying} ($\indq$)\label{indq} family consists of those querying rules in $\genq$ where each of the $d$ servers to be queried is chosen independently according to some (but not necessarily the same) probability distribution over the set of classes $\Scal$.  Consider the following example of a policy in $\indq$ when $s=d=3$: always query at least one class-$1$ server, exactly one class-$2$ server, and either an additional class-$1$ server or a class-$3$ server with equal probability.  Note that we ignore the possibility of a single server being queried more than once, as we are primarily concerned with the setting where the number of servers in each class $k_i\to\infty$.
    \item The \textbf{Independent and Identically Distributed Querying} ($\iidq$) \label{iidq}family consists of those querying rules in $\genq$ where each of the $d$ servers to be queried are chosen independently according to \emph{the same} probability distribution over the set of classes $\Scal$, and hence, the random vector $\Dvec$ is drawn from a multinomial distribution under $\iidq$ querying.  $\iidq$ is a subfamily of $\indq$.
    \item The \textbf{Deterministic Class Mix} ($\detq$)\label{detq} family consists of those querying rules in $\genq$ that always query the same class mix for some fixed class mix $\dvec\in\Dcal$.  $\detq$ is a subfamily of $\indq$.
    \item The \textbf{Single Random Class} \label{srcq}($\srcq$) family consists of those querying rules in $\genq$ that select one of the $s$ server types according to some probability distribution over the set of classes $\Scal$ and then queries $d$ servers all of that class.
    \item The \textbf{Single Fixed Class} \label{sfc}($\sfcq$) family consists of those querying rules that always query $d$ class-$i$ servers for some fixed class $i\in\Scal$.  Such rules essentially discard all servers except those of the chosen class, rendering the system homogeneous.  The $\sfcq$ family consists of only $s$ querying rules and is precisely the intersection of the $\iidq$ and $\detq$ families as well as the intersection of $\srcq$ and any (nonzero) number of the $\indq$, $\iidq$ and $\detq$ families.
    \item The \textsf{Uniform Querying} ($\uniq$)\label{uniq} rule is equally likely to query any combination of $d$ servers. To elaborate, the $\uniq$ querying rule is a member of the $\iidq$ family where each of the $d$ servers queried is a class-$i$ server with a probability equal to the fraction of servers that belong to class $i$ (i.e., with probability $q_i$).
    \item The \textsf{Balanced Routing} ($\brq$)\label{br} rule queries $d$ servers independently, with the probability that any given server is queried being proportional to its speed.  To elaborate, the $\brq$ querying rule is a member of the $\iidq$ family where each of the $d$ servers queried is a class-$i$ server with a probability equal to the fraction of the total system-wide service capacity provided by class-$i$ servers (i.e., with probability $\mu_iq_i$).  
\end{itemize}

    \begin{remark}
    In \cite{chen2012asymptotic}, \textsf{Balanced Routing}  referred to what would be understood in our framework as the \emph{dispatching policy} constructed from (i) what we call the \textsf{Balanced Routing} \emph{querying} rule and (ii) the \textsf{Join the Shortest Queue} \emph{assignment} rule.  From this point forward, in our paper we use the acronym $\brq$ to refer to the \textsf{Balanced Routing} querying rule and not the dispatching policy.
    \end{remark}

Figure~\ref{fig:venn-querying} depicts the set inclusion relationships between querying rule families and individual querying rules described above.

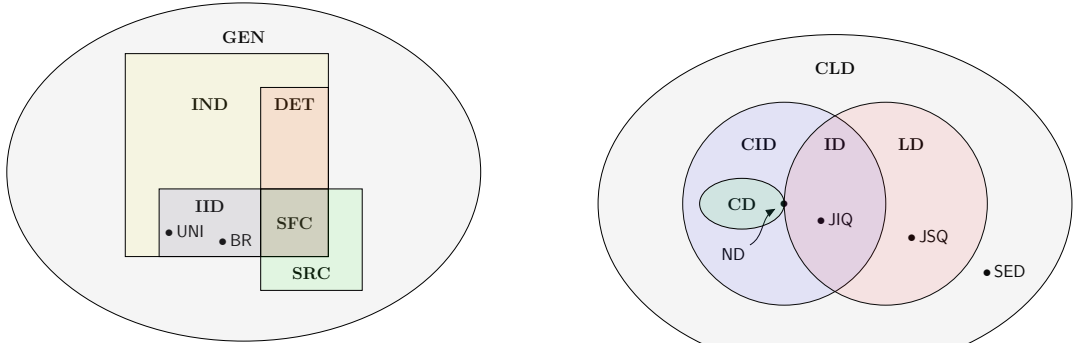
\begin{figure}
\begin{center}
\begin{subfigure}[b]{0.4\textwidth}
    \scalebox{0.45}{
    \begin{tikzpicture}
\node at (4,7) {{\Large $\genq$}};
\node at (3,5) {{\Large $\indq$}};
\node at (5.5,5) {{\Large $\detq$}};
\node at (6,0) {{\Large $\srcq$}};
\node at (3,2) {{\Large $\iidq$}};
\node at (5.5,1.5) {{\Large $\sfcq$}};
\node at (2.25,1.25) {{\Large $\bullet\, \uniq$}};
\node at (3.75,1) {{\Large $\bullet\, \brq$}};
\draw[fill=gray,  thick,fill opacity=0.08] (4,3) ellipse(7 and 5);
\draw[fill=yellow,thick,fill opacity=0.08] (0.5,0.5) rectangle (6.5,6.5);
\draw[fill=blue,  thick,fill opacity=0.08] (1.5,0.5) rectangle (6.5,2.5);
\draw[fill=red,   thick,fill opacity=0.08] (6.5,0.5) rectangle (4.5,5.5);
\draw[fill=green, thick,fill opacity=0.08] (4.5,2.5) rectangle (7.5,-0.5);
\end{tikzpicture}}
    \caption{Querying rule families.  Note that $\sfcq$ is the intersection of any two of the $\iidq$, $\detq$, and $\srcq$ families.  Moreover, we have  $\sfcq=\indq\cap\srcq$.}
\label{fig:venn-querying}
\end{subfigure}
\hspace{1cm}
\begin{subfigure}[b]{0.4\textwidth}
    \scalebox{0.45}{
    \begin{tikzpicture}
\node at (0,4) {{\Large $\clda$}};
\node at (-2.25,1.75) {{\Large $\cida$}};
\node at (2.25,1.75) {{\Large $\lda$}};
\node at (0,1.75) {{\Large $\ida$}};
\node at (-2.75,0) {{\Large $\cda$}};
\node at (0,-0.5) {{\Large $\bullet\, \jiq$}};
\node at (2.75,-1) {{\Large $\bullet\, \jsq$}};
\node at (5,-2) {{\Large $\bullet\, \sed$}};
\node[name=my_label] at (-3,-1.5) {{\Large $\nda$}};
\node[name=my_point,minimum size=0.2] at (-1.5,0) {{\Large $\bullet$}};
\draw[thick,arrows={-triangle 45}] (my_label) to[out=30,in=210] (my_point);
\draw[fill=blue,thick,fill opacity=0.08] (-1.5,0) circle(3);
\draw[fill=red,thick,fill opacity=0.08] (1.5,0) circle(3);
\draw[fill=gray,thick,fill opacity=0.08] (0,0) ellipse(7 and 5);
\draw[fill=green,thick,fill opacity=0.08] (-2.75,0) ellipse(1.25 and 0.75);
\end{tikzpicture}}
    \caption{Assignment rule families. Note that we have $\cda\cap\ida=\{\nda\}$.\ \\ }
\label{fig:venn-assignment}
\end{subfigure}
\end{center}
\caption{Set inclusion diagrams for the querying rule families (a) and assignment rule families (b) discussed in this paper.  In both diagrams rule families are shown as regions and individual rules are shown as points.}
\end{figure}


\subsection{Overview of assignment rules}\label{sec:lbia}

Once a set of servers has been queried, the job is assigned to one of these servers according to an  \emph{assignment rule}, which specifies a distribution over the servers queried. 
Our assignment rules are allowed to depend on \emph{state} information, consisting of knowledge of each queried server's class (and hence, their associated $\mu_i$ and $q_i$ values) and knowledge of the queue length---including the job or jobs in service, if any---at each queried server.
We restrict attention to assignment rules that satisfy restrictions analogous to those adopted for our querying rules.

\begin{definition}
 An assignment rule is \emph{static} if each assignment decision is made without \emph{direct} regard to past querying or assignment decisions (although such decisions can impact the state at a server, which assignment rules may use). 
\end{definition}

\begin{remark}
More formally, let $\boldsymbol{X}_t$ denote the state of the \emph{entire} system at the time of the $t$-th assignment (including the queue length at and class of each of the servers in the system) and let $\vec{\Avec}_t$ denote the result of the $t$-th query (by analogy with the notation $\Avecv$, which we introduce in Section~\ref{sec:clda-formal}).  Let $\mathscr{F}_t$ denote the natural filtration of $\{ \boldsymbol{X}_t, \vec{\Avec}_t \}$.  An assignment policy is \emph{static} if the (potentially random) assignment choice given $\vec{\Avec}_t$ is the same as the assignment choice given $\mathscr{F}_t$.
\end{remark}

\begin{definition}
A static assignment rule is \emph{symmetric} if it does not use information about the specific \emph{identities} of the queried servers and can only use their \emph{state} information. That is, given a set of queried servers with identical states, the job is equally likely to be assigned to any one of those servers and the probability with which that job is assigned to one of those servers depends only on the state (and not the identities) of those servers and the states (and not the identities) of the other queried servers.
\end{definition}

We consider six families of static and symmetric assignment rules.  We proceed to describe these families, which differ from one another in the ways they can differentiate the states of the queried servers for the the purpose of making assignment decisions:
\begin{itemize}
    \item The \textsf{Non-Differentiated} ($\nda$)\label{nda} assignment rule cannot differentiate between server states. This is equivalent to uniform assignment among the servers in the query.  We note that using the $\nda$ assignment rule is antithetical to the purpose of the power-of-$d$ paradigm, as an equivalent dispatching policy can always implemented with $d=1$.
    \item Assignment rules in the \textbf{Class Differentiated} ($\cda$) \label{cda} family may differentiate between server states only on the basis of class information.
    \item Assignment rules in the \textbf{Idleness Differentiated} ($\ida$) \label{ida} family may differentiate between server states only on the basis of idleness information, e.g., \textsf{Join the Idle Queue} ($\jiq$).
    \item Assignment rules in the \textbf{Length Differentiated} ($\lda$) \label{lda} family may differentiate between server states only on the basis of queue-length information, e.g., \textsf{Join the Shortest Queue} ($\jsq$).
    \item Assignment rules in the \textbf{Class and Idleness Differentiated} ($\cida$) \label{cida} family may differentiate between server states only on the basis of class and idleness information.
    \item Assignment rules in the \textbf{Class and Length Differentiated} ($\clda$) \label{clda} family may differentiate between server states on the basis of both class and queue-length information, e.g., \textsf{Shortest Expected Delay} ($\sed$).
\end{itemize}


As shown in Figure~\ref{fig:venn-assignment}, the $\clda$ family includes all of the other assignment rule families under consideration.  Naturally, among the dispatching policies that we consider, those that achieve the best performance (i.e., the lowest mean response time) necessarily make use of the querying rules in the $\clda$ family.
Specific policies that belong to only the $\clda$ family (among the six mentioned above) may be amenable to numerical response time approximation.
However, the curse of dimensionality frequently obstructs the use of optimization techniques for the systematic discovery of strong-performing policies within this family.
Meanwhile, the study of the $\lda$ family can exhibit complications similar to those exhibited by $\clda$, while lacking the advantage of exploiting heterogeneity to obtain low response times.
Therefore, $\cida$---which subsumes $\cda$ and $\ida$---emerges as the richest family under consideration that is amenable to analysis, so we devote Sections~\ref{sec:analysis}--\ref{sec:numerical} to exploring this family of assignment rules (in conjunction with the various families of querying rules introduced in Section~\ref{sec:querying}).  We explore the wider $\clda$ family of assignment rules in Section~\ref{sec:construction}, where we leverage our extensive study of $\cida$-driven dispatching policies (presented in the aforementioned sections) to find superior policies with assignment rules in $\clda$.

\section{Analysis of $\DP{\qrf}{\cida}$}
\label{sec:analysis}
In this section, we examine the $\cida$ family of assignment rules in detail.  We provide a formal presentation of this family (Section~\ref{sec:cida-formal}), prove stability results (Section~\ref{sec:stability}), and present an analysis of the mean response time of the $\DP{\genq}{\cida}$ dispatching policies (Section~\ref{sec:cida-analysis}).

\subsection{Formal presentation of the $\cida$ family of assignment rules}\label{sec:cida-formal}

Assignment rules in the $\cida$ family are---as the family's name clearly suggests---length-blind but idle-aware, i.e., such an assignment rule can observe and make assignment decisions based on the idle/busy status of each of the queried servers, but it cannot observe the queue length at each busy server (of course, the queue length at each idle server must be 0).  By eschewing examining detailed queue length information, we facilitate tractable analysis.  Meanwhile, idle-awareness motivates the introduction of some new notation: we encode the idle/busy statuses of the queried servers by $\avec\equiv(a_1,\ldots,a_s)$\label{avec}, where $a_i$ \label{a_i}is the number of \emph{idle} class-$i$ servers among the $d_i$ queried.  The set of all possible $\avec$ vectors is given by $\Acal\equiv\{\avec\colon a_1+\cdots+a_s\le d\}$.\label{Acal}  Note that $a_i$ and $\avec$ are realizations of the random variable $A_i$\label{A_i} and the random vector $\Avec$\label{Avec} (which are defined analogously to $D_i$ and $\Dvec$), respectively.

Formally, this assignment rule is given by a family of functions $\alpha_i\colon\mathcal A\times\mathcal D\to[0,1]$\label{alpha_i(a,d)} parameterized by $i\in\mathcal S$.  For all $\avec\in\mathcal A$ and $\dvec\in\mathcal D$ such that $\avec\le\dvec$ (element-wise) these families must satisfy $\sum_{i\in\Scal}\alpha_i(\avec,\dvec)=1$ and $\alpha_i(\avec, \dvec)=0$ if $d_i=0$. Given such a family of functions (together with a query resulting in vectors $\avec\in\mathcal A$ and $\dvec\in\mathcal D$) the dispatcher sends the job to a class-$i$ server with probability $\alpha_i(\avec,\dvec)$.  
At this point we assign to an idle class-$i$ server (if possible) or a busy class-$i$ server (otherwise), chosen uniformly at random.

We prune the set of assignment rules by avoiding rules that allow assignment to a slower server when a \emph{faster idle} server has been queried.  That is, $\alpha_i(\avec,\dvec)=0$ whenever there is a class $j<i$ such that \mbox{$a_j\ge1$}.  Moreover, whenever $\avec\neq\zerovec$, the value of $\alpha_i(\avec,\dvec)$ depends only on the realized value of the random variable $J\equiv\min\{j\in\Scal\colon A_{j}>0\}$\label{J}---the class of the \emph{fastest idle queried server}---and on $\dvec$ (specifically, on the realization of the random set $\{j<J\colon d_j>0\}$).  For notational convenience, we take $\min\emptyset\equiv s+1$, so that $J$ is defined on $\Scalb\equiv\Scal\cup\{s+1\}=\{0,1,2,\ldots,s,s+1\}$\label{sbar} and $J=s+1$ when all queried servers are busy, in which case there is no idle server and we can consider the (non-existent) fastest idle queried server as belonging to (the non-existent) class $s+1$.  This structure allows us to introduce the following abuse of notation that will facilitate the discussion of our analysis: $\alpha_i(j,\dvec)\equiv\alpha_i(\avec,\dvec)$\label{alphaijd} for all $j\in\mathcal \Scalb$ and $\avec\in\mathcal A$ such that $J=j$ whenever $\Avec=\avec$, i.e., such that $j=\min\{j'
\in\mathcal S\colon a_{j'}>0\}$.  Note that as a consequence of this notation, we have  $\alpha_i(s+1,\dvec)=\alpha_i(\zerovec,\dvec)$.  Further note that we must have $\alpha_i(j,\dvec)=0$ whenever $d_i=0$ (we cannot send the job to a server that was not queried) and moreover we set $\alpha_i(j,\dvec)=0$ whenever $d_j=0$ and $j\neq s+1$ (the fastest queried idle server must of course be queried).




\subsection{Stability}\label{sec:stability}

In this section, we identify necessary and sufficient conditions for the existence of a stable dispatching policy within the $\DP{\qrf}{\cida}$ family for the various families of querying rules, $\qrf$, presented in Section~\ref{sec:querying}.
We say that the system is stable if the underlying Markov chain is positive recurrent. 
This is a necessary condition for achieving finite mean response time.
In order to establish stability, it is sufficient to show that, when busy, each server experiences an average arrival rate that is less than its service rate. This implies that the mean time between visits to the idle state is finite, and hence that the underlying Markov chain is positive recurrent as required.
Let $\lb$\label{lib} denote the average arrival rate to a busy class-$i$ server.

\begin{definition}
The system is stable if, for all server classes $i\in\Scal$, we have $\lb < \mu_i$.
\end{definition}

\begin{proposition}
\label{prop:stability}
Recalling that $\lambda$ is the average arrival rate per server (i.e., $\lambda k$ is the total arrival rate to the system), the following necessary and sufficient conditions for stability hold:
\begin{enumerate}
    \item There exists a policy in the $\DP{\srcq}{\cida}$ family such that the system is stable if and only if $\lambda < 1$.
    \item There exists a policy in the $\DP{\sfcq}{\cida}$ family such that the system is stable if and only if $\lambda < \max_j \mu_j q_j$.
    \item\label{item:det} Consider a dispatching policy in the $\DP{\detq}{\cida}$ family, where the query mix is always $\dvec$ (note that each individual policy within $\DP{\detq}{\cida}$ has only one query mix). The system is stable if and only if $\displaystyle{\lambda < \sum_{i\colon d_i>0} \mu_i q_i}$.
    \item\label{item:br} Under $\DP{\brq}{\cida}$, the system is stable if and only if $\lambda < 1$.
    \item There exists a policy in $\DP{\iidq}{\cida}$ such that the system is stable if and only if $\lambda < 1$.
    \item There exists a policy in each of $\DP{\indq}{\cida}$ and $\DP{\genq}{\cida}$ such that the system is stable if and only if $\lambda < 1$.
\end{enumerate}
\end{proposition}
\begin{proof} We prove each time separately:
\begin{enumerate} 
\item Consider a querying rule in $\srcq$ where the probability that all queried servers are of class $i$ is given by $\mu_i q_i$.
Then, by Poisson splitting, the class-$i$ servers act like a homogeneous system, independent of all other server classes, with a total arrival rate $\lambda k \mu_i q_i = \lambda k_i \mu_i$. Given that only class-$i$ servers are present in the query, the $\cida$ assignment rule will assign the arriving job to an idle server, if one is present in the query, and a busy server chosen uniformly at random (among the servers in the query) if not. This assignment rule is symmetric among class-$i$ servers, and so the arrival rate to an individual class-$i$ server is $\lambda \mu_i$, which is less than $\mu_i$, ensuring the stability of the system, provided that $\lambda < 1$.

\item $\sfcq$ effectively throws out all server classes except one, which we will call class $i$; by a similar argument as in the proof for $\srcq$, the class-$i$ subsystem will remain stable provided that $\lambda < \mu_i q_i$. Then the largest stability region is achieved by selecting the server class with the largest total capacity.

\item Given that we always query according to some fixed query mix $\dvec\in\Dcal$,  construct an assignment rule in $\cida$ (yielding a dispatching policy in $\DP{\detq}{\cida}$) under which, for all $i\in\Scal$ such that $d_i > 0$, the job is dispatched to a queried class-$i$ server (chosen uniformly at random without considering any idle/busy statuses) with probability $\mu_i q_i \left/ \displaystyle{\sum_{j:d_j > 0} \mu_j q_j}\right.$  (note that this assignment rule is a member of $\cda\subseteq\cida$ as it ignores idle/busy statuses, and therefore does not adhere to our pruning of the space of assignment rules). Then the total arrival rate to class-$i$ servers is \[\lambda k\cdot \frac{\mu_i q_i}{\displaystyle{\sum_{j:d_j>0} \mu_j q_j}} = \mu_i k_i \cdot\frac{\lambda}{\displaystyle{\sum_{j:d_j>0} \mu_j q_j}},\] which is less than $\mu_i k_i$, ensuring stability of the class-$i$ servers, provided that $\lambda < \displaystyle{\sum_{j:d_j>0} \mu_j q_j.}$

\item From~\cite{chen2012asymptotic}, we have that $\DP{\brq}{\jsq}$ is stable if and only if $\lambda < 1$.
For all $(i,\dvec)\in\Scal\times\Dcal$ let $\beta_i(\dvec)$ denote the probability that an arriving job is sent to a class-$i$ server under $\DP{\brq}{\jsq}$, given that the query mix is $\dvec$ (i.e., $\beta_i(\dvec)$ is the probability that the shortest queue is at a class-$i$ server, given query mix $\dvec$).
Now form a policy in the family $\DP{\brq}{\cida}$ by sending the job to a queried class-$i$ server (chosen uniformly at random without considering any idle/busy statuses) with probability $\beta_i(\dvec)$ for all $(i,\dvec)\in\Scal\times\Dcal$ (note that this assignment rule is a member of $\cda\subseteq\cida$ as it ignores idle/busy statuses, and therefore does not adhere to our pruning of the space of assignment rules).
The probability that an arriving job is dispatched to a class-$i$ server is the same under this newly defined policy in $\DP{\brq}{\cida}$ as under $\DP{\brq}{\jsq}$; the only difference is that now all jobs can be viewed as being routed entirely probabilistically.
This will not change the stability region, as $\lb$ remains unchanged for all $i\in\Scal$.

\item This follows from the stability condition for $\DP{\brq}{\cida}$, which is a member of $\DP\iidq\cida$.

\item This follows from item 5 above and the fact that $\DP\iidq\cida\subseteq\DP\indq\cida\subseteq\DP\genq\cida$.
\end{enumerate}
\end{proof}

Note that in proving the existence of a stable dispatching policy in the $\DP{\detq}{\cida}$ and $\DP{\brq}{\cida}$ families (items \ref{item:det} and \ref{item:br} of Proposition~\ref{prop:stability}, respectively), we constructed stable dispatching policies where the assignment rules were members of $\cda\subseteq\cida$, and hence, did not adhere to our pruning rules.  It is not hard to modify these policies to also prove the existence of stable dispatching policies within these families that make use of idle/busy statuses and adhere to our pruning rules.  Consider the simple modification where whenever the original policy would assign the job to a server that is slower than the fastest idle server (or to a busy server of the same speed), instead assign the job to the fastest idle server (if there is more than one fastest idle server, assign the job to one of them chosen uniformly at random).
This modification decreases the arrival rate to busy servers and increases the arrival rate to idle servers, which cannot destabilize the system.

We also present the following result, which amounts to a necessary condition for stability under the $\uniq$ querying rule and any assignment rule:
\begin{proposition}
For any dispatching policy using the $\uniq$ querying rule, the system is unstable if there exists a server class $i\in\Scal$ such that $\lambda > \mu_i/q_i^{d-1}$.
\end{proposition}
\begin{proof}
Under $\uniq$, a query mix consists of only class-$i$ servers---and hence, the arriving job \emph{must} be dispatched to a class-$i$ server under any assignment rule---with probability $q_i^d$. The total arrival rate to the class-$i$ subsystem is then greater than or equal to $\lambda k q_i^d$. The system is unstable if this total arrival rate is greater than the capacity of the class-$i$ subsystem, i.e., if $\lambda k q_i^d > \mu_i k_i$, or, equivalently, if $\lambda > \mu_i / q_i^{d-1}$.
\end{proof}

\subsection{Mean response time analysis}
\label{sec:cida-analysis}

We proceed to present a procedure for determining the mean response time $\ep[T]$ under ${\DP{\qr}{\ar}}$ for any static symmetric querying rule $\qr$ (i.e., any $\qr\in\genq$) and any $\ar\in\cida$ that yield a stable system.

We carry out all analysis in steady-state and rely on mean-field theory.
We let $k\to\infty$, holding $q_i$ fixed for all $i\in\Scal$; consequently, we also have $k_i\to\infty$ for all $i\in\Scal$. We further assume that \emph{asymptotic independence} holds in this limiting regime, meaning that (i)~the states of (i.e., the number of jobs at) all servers are independent, and (ii)~all servers of the same class behave stochastically identically (see Appendix B for simulation evidence in support of this assumption).  With the asymptotic independence assumption in place, we now find the overall mean response time as follows:

\begin{proposition}
\label{prop:ET}
Let $\li$ and $\lb$ denote respectively the arrival rates to idle and busy class-$i$ servers.
Then the overall system mean response time is 
\begin{equation}
\label{eq:ET}
    \ep[T] = \frac{1}{\lambda}\sum_{i=1}^s q_i\left(\frac{(1-\rho_i)\li+\rho_i\lb}{\mu_i-\lb}\right),
\end{equation}
where $\rho_i$ is the fraction of time that a class-$i$ server is busy, given by
\begin{align}\label{eq:rhoi}
\rho_i&\equiv\frac{\li}{\mu_i-\lb + \li}.
\end{align}
\end{proposition}



\begin{proof}
First observe that under our querying and assignment rules, servers of the same class are equally likely to be queried and, within a class, servers with the same idle/busy status are equally likely to be assigned a job.  Hence, by Poisson splitting, it follows that (for any $i\in\Scal$) each class-$i$ server experiences status-dependent Poisson arrivals with rate $\li$\label{lii} when idle and rate $\lb$ when busy.  

Now observe that each class-$i$ server, when busy, operates exactly like a standard M/M/1 system (under the chosen work-conserving scheduling rule) with arrival rate $\lb$ and service rate $\mu_i$.  Since, by virtue of their own presence, jobs experience only busy systems, the mean response time experienced by jobs at a class-$i$ server---which we denote by $\mathbb E[T_i]$---is $1/(\mu_i-\lb)$.\label{T_i} Furthermore, standard M/M/1 busy period analysis gives the expected time of the busy period duration at a class-$i$ server as $\mathbb E[B_i]\equiv1/(\mu_i-\lb)$\label{B_i}; we note that the standard analysis of the M/M/1 queueing system also tells us that while $\mathbb E[B_i]=\mathbb E[T_i]$, $B_i$ and $T_i$ are \emph{not} identically distributed.

Applying the Renewal Reward Theorem immediately yields that $\rho_i$\label{rho_i} (the fraction of time that a class-$i$ server is busy) is as given in Equation~\ref{eq:rhoi} as claimed:
\begin{align}\label{eq:rhoi}
\rho_i&=\frac{\mathbb E[B_i]}{1/\li+\mathbb E[B_i]}=\frac{\li}{\mu_i-\lb + \li}.
\end{align}

Finally, we find the system's overall mean response time by taking a weighted average of the mean response times at each server class.  Let $\lambda_i\equiv(1-\rho_i)\li+\rho_i\lb$\label{li} denote the average arrival rate experienced by a class-$i$ server.  Recalling that $q_i=k_i/k$, it follows that the proportion of jobs that are sent to a class-$i$ server is $k_i\lambda_i/(k\lambda)=q_i\lambda_i/\lambda$, and hence 
\begin{equation}
\label{eq:ET}
    \ep[T] = \sum_{i=1}^s \left(\frac{q_i\lambda_i}{\lambda}\right)\ep[T_i]
    =\frac{1}{\lambda}\sum_{i=1}^s q_i\left(\frac{(1-\rho_i)\li+\rho_i\lb}{\mu_i-\lb}\right),
\end{equation}
which completes the proof.\hfill\qedsymbol

\end{proof}

\begin{remark}
Note that while \emph{mean} response times are insensitive to the choice of (size-blind) scheduling rule, the distribution (and higher moments) of response time do not exhibit this insensitivity.  The same method presented in this section can also allow one to readily obtain the Laplace Transform of response time under many work-conserving scheduling rules.  For example, under First Come First Served ($\mathsf{FCFS}$) scheduling one could use the result that $\widetilde{T_i}(w)=(\mu_i-\lambda_i)/(\mu_i-\lambda_i+w)$ 
for an M/M/1/$\mathsf{FCFS}$ with arrival and service rates $\lambda_i$ and $\mu_i$, respectively, to obtain the overall transform of response time $\widetilde{T}(w)$.  
\end{remark}

In order to use Proposition~\ref{prop:ET} to determine $\mathbb E[T]$ values, we must be able to compute the arrival rates $\li$ and $\lb$ for each $i\in\Scal$.  The following notation will prove useful in expressing these rates: for all  $i\in\Scalb\equiv\{1,2,\ldots,s+1\}$ and $\dvec\in\Dcal$, we let $\bi$ denote the probability that all queried servers that are faster than those in class $i$ are busy (i.e., all queried servers with classes in $\{1,2,\ldots,i-1\}$ are busy). It immediately follows that
\begin{align}\label{eq:b_i(d)}
\bi\equiv\mathbb P(A_1=\cdots=A_{i-1}=0|\Dvec=\dvec)=\prod_{\ell=1}^{i-1}\rho_{\ell}^{d_{\ell}}.
\end{align}
\begin{remark}
Note that for all $\dvec\in\Dcal$, we have $b_1(\dvec)=1$ as it is vacuously true that all queried servers faster than server 1 are busy as no such servers exist. Moreover, $\bsp$ denotes the probability that all queried servers are busy given that $\Dvec=\dvec$.
\end{remark}

In the following theorem, we present a pair of equations (parameterized by $i\in\Scal$) for $\li$ and $\lb$.
\begin{theorem}
\label{thm:lambdaIB}
For all $i\in\Scal$, the arrival rates to idle and busy class-$i$ servers (i.e., $\li$ and $\lb$, respectively), satisfy
\begin{align}
    \li&=\frac{\lambda}{q_i}\sum_{\dvec\in\Dcal}\left\{d_i\bi p(\dvec)\alpha_i(i,\dvec)\sum_{a_i=1}^{d_i}\binom{d_i-1}{a_i-1}\dfrac{(1-\rho_i)^{a_i-1}\rho_i^{d_i-a_i}}{a_i}\right\}\label{eq:li-new}\\
    \lb&=\frac{\lambda}{q_i\rho_i}\sum_{\dvec\in\Dcal}\left\{p(\dvec)\sum_{j=i+1}^{s+1}\bj\left(1-\rho_j^{d_j}\right)\alpha_i(j,\dvec)\right\},\label{eq:lb-new}
\end{align}
where we use the abuse of notation $\rsd\equiv0$.
\end{theorem}

Theorem~\ref{thm:lambdaIB} yields $2s$ equations, which we can solve as a system for the $2s$ unknowns $\{\li\}_{i\in\Scal}$ and $\{\lb\}_{i\in\Scal}$, where we take $\{\rho_i\}_{i=1}^{s}$ and $\{\bi\}_{i=1}^{s+1}$ to be as defined by Equations~\eqref{eq:rhoi} and \eqref{eq:b_i(d)}, respectively.  With the $\li$ and $\lb$ (and consequently, the $\rho_i$) values determined for all $i\in\Scal$, we can then compute $\ep[T]$ directly from Equation~\eqref{eq:ET}, completing our analysis.

The rest of this section is devoted to proving Theorem~\ref{thm:lambdaIB}, by way of two lemmas.  Both lemmas will be concerned with the quantities $\ri$ and $\rb$, defined for all $i\in\Scal$ as follows: for all $\dvec\in\Dcal$ for which $d_i>0$, $\ri$ (respectively, $\rb$)\label{ribd}\label{riid} denotes the probability that the job is assigned to the tagged (class-$i$) server under query mix $\dvec$ given that the tagged server is queried and idle (respectively, busy).  Meanwhile, for all $\dvec\in\Dcal$ for which $d_i=0$, we adopt the convention where $\ri\equiv0$ and $\rb\equiv0$.

\begin{lemma}
\label{lemma:lambdaIB}
The arrival rates $\li$ and $\lb$ are given by:
\begin{align}
    \label{eq:li}
    \li&=\frac{\lambda}{q_i}\sum_{\mathbf \dvec\in\Dcal}d_ip(\dvec)\ri\\ 
    \label{eq:lb}
    \lb&=\frac{\lambda}{q_i}\sum_{\mathbf \dvec\in\Dcal}d_ip(\dvec)\rb.
\end{align}
\end{lemma}
\begin{proof}
Recall that the rate at which the tagged server is queried does not depend on its idle/busy status.  
Given query mix $\dvec$, the probability that the query includes the tagged server is $d_i/k_i$ (by symmetry).  
Because a query is of mix $\dvec$ with probability $p(\dvec)=\pr(\Dvec=\dvec)$, the tagged server is queried at rate \[\lambda k\sum_{\dvec\in\Dcal}\left(\frac{d_i}{k_i}\right)p(\dvec)=\frac{\lambda}{q_i}\sum_{\dvec\in\Dcal}d_i p(\dvec).\]  Of course, the tagged server's presence in the query does not guarantee that the job will be assigned to it.
The arrival rate \emph{from queries with mix $\dvec$} observed by the tagged server when it is idle is \[\lambda k \left(\frac{d_i}{k_i}\right)p(\dvec)\ri=\left(\frac{\lambda}{q_i}\right)d_i p(\dvec)\ri,\] with the analogous expression holding when the tagged server is busy. It follows that the overall arrival rates to an idle and busy class-$i$ server (i.e., $\li$ and $\lb$, respectively) are as claimed.\hfill\qedsymbol
\end{proof}

\begin{lemma}
\label{lemma:ri}
For all $i\in\Scal$ and all $\dvec\in\Dcal$ such that $d_i>0$, the probability that the job is assigned to the tagged class-$i$ server under query mix $\dvec$ given that the tagged server is queried and idle is
\begin{align}
\label{eq:ri}
\ri&=\bi\alpha_i(i,\dvec)\sum_{a_i=1}^{d_i}\binom{d_i-1}{a_i-1}\dfrac{(1-\rho_i)^{a_i-1}\rho_i^{d_i-a_i}}{a_i}.
\end{align}
\end{lemma}

\begin{proof}
 Observe that since we are assuming that the assignment policy $\ar\in\cida$, the job can be assigned to the tagged server only if all faster servers in the query are busy, which occurs with probability $\bi$ (see Equation~\ref{B_i} for details) for a given query mix $\dvec\in\Dcal$.  
If this is the case, then with probability $\alpha_i(i,\dvec)$ the job is assigned to an idle class-$i$ server chosen uniformly at random; hence, the tagged server is selected among the $a_i$ idle class-$i$ servers with probability $1/a_i$.
Enumerating over all possible cases of $A_i=a_i$ when the tagged class-$i$ server is idle, we find the probability that the tagged server is assigned the job when queried with mix $\dvec$:
\begin{align}
\ri&=\bi\alpha_i(i,\dvec)\sum_{a_i=1}^{d_i}\frac{\mathbb P(A_i=a_i|\Dvec=\dvec,\mbox{tagged class-$i$ server is idle})}{a_i}\notag\\
&=\bi\alpha_i(i,\dvec)\sum_{a_i=1}^{d_i}\binom{d_i-1}{a_i-1}\dfrac{(1-\rho_i)^{a_i-1}\rho_i^{d_i-a_i}}{a_i},
\end{align}
where the latter equality follows from the fact that $A_i\ge1$ when the tagged idle server is busy, and so \[(A_i|\Dvec=\dvec,\mbox{tagged class-$i$ server is idle})\sim(A_i|\Dvec=\dvec,A_i\ge1)\sim\operatorname{Binomial}(d_i-1,1-\rho_i)+1,\] which is in turn a consequence of our asymptotic independence assumption.\hfill\qedsymbol
\end{proof}



\begin{lemma}
\label{lemma:rb}
For all $i\in\Scal$ and all $\dvec\in\Dcal$ such that $d_i>0$, the probability that the job is assigned to the tagged class-$i$ server under query mix $\dvec$ given that the tagged server is queried and busy is
\begin{align}\label{eq:rb}
\rb= &\frac{1}{d_i\rho_i}\sum_{j=i+1}^{s+1}\bj\left(1-\rho_j^{d_j}\right)\alpha_i(j,\dvec),
\end{align}
where (as in the statement of Theorem~\ref{thm:lambdaIB}) we use the abuse of notation $\rsd\equiv0$.
\end{lemma}

\begin{proof}
We determine $\rb$ by conditioning on the random variable $J$, denoting the class of the fastest idle queried server (see Section~\ref{sec:cida-formal} for details).  Recall that $J\equiv\min\{j\in\Scal\colon A_j>0\}$, where we take $\min\emptyset\equiv s+1$, so that $J=s+1$ whenever all servers are busy. Letting $\rbj{j}$ denote the probability that the job is assigned to the tagged (class-$i$) server under query mix $\dvec$ given that $J=j$ and the tagged server is queried and busy, the law of total probability yields
\begin{align}\label{eq:rb_total_prob}
\rb&=\sum_{j=1}^{s+1} \rbj{j}\cdot\pr(J=j|\Dvec=\dvec,\mbox{tagged class-$i$ server is busy}).
\end{align}

In order to compute $\rb$, we first observe that, for all $j\in\Scalb$, the job is assigned to some class-$i$ server with probability $\alpha_i(j,\dvec)$ (recall that $\alpha_i(s+1,\dvec)\equiv\alpha_i(\zerovec,\dvec)$ in our abuse of notation), and hence the probability that the job is assigned to \emph{some} class-$i$ server given that $J=j$ is
\begin{align}
\rbj{j}&=\frac{\alpha_i(j,\dvec)}{d_i}\label{eq:rb-conditional}.
\end{align}
It now remains to determine $\pr(J=j|\Dvec=\dvec,\mbox{tagged class-$i$ server is busy})$.  First, we address the case where $J=j$ for some $j\in\Scal$.  Since $\ar\in\cida$,  whenever $j\in\{1,2,\ldots,i\}$, we must have $\alpha_i(j,\dvec)=0$ as the query contains an idle server at least as fast as the tagged (class-$i$) server (which happens to be busy).  Hence, we may restrict attention to $j>i$, in which case---recalling that $\bj$ denotes the probability that all queried servers faster than the tagged (class-$i$ server) are busy, as given by Equation~\eqref{eq:b_i(d)}---we have \begin{align}
\pr(J=j|\Dvec=\dvec,\mbox{tagged class-$i$ server is busy})=\mathbb P(J=j|\Dvec=\dvec,A_i<D_i)&=\frac{\bj\left(1-\rho_j^{d_j}\right)}{\rho_i},\label{eq:pJ=j}
\end{align}
where we recall that $\rsd\equiv0$ and note that the $1/\rho_i$ factor is introduced due to the fact that $A_i<D_i$ (because the server is known to be busy).

We can now combine Equations~\eqref{eq:rb_total_prob}, \eqref{eq:rb-conditional}, and \eqref{eq:pJ=j} together with the fact that $\alpha_i(j,\dvec)=0$ whenever $j\in\{1,2,\ldots,i\}$ in order to obtain the claimed formula for $\rb$.\hfill\qedsymbol
\end{proof}

The proof of Theorem~\ref{thm:lambdaIB} follows from Lemmas~\ref{lemma:lambdaIB},~\ref{lemma:ri},~and~\ref{lemma:rb}, together with the convention that $\ri\equiv0$ and $\rb\equiv0$ whenever $d_i=0$.

\section{Finding optimal dispatching under $\cida$ assignment}
\label{sec:optimization}
Based on the analysis in the previous section, we can now write a nonlinear program for determining optimal dispatching policies in the $\DP{\qrf}{\arf}$ family for various choices of $\qrf$.  This amounts to jointly determining an optimal probability distribution $p$ over query mixes and an optimal family of functions constituting the assignment rule $\alpha_i$ (for $i\in\Scal$).

Each choice of querying rule family $\qrf$ yields a different optimization problem. All of these optimization problems can be formulated to share a common objective function.  Meanwhile, the set of permissible querying rules (i.e., the chosen querying rule family $\qrf$) restricts the set of feasible decision variables.  Naturally, formulating problems in this way, if $\qrf'\subseteq\qrf$,  then the feasibility region of the optimization problem associated with $\DP{\qrf'}{\cida}$ is contained within that associated with $\DP{\qrf}{\cida}$, and hence, all such optimization problems have feasibility regions contained within that of $\DP{\genq}{\cida}$.  Consequently, if we can solve the problem associated with $\DP{\genq}{\cida}$, then solving a  problem associated with $\DP{\qrf}{\cida}$ for another querying rule family $\qrf$ will never yield a policy that results in a strictly lower mean response time than the one we have already found.  In fact, the problem associated with $\DP{\genq}{\cida}$ can be viewed as a ``relaxation'' of the others.

While the above discussion seems to suggest that one need only study the optimization problem associated with $\DP{\genq}{\cida}$, there are several reasons for studying problems associated with $\DP{\qrf}{\cida}$ for other querying rule families, $\qrf\subseteq\genq$.  First, as discussed in Appendix D, many of the feasibility regions associated with the other optimization problems can be expressed as polytopes in a space with far fewer dimensions than those studied under $\genq$, suggesting that these other problems might be solved more efficiently.  Numerical evidence that we will present in Section~\ref{sec:graphs} corroborates this suggestion.  Second, as we shall discuss in detail throughout Section~\ref{sec:numerical}, these problems are often prohibitively difficult to solve, so we rely on heuristics to find strong performing (although not necessarily optimal) solutions within each family of dispatching policies.  Therefore, it will sometimes be the case that even though $\qrf'\subseteq\qrf$, a heuristic (rather than truly optimal) ``solution'' to a problem associated with $\DP{\qrf'}{\cida}$ may outperform those obtained from $\DP{\qrf}{\cida}$.  Finally, some families of rules with simpler structures may be more desirable for practical implementation purposes.

Before presenting our optimization problems, we note that we have not consistently formulated each problem as a restriction on the problem associated with $\DP{\genq}{\cida}$.  While for any $\DP{\qrf}{\cida}$ there exists at least one formulation of the optimization problem that resembles that of $\DP{\genq}{\cida}$ with additional constraints, we have opted for a more ``natural'' approach where we tailor the optimization problem for each dispatching policy $\DP{\qrf}{\cida}$ to the structure of the choice of querying rule family $\qrf$.

\begin{remark}
The optimization problems that we study are of the form where we minimize $f\colon\mathcal X\to\mathbb R$ on the feasible set $\mathcal X$ such that for each $x\in\mathcal X$, $x$ corresponds to a dispatching policy that yields an overall mean response time $\mathbb E[T]=f(x)$.  We say that two optimization problems with feasible regions $\mathcal X_1$ and $\mathcal X_2$, respectively are \emph{equivalent formulations} of one another if both (i) for each $x_1\in\mathcal X_1$, there exists an $x_2\in\mathcal X_2$ such that the policies corresponding to $x_1$ in the first problem and $x_2$ in the second yield stochastically identical systems, and (ii) the analogous statement holds for each $x_2\in\mathcal X_2$.  While all formulations of a given problem have solutions that yield identical system behavior, some formulations may be more tractable (or more amenable to heuristic analysis) than others.
\end{remark}


\subsection{Finding optimal $\DP{\genq}{\cida}$ dispatching policies}

We begin by considering the case where $\qrf=\genq$, i.e., the case where we allow for all possible (static symmetric) querying rules, where all functions $p\colon\mathcal D\to[0,1]$ are valid so long as $\sum_{\dvec\in\Dcal} p(\dvec)=1$. 

Since both $p$ and all of the $\alpha_i$ functions take arguments from a domain with finitely many elements, we would like to treat each \emph{evaluation} of these functions as a decision variable, i.e., we would like to treat $p(\dvec)$ for each $\dvec\in\Dcal$ and $\alpha_i(\avec,\dvec)$ for each triple $(i,\avec,\dvec)\in\Scal\times\Acal\times\Dcal$ (or $\alpha_i(j,\dvec)$ for each triple $(i,j,\dvec)\in\Scal\times\Scalb\times\Dcal$ when using our abuse of notation) as decision variables, with appropriate constraints.  However, as we have discussed earlier, we have pruned the decision space so that $\alpha_i(\avec,\dvec)$ depends only on the class of the fastest idle queried server $J\equiv\min\{j\in\Scal\colon A_j>0\}$ realized under the event $(\Avec,\Dvec)=(\avec,\dvec)$ and on the (realized value of the) set of classes of queried servers that are faster than class-$J$ servers $\{j< J\colon D_j>0\}$ under the same event.  For example consider a setting where $s=4$ and $d=6$, where $\dvec_1=(4,0,1,1)$, $\avec_1=(0,0,1,1)$, $\dvec_2=(2,0,3,1)$, and $\avec_2=(0,0,3,0)$.  Under both the events $(\Avec,\Dvec)=(\avec_1,\dvec_1)$ and $(\Avec,\Dvec)=(\avec_2,\dvec_2)$, we have $J=3$, while $\{j< J\colon D_j>0\}=\{1\}$, and so we must have $\alpha_i(\avec_1,\dvec_1)=\alpha_i(\avec_2,\dvec_2)$---and equivalently, using our abuse of notation, we must have $\alpha_i(3,\dvec_1)=\alpha_i(3,\dvec_2)$---for all $i\in\{1,2,3,4\}$.  

The pruning described above could be enforced in our optimization problem through the introduction of constraints, but we may also approach pruning more directly by reducing the set of decision variables.  We opt for the latter, to which end we introduce the map  $\gamma\colon\Scalb\times\Dcal\to\Dcal$.  In order to define $\gamma$, let $I\{\cdot\}$ denote the indicator function, let $\evec_i$ denote the $i$-th $s$-dimensional unit vector (so that, e.g., when $s=4$, we have $\evec_3\equiv(0,0,1,0)$) and let $h(\dvec)\equiv\min\{\ell\in\Scal\colon d_\ell>0\}$\label{h(d)} denote the class of the fastest queried server (regardless of whether this server is idle or busy).  The map $\gamma$ is defined as follows: \[\gamma(j,\dvec)\mapsto\sum_{i=1}^sI\{d_i>0\mbox{ and } i\le j\}\evec_i+\left(d-\sum_{i=1}^s I\{d_i>0\mbox{ and }i\le j\}\right)\evec_{h(\dvec)},\]\label{gamma}  (so that, e.g., when $s=d=8$, we have $\gamma(5,(0,2,1,0,3,0,2,0))=(0,6,1,0,1,0,0,0))$). Given some $j\in\Scalb$ and $\dvec\in\Dcal$, $\gamma(j,\dvec)$ is the unique query mix with the maximum possible number of queried class-$h(\dvec)$ servers such that the realized value of the set $\{j< J\colon D_j>0\}$ is the same under events $(J,\Dvec)=(j,\dvec)$ and $(J,\Dvec)=(j,\gamma(j,\dvec))$---thus guaranteeing that $\alpha_i(j,\dvec)$ and $\alpha_i(j,\gamma(j,\dvec))$ are identical due to the pruning.  We note that the fact that the number of class-$h(\dvec)$ queried servers is maximized is not of any particular significance; rather, the map $\gamma$ allows us to specify a unique query mix to act as a ``representative'' for all query mixes that would be treated in the same way by the assignment rule under a given realization of the random variable $J$.  Returning to our optimization problem, observe that we can reduce the dimensionality of the feasible region by assigning values only to those $\alpha_i(j,\dvec)$ when $(i,j,\dvec)\in\Tcal$ where the set $\Tcal$ represents a pruned set of triples $(i,j,\dvec)$, for which each $\alpha_i(j, \dvec)$ can be assigned a distinct nonzero value in formulating an assignment rule:
\begin{align}\label{eq:tcal}
\Tcal\equiv\left\{(i,j,\dvec)\in\Scal\times\Scalb\times\Dcal\colon i\le j,\,d_i>0,\,(j\le s)\implies d_j>0,\,\gamma(j,\dvec)=\dvec\right\}.
\end{align}
Meanwhile, wherever the optimization problem would make reference to $\alpha_i(j,\dvec)$, we instead write the decision variable $\alpha_i(j,\gamma(j,\dvec))$ as both values are the same.  Furthermore, as defined in Equation~\eqref{eq:tcal}, $\Tcal$ excludes triples $(i,j,\dvec)\in\Scal\times\Scalb\times\Dcal$ where (i) $j<i$, (ii) $d_i=0$,  or (iii) $d_j=0$ and $j\neq s+1$.  Defining $\Tcal$ in such a way allows us to omit $\alpha_i(j,\dvec)$ for such triples, as all of these values must be $0$ (see Section~\ref{sec:lbia} for details).  In order to write the $\sum_{i=1}^s\alpha_i(j,\dvec)=1$ constraints concisely, without reference to $\alpha_i(j,\dvec)$ for triples $(i,j,\dvec)\not\in\Tcal$, we need a way to specify those $(j,\dvec)$ pairs that can form a triple $(i,j,\dvec)\in\Tcal$ with one or more classes $i\in\Scal$, so we also introduce the notation $\Pcal$ to denote such pairs:
\begin{align}\label{eq:pcal}
\Pcal\equiv\left\{(j,\dvec)\in\Scalb\times\Dcal\colon(\exists i\in\Scal\colon (i,j,\dvec)\in\Tcal)\right\}.
\end{align}
Similarly, in expressing the inner sum in Equation~\eqref{eq:lb-new}, we avoid making reference to the same forbidden triples and ensure that $j\ge i+1$, by introducing the following notation for any fixed $(i,\dvec)\in\Scal\times\Dcal$: 
\begin{align}\label{eq:Jid}
\Jcal_i(\dvec)&\equiv\{j\in\{i+1,i+2,\ldots,s+1\}\colon (i,j,\gamma(j,\dvec))\in\Tcal\}.
\end{align}

Finally, building upon our analysis in Section~\ref{sec:analysis} (including requiring $\lb <\mu_i$ for all $i\in\Scal$ in order to guarantee stability), we have the following optimization problem:

\begin{tcolorbox}
{\small \textbf{Optimization Problem for the $\DP{\genq}{\cida}$ family}

Given values of $s$, $d$, $\lambda$, and $\mu_i$ and $q_i$ (both given for all $i\in\Scal$),  determine \textbf{nonnegative} values of the decision variables $\li$ and $\lb$ (both for all $i\in\Scal$), $p(\dvec)$ (for $\dvec\in\Dcal$), and $\alpha_i(j,\dvec)$ (for all $(i,j,\dvec)\in\Tcal$) that solve the following nonlinear program:
\begin{align*}
    \min\quad&\frac{1}{\lambda}\sum_{i=1}^s q_i\left(\frac{(1-\rho_i)\li+\rho_i\lb}{\mu_i-\lb}\right)\\
    \mathrm{s.t.}\quad&  \li=\frac{\lambda}{q_i}\sum_{\dvec\in\Dcal}\left\{d_i\bi p(\dvec)\alpha_i(i,\gamma(i,\dvec))\sum_{a_i=1}^{d_i}\binom{d_i-1}{a_i-1}\dfrac{(1-\rho_i)^{a_i-1}\rho_i^{d_i-a_i}}{a_i}\right\}&(\forall i\in\Scal)\\
    &\lb=\frac{\lambda}{q_i\rho_i}\sum_{\dvec\in\Dcal}\left\{p(\dvec)\sum_{\mathclap{j\in\Jcal_{i}(\dvec)}}\bj\left(1-\rho_j^{d_j}\right)\alpha_i(j,\gamma(i,\dvec))\right\}&(\forall i\in\Scal)\\
    &\lb <\mu_i&(\forall i\in\Scal)\\
    &\sum_{\dvec\in\Dcal} p(\dvec)=1\\
    &\sum_{\mathclap{\substack{i\in\Scal:\\ (i,j,\dvec)\in\Tcal}}}\alpha_i(j,\dvec)=1&(\forall (j,\dvec)\in\Pcal)
\end{align*}
where for all $i\in\Scal$ in writing $\rho_i$ we are denoting the expression $\li/\left(\mu_i-\lb+\li\right)$ with $\rsd\equiv1$ for all $\dvec\in\Dvec$, and for all $j\in\Scalb$ in writing $\bj$ we are denoting the expression $\prod_{\ell=1}^{j-1}\left(\lil/\left(\mu_\ell-\lbl+\lil\right)\right)^{d_\ell}$.

The $p(\dvec)$ values (for all $\dvec\in\Dcal$) and the $\alpha_i(j,\dvec)$ values (for all $(i,j,\dvec)\in\Tcal$ together with $\alpha_i(j,\dvec)=0$ for all $(i,j,\dvec)\not\in\Tcal$) from an optimal solution specify the querying and assignment rules associated with an optimal $\DP{\genq}{\cida}$ dispatching policy, respectively.}
\end{tcolorbox}

\subsection{Finding optimal $\DP{\iidq}{\cida}$ dispatching policies}\label{sec:opt-iid}

We now turn our attention to the case where $\qrf=\iidq$ as it is simpler to address $\iidq$ before the more general (but less general than $\genq$) $\indq$ family.  Under the $\iidq$ querying rule, the $d$ servers are queried independently according to an identical probability distribution over the set of server classes $\Scal$.  For any querying rule $\qr\in\iidq$ we can express the probability distribution $p$ over query mixes $\Dvec$ in terms of an auxiliary distribution $\indp$ over the \emph{set of classes} $\Scal$.  Specifically, we express $\indp$ as a function $\indp\colon\Scal\to[0,1]$ that is subject to the constraint $\sum_{i=1}^s\indp(i)=1$, where $\indp(i)$ is the probability that an arbitrary queried server is of class $i$.  In particular, due to the structure of querying rules in $\iidq$, we query $d$ servers independently according to $\iidq$ according to $\indp$ upon each arrival, yielding \[p(\dvec)=\binom{d}{d_1,d_2,\ldots,d_s}\prod_{i=1}^s\indp(i)^{d_i}=d!\prod_{i=1}^s\frac{\indp(i)^{d_i}}{d_i!},\]  that is, any $\qr\in\iidq$ makes independent queries with querying mixes $\Dvec$ that are \emph{multinomially distributed} random vectors.  Moreover, such a querying rule $\qr$ is uniquely identified by $\indp$.  The above observations allow us to express the optimization problem for $\DP{\iidq}{\cida}$ as a modification of the optimization problem for $\DP{\genq}{\cida}$ (see Appendix C).

\subsection{Finding optimal $\DP{\indq}{\cida}$ dispatching policies}\label{sec:ind}

When $\qr\in\iidq$, the function $\indp\colon\Scal\to[0,1]$ governed the probability distribution by which servers were queried: specifically, $\indp(i)$ denoted the probability that any individual queried server is of class $i$.  We extend this notion to the case where we can query $d$ servers according to potentially different distributions (i.e., when $\qr\in\indq$) as follows: let $\indp_1,\indp_2,\ldots,\indp_d\colon\Scal\to[0,1]$ denote a family of functions such that $\indp_\ell(i)$ denotes the probability that upon any job's arrival, the $\ell$-th server queried is of class $i$. 

\begin{remark}
All servers are queried simultaneously (under all querying rules, including those contained $\indq$ in particular), and so the order of the queries is irrelevant (i.e., a querying rule specified by $\indp_1,\indp_2,\ldots,\indp_d$ performs indistinguishably from one specified by $\indp_1'=\indp_{\sigma(1)},\indp_2'=\indp_{\sigma(2)},\ldots,\indp_d'=\indp_{\sigma(d)}$, for any permutation $\sigma\colon\{1,2,\ldots,d\}\to\{1,2,\ldots,d\}$).
\end{remark}

We would like to define $p(\dvec)$ in terms of $\indp_1,\indp_2,\ldots,\indp_d$.  With this end in mind, we introduce some additional notation: let \label{Qcal}
\begin{align}
    \Qcal&\equiv\{1,2,\ldots,d\}\qquad\mbox{and}\qquad\vec{\Qcal}\equiv(\Qcal_1,\Qcal_2,\ldots,\Qcal_s),\label{eq:qcal}
\end{align}
    where each $\Qcal_\ell$ is a subset of $\Qcal$ (i.e., $\vec{\Qcal}$ is an $s$-tuple of subsets of $\Qcal$), and let  $\Bcal(\dvec)$ denote the set of all $s$-tuples $\vec{\Qcal}$ that form a partition of $\Qcal$ such that the $\ell$-th entry of $\vec{\Qcal}$  contains exactly $\ell_i$ elements.  That is,
\[\Bcal(\dvec)=\left\{\vec{\Qcal}\colon (\forall \ell\in\mathcal S\colon \Qcal_\ell\subseteq\mathcal \Qcal,\,|\Qcal_\ell|=d_\ell),\bigcup_{\ell=1}^s\Qcal_\ell=\Qcal\right\}.\] Crucially, $\Bcal(\dvec)$ corresponds to the ways that $d$ queries can result in the query mix $\dvec$. With the above notation defined, we can now write the following:
\[p(\dvec)=\sum_{\vec{\Qcal}\in\Bcal(\dvec)}\prod_{i=1}^s\prod_{u\in\Qcal_i}\indp_u(i).\] The optimization problem for the $\DP{\indq}{\cida}$ family of dispatching policies---presented in Appendix C---then follows readily from that of $\DP{\genq}{\cida}$ family.

\subsection{Finding optimal $\DP{\detq}{\cida}$ dispatching policies}

We now address the case where $\qrf=\detq$.  For any $\qr\in\detq$, $p$ is such that there exists some specifically designated $\dvec\in\Dcal$ where $p(\dvec)=1$ and $p(\dvec')=0$ for all $\dvec'\in\Dcal\backslash\{\dvec\}$.  That is, we always query deterministically, so that $\Dvec=\dvec$ upon every job arrival.  When attempting to find optimal $\DP{\detq}{\cida}$ dispatching policies, we need to evaluate the mean response time under each $\DP{\mathsf{DQR}_{\dvec}}{\cida}$ dispatching policy, where $\mathsf{DQR}_{\dvec}$\label{dqrsub} is an \emph{individual querying rule} in $\detq$ that always queries a set of servers with query mix $\dvec$ (further note that all querying rules in $\detq$ are of this form; hence, we have  $\detq=\{\mathsf{DQR}_{\dvec}\colon \dvec\in\Dcal\}$).  Then we choose the value of $\dvec$ (and the corresponding policy $\ar\in\cida$) that yields the best mean response time.  Hence, the optimization problem for the $\DP{\detq}{\cida}$ family of dispatching policies consists of solving $|\Dcal|$ optimization \emph{subproblems} and then comparing the objective values of these subproblems.  While our approach can be seen as a disjunctive nonlinear program composed of a single objective function on the union of several nonlinear feasibility regions, one could also approach this optimization problem as a single mixed integer nonlinear program (MINLP).

Since we need only consider $\Dvec=\dvec$ in each subproblem, we can dispense with the need for the map $\gamma$ in this setting, however it will be useful to introduce the following analogues of $\Scal$, $\Tcal$, $\Pcal$, and $\Jcal_i(\dvec)$:
\begin{align}
    \Scal(\dvec)&\equiv\{i\in\Scal\colon d_i>0\}\label{eq:scal-d}\\
    \Tcal(\dvec)&\equiv\left\{(i,j)\in\Scal\times\Scalb\colon i\le j,\,d_i>0,\,(j\le s)\implies d_j>0\right\}\label{eq:tcal-d}\\
    \Pcal(\dvec)&\equiv\left\{j\in\Scalb\colon(\exists i\in\Scal\colon(i,j)\in\Tcal(\dvec))\right\}=\{j\in\Scal\colon d_j>0\}\cup\{s+1\}\label{eq:pcal-d}\\
    \Jcal_i(\dvec)&\equiv\left\{j\in\{i+1,i+2,\ldots,s+1\}\colon (i,j)\in\Tcal(\dvec)\right\}.\label{eq:Jcal-d}
\end{align}
While Equation~\eqref{eq:Jcal-d} may appear to be a \emph{redefinition} of $\Jcal_i(\dvec)$ it is actually consistent with the earlier definition provided in Equation~\eqref{eq:Jid}.

With the above notation defined we can formulate the optimization problem for the $\DP{\detq}{\cida}$ family of dispatching policies, which we present in Appendix C.

\subsection{Finding optimal $\DP{\srcq}{\cida}$ dispatching policies}\label{sec:opt-src}

We first observe that when a dispatching policy's querying rule $\qr\in\srcq$, then assignment decisions under that policy are always among servers of the same class, so if we further impose that the dispatching policy's assignment rule $\ar\in\cida$, then the assignment decision amounts to sending the job to an idle queried server (chosen uniformly at random) whenever the query includes such a server and to any (busy) server (chosen uniformly at random) otherwise. Just as the querying rules in the $\iidq$ family were uniquely specified by some probability distribution over the serer classes, $\indp$, querying rules in the $\srcq$ family are also specified by such a probability distribution, which we denote by $\srcp$. Specifically, let $\srcp\colon\Dcal\to[0,1]$ be a distribution that satisfies $\mathbb P(\Dvec=d \evec_i)=\srcp(i)$ for all $i\in\mathcal S$ and $\mathbb P(\Dvec=\dvec)=0$ for all $\dvec\in\mathcal D\backslash\{d\evec_1,d\evec_2,\ldots,d\evec_s\}$, where $d\evec_i$ corresponds to the query mix where $d$ class-$i$ servers have been queried (i.e., it denotes the vector of length $s$ with all zero entries, except for an entry of $d$ at position $i$). In particular, the $\alpha_i(j,\dvec)$ values are immaterial, and we do not need to optimize over them.  This yields the associated optimization problem provided in Appendix C.

\subsection{Finding optimal $\DP{\sfcq}{\cida}$ dispatching policies}

As we have remarked earlier there are exactly $i$ policies $\sfcq$ contains exactly $s$ querying rules.  Specifically, we always query $d$ class-$i$ servers for some fixed $i\in\Scal$.  As the querying rule is specified by the choice of this fixed value of $i$ alone, we can disregard querying probabilities.  Moreover, as all queried servers are of the same class, we can also disregard assignment probabilities.  Hence, optimization amounts to choosing the fixed value of $i\in\Scal$ that minimizes the mean response time.  To this end, as in the case where $\qrf=\detq$, we make use of subproblems (and alternatively could have made uses of integer variables), although this time we only have $s$ such subproblems.  The associated optimization problem is presented in Appendix C.  

We note that the solution to each subproblem does not depend on the objective function (although objective function values must be computed to find $i^{*}$, the index of the subproblem with the lowest objective function value).  Moreover, each subproblem will have at most one feasible solution.  Essentially, solving each subproblem merely requires one to solve a system of two nonlinear equations in two constrained unknowns: $\li\in[0,\infty)$ and $\lb\in[0,\mu_i)$.

\subsection{Finding optimal $\DP{\qr}{\cida}$ dispatching policies}

We also formulate an optimization rule in order to determine optimal assignment rules $\ar\in\cida$ given any individual querying rule $\qr\in\qrf$ as specified by some $p\colon\mathcal D\to[0,1]$, as long as that querying rule yields a stable system under some assignment rule.  Since the probability distribution over $\Dcal$ is specified, we need only determine the assignment probabilities.  This optimization problem (i.e., the one associated with $\DP{\qr}{\cida}$) is presented in Appendix C.

We are particularly interested in the $\uniq$ and $\brq$ rules, defined by  \[p(\dvec)=\binom{d}{d_1,d_2,\ldots,d_s}\left(\frac1{s^d}\right)=d!\left/\left(s^d\prod_{i=1}^{s}d_i!\right)\right.\] and \[p(\dvec)=\binom{d}{d_1,d_2,\ldots,d_s}\prod_{i=1}^s(\mu_i q_i)^{d_i}=d!\prod_{i=1}^s\frac{\left(\mu_iq_i\right)^{d_i}}{d_i!},\] respectively, in accordance with the fact that both are members of the $\iidq$ family (see Section~\ref{sec:opt-iid}).  Note that some other specific querying rules---especially those that never query servers of one or more classes---allow for significant pruning of the space of assignment rule.


\section{Numerical results for $\cida$ assignment}
\label{sec:numerical}
In this section we compare the performance of dispatching policies found by numerically solving our optimization problems associated with $\DP{\qrf}{\cida}$ for various querying rule families $\qrf$ (including the single-rule family $\qrf=\{\brq\}$).  In these comparisons we will examine both performance (i.e., mean response time), and the computation time associated with determining the optimal parameters for the querying and/or assignment rules.

\subsection{Parameter settings}
\label{sec:parameters}

We provide numerical results for a variety of parameter settings (i.e., problem instances), where each parameter setting consists of a choice of $s$, $d$, $\lambda$, $\mu_1,\ldots,\mu_s$, and $q_1,\ldots,q_s$.  Our choices of $\lambda$, $s$, and $d$ resemble a full factorial design, while our choices of $\mu_1,\ldots,\mu_s$ and $q_1,\ldots,q_s$ clearly depend on $s$ and are subject to the normalization constraint $\sum_{i=1}^s\mu_iq_i=1$.  \label{r_i}

Specifically, we examine all combinations of $s$, $d$, and $\lambda$, where $s,d\in\{2,3,4\}$ and $\lambda\in\{0.05,0.10,\ldots,0.95\}$ (although when plotting curves, we instead consider $\lambda\in\{0.02,0.04,\ldots,0.98\}$).  For each $(s,d,\lambda)$ setting, we then consider one set of $\mu_1,\ldots,\mu_s$ corresponding to each subset of $s-1$ elements of $\{1.25,1.50,2,3,5\}$: for each $\{R_1,R_2,\ldots,R_{s-1}\}\subseteq\{1.25,1.50,2,3,5\}$, ordered so that $R_1>R_2>\ldots>R_{s-1}$, we let $\mu_i=R_i\mu_s$ for each $i\in\Scal\backslash\{s\}$ (i.e., $R_i\equiv\mu_i/\mu_s$).\label{Ri}.  That is, in each parameter setting each server that does not belong to the slowest class runs at a speed that is $25\%$, $50\%$, $100\%$, $200\%$, or $400\%$ faster than the speed of the slowest server, with each parameter setting accommodating $s-1$ such speedup factors.  The speed of the slowest server depends on the values of $q_1,q_2,\ldots,q_s$ (see below), as follows: \[\mu_s=\left(q_s+\sum_{i=1}^{s-1}R_iq_i\right)^{-1}.\]  Meanwhile, for each $(s,d,\lambda,R_1,R_2,\ldots,R_{s-1})$ setting, we consider the following $(q_1,q_2,\ldots,q_{s})$ combinations:
\[\left\{(q_1,q_2,\ldots,q_{s})\in\mathbb Q^
s\colon (\forall i\in\Scal\colon 6q_i\in\mathbb Z,\,q_i>0),\, \sum_{i=1}^s q_i=1 \right\}.\] That is, we consider all (and only those) combinations $(q_1,q_2,\ldots,q_{s})$ where we can view each server class as holding a (nonzero integer) number of ``shares''---out of a total of 6 such shares---with each class being allocated a number of servers proportional to the number of shares it holds.
This methodology for selecting $\mu_i$ and $q_i$ values was chosen to allow for a wide variety of parameter settings while ensuring that in each setting no class is particularly under- or over-represented nor so much slower or faster than others.  In this way, we avoid extreme parameter settings that render certain classes (and hence, certain aspects of querying and assignment rules) inconsequential.

Note that there are 3 choices for $s$, 3 choices for $d$, 19 choices for $\lambda$, $\binom{5}{s}$ choices of speed configurations for each choice of $s$, and also $\binom{5}{s}$ ``share'' configurations for each choice of $s$ (that is, $5$ choices of each configuration when $s=4$ and $10$ choices of each configuration when $s=2$ or $s=3$).  Hence, we consider a total of $(3)(19)\left(5^2+10^2+10^2\right)=12\,825$ parameter settings. These parameter settings can be broken down into $1\,875$ settings for each of the $19$ $\lambda$ values.  Alternatively, they can broken down by the choice of $s$: $1\,425$ settings where $s=2$, and $5\,700$ settings each when $s=3$ and $s=4$.

\subsection{Numerical optimization methodology and notation}
\label{sec:opt-methods}

All of the numerical results we present throughout this section were obtained using code written in the programming language Julia.  We used the JuMP package \cite{dunning2017jump} in Julia to define our optimization models (see Section~\ref{sec:optimization}), and we solved these problems using the Interior Point Optimizer (IPOPT) optimization package \cite{Lubin2015,Wachter2006}.  Note that due to the presence of nonconvexity in our optimization problems,  IPOPT does not consistently yield globally optimal solutions.  Hence, for each policy family, we should view the associated ``optimal'' solution yielded by IPOPT as being the parameters of a heuristically chosen policy belonging to that family.  For further implementation details and small caveats to the results presented in this section, see Appendix E.  

Now consider an arbitrary querying rule family $\qrf$ and an arbitrary assignment rule family $\arf$.  Let the (admittedly cumbersome) notation \label{ipoptd}\label{ipoptq}\label{ipopta} \[\ipoptd{\DP{\qrf}{\arf}}\equiv\DP{\ipoptq{\DP{\qrf}{\arf}}}{\ipopta{\DP{\qrf}{\arf}}}\] denote the dispatching policy specified by the IPOPT solution to the optimization problem associated with the $\DP\qrf\arf$ family of dispatching policies (assuming such an optimization problem exists, has been identified, and can be implemented and given to IPOPT).  That is, for a querying rule family $\qrf$ (e.g., $\genq$), we use IPOPT to ``solve'' an optimization problem that involves jointly selecting querying and assignment probabilities, resulting in a querying rule (belonging to $\qrf$), which we denote by $\ipoptq{\DP{\qrf}{\arf}}$, and an assignment rule (belong to $\arf$), which we denote by $\ipopta{\DP{\qrf}{\arf}}$. Recall that in all of the optimization problems that we have proposed in Section~\ref{sec:optimization}, we have always considered $\arf=\cida$, so to alleviate the burden imposed by this cumbersome notation, we can omit the reference to the assignment rule family whenever we take it to be $\cida$.  That is, we take $\cida$ as the ``default'' assignment rule family and use the notation $\ipoptd{\qrf}\equiv\DP{\ipoptq{\qrf}}{\ipopta{\qrf}}$, where we let $\ipoptq{\qrf}\equiv\ipoptq{\DP{\qrf}{\cida}}\in\qrf$ and $\ipopta{\qrf}\equiv\ipopta{\DP{\qrf}{\cida}}\in\cida$, from which it follows that $\ipoptd{\qrf}=\ipoptd{\DP{\qrf}{\cida}}$.

\begin{remark}
We abuse this notation by adapting it for use with specific policies, rather than only families, so that, e.g.,  $\ipoptd{\brq}\equiv\ipoptd{\{\brq\}}$, and  $\ipoptd{\DP{\srcq}{\jsq}}\equiv\ipoptd{\DP{\srcq}{\{\jsq\}}}$.
\end{remark}

\subsection{Comparison of querying rule families with respect to $\mathbb E[T]$ and optimization runtime}
\label{sec:graphs}

We proceed to evaluate the performance of the $\ipoptd\genq$, $\ipoptd\indq$, $\ipoptd\iidq$, $\ipoptd\srcq$, and $\ipoptd\brq$ dispatching policies.  We omit examination of the $\detq$ and $\sfcq$ querying rule families as well as the $\uniq$ querying rule, as under many of our parameter settings, any dispatching policy constructed from such querying rules yields an unstable system (see Section~\ref{sec:stability} on stability and Section~\ref{sec:parameters} on our parameter settings). We examine the performance of $\ipoptd\detq$ across a small set of parameters (taken from our earlier work in \cite{gardner2020scalable}) at the end of this section.

We evaluate the $\mathbb E[T]$ values yielded by each of the dispatching policies under consideration, for each of the $12\,825$ parameter settings described in Section~\ref{sec:parameters}. For each policy, we then compute the mean and median value of both $\mathbb E[T]$ and the optimization runtime (measured in seconds) across all of our parameter settings.
Figure~\ref{fig:mean} illustrates the tradeoff between $\mathbb E[T]$ and optimization runtime as aggregated across our parameter settings.
In Figure~\ref{fig:mean}~(left) we plot the (mean $\mathbb E[T]$, mean runtime) pairs associated with each policy, while in Figure~\ref{fig:mean}~(right) we plot the analogous pairs for median values.
Before describing Figure~\ref{fig:mean} in detail, we introduce one additional policy, motivated by the surprising observation that both $\ipoptd\indq$ and $\ipoptd\iidq$ outperform $\ipoptd\genq$ with respect to the mean value of $\mathbb E[T]$ across the parameter set, with $\ipoptd\indq$ outperforming $\ipoptd\genq$ with respect to the analogous median value as well. We observe this despite the fact that $\iidq\subseteq\indq\subseteq\genq$, which means that the best $\genq$-driven dispatching policy \emph{must} perform at least as well as the best $\indq$- and $\iidq$-driven policies; unfortunately, as IPOPT does not consistently find true optimal solutions, the solution found by IPOPT for a particular family can occasionally outperform the solution it finds for a more general family. We can construct a new policy to remedy the somewhat lackluster performance of $\ipoptd\genq$ by exploiting the fact that we can seed IPOPT with a feasible solution before running it to solve an optimization problem.  Thus far, we have only discussed results which were obtained by running IPOPT without seeding it with an initial solution, however, IPOPT frequently yields noticeably better solutions to the optimization problem associated with the $\DP\genq\cida$ family when seeded with the IPOPT solution associated with the $\DP\indq\cida$ (as compared to the solution yielded by the ``unseeded'' problem associated with the $\DP\genq\cida$).  We use the notation $\genseed\equiv\DP\genseedq\genseeda$ \label{genseedd} \label{genseedq} \label{genseeda} to refer to this new heuristic dispatching policy (see Appendix E for details).  One can similarly construct other heuristics for choosing initial values (e.g., seeding the $\iidq$ optimization problem with the solution IPOPT found for the $\brq$ optimization problem, and even using the solution to the aforementioned seeded problem as a seed for the $\indq$ optimization problem, etc.); we extensively explored different heuristics for choosing initial values (e.g., seeding the $\iidq$ optimization problem with the solution IPOPT found for the $\brq$ optimization problem); we found that---unlike $\genseeda$---other alternative heuristics yielded negligible benefits in comparison to their ``unseeded' counterparts.

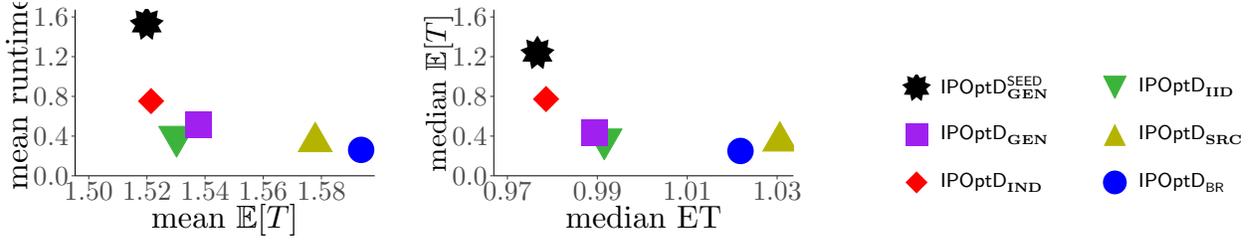
\begin{figure}
    \centering
    \begin{tabular}{ccc}
    \scalebox{.5}{
\begin{tikzpicture}[x=1pt,y=1pt]
\definecolor{fillColor}{RGB}{255,255,255}
\path[use as bounding box,fill=fillColor,fill opacity=0.00] (0,0) rectangle (289.08,180.67);
\begin{scope}
\path[clip] (  0.00,  0.00) rectangle (289.08,180.67);
\definecolor{drawColor}{RGB}{255,255,255}
\definecolor{fillColor}{RGB}{255,255,255}

\path[draw=drawColor,line width= 0.6pt,line join=round,line cap=round,fill=fillColor] (  0.00,  0.00) rectangle (289.08,180.68);
\end{scope}
\begin{scope}
\path[clip] ( 56.90, 49.41) rectangle (283.58,175.17);
\definecolor{drawColor}{RGB}{255,255,255}

\path[draw=drawColor,line width= 0.6pt,line join=round] ( 56.90, 49.41) --
	(283.58, 49.41);

\path[draw=drawColor,line width= 0.6pt,line join=round] ( 56.90, 79.47) --
	(283.58, 79.47);

\path[draw=drawColor,line width= 0.6pt,line join=round] ( 56.90,109.53) --
	(283.58,109.53);

\path[draw=drawColor,line width= 0.6pt,line join=round] ( 56.90,139.60) --
	(283.58,139.60);

\path[draw=drawColor,line width= 0.6pt,line join=round] ( 56.90,169.66) --
	(283.58,169.66);

\path[draw=drawColor,line width= 0.6pt,line join=round] ( 67.20, 49.41) --
	( 67.20,175.17);

\path[draw=drawColor,line width= 0.6pt,line join=round] (111.20, 49.41) --
	(111.20,175.17);

\path[draw=drawColor,line width= 0.6pt,line join=round] (155.20, 49.41) --
	(155.20,175.17);

\path[draw=drawColor,line width= 0.6pt,line join=round] (199.20, 49.41) --
	(199.20,175.17);

\path[draw=drawColor,line width= 0.6pt,line join=round] (243.20, 49.41) --
	(243.20,175.17);
\definecolor{fillColor}{RGB}{0,0,255}

\path[fill=fillColor] (273.28, 69.03) circle (  9.99);
\definecolor{drawColor}{RGB}{0,255,0}
\definecolor{iidjiq}{RGB}{60,180,60}

\path[draw=iidjiq,line width= 0.4pt,line join=round,line cap=round,fill=iidjiq] (133.64, 63.53) --
	(147.10, 86.84) --
	(120.18, 86.84) --
	cycle;
\definecolor{fillColor}{RGB}{255,0,0}

\path[fill=fillColor] (104.36,105.96) --
	(114.35,115.96) --
	(124.35,105.96) --
	(114.35, 95.97) --
	cycle;
\definecolor{srcjiq}{RGB}{180,180,0}

\path[fill=srcjiq] (238.52, 90.62) --
	(251.98, 67.30) --
	(225.06, 67.30) --
	cycle;
\definecolor{genjiqwoind}{RGB}{160,32,240}

\path[fill=genjiqwoind] (139.86, 77.95) --
	(159.85, 77.95) --
	(159.85, 97.94) --
	(139.86, 97.94) --
	cycle;
\definecolor{genjiqind}{RGB}{0,0,0}

\node[star,star points=9, fill=genjiqind, scale=  1.70] at (111.04, 164.61){};

\end{scope}
\begin{scope}
\path[clip] (  0.00,  0.00) rectangle (289.08,180.67);
\definecolor{drawColor}{RGB}{0,0,0}

\path[draw=drawColor,line width= 0.6pt,line join=round] ( 56.90, 49.41) --
	( 56.90,175.17);
\end{scope}
\begin{scope}
\path[clip] (  0.00,  0.00) rectangle (289.08,180.67);
\definecolor{drawColor}{gray}{0.30}

\node[text=drawColor,anchor=base east,inner sep=0pt, outer sep=0pt, scale=  1.90] at ( 51.95, 42.86) {0.0};

\node[text=drawColor,anchor=base east,inner sep=0pt, outer sep=0pt, scale=  1.90] at ( 51.95, 72.93) {0.4};

\node[text=drawColor,anchor=base east,inner sep=0pt, outer sep=0pt, scale=  1.90] at ( 51.95,102.99) {0.8};

\node[text=drawColor,anchor=base east,inner sep=0pt, outer sep=0pt, scale=  1.90] at ( 51.95,133.06) {1.2};

\node[text=drawColor,anchor=base east,inner sep=0pt, outer sep=0pt, scale=  1.90] at ( 51.95,163.12) {1.6};
\end{scope}
\begin{scope}
\path[clip] (  0.00,  0.00) rectangle (289.08,180.67);
\definecolor{drawColor}{gray}{0.20}

\path[draw=drawColor,line width= 0.6pt,line join=round] ( 54.15, 49.41) --
	( 56.90, 49.41);

\path[draw=drawColor,line width= 0.6pt,line join=round] ( 54.15, 79.47) --
	( 56.90, 79.47);

\path[draw=drawColor,line width= 0.6pt,line join=round] ( 54.15,109.53) --
	( 56.90,109.53);

\path[draw=drawColor,line width= 0.6pt,line join=round] ( 54.15,139.60) --
	( 56.90,139.60);

\path[draw=drawColor,line width= 0.6pt,line join=round] ( 54.15,169.66) --
	( 56.90,169.66);
\end{scope}
\begin{scope}
\path[clip] (  0.00,  0.00) rectangle (289.08,180.67);
\definecolor{drawColor}{RGB}{0,0,0}

\path[draw=drawColor,line width= 0.6pt,line join=round] ( 56.90, 49.41) --
	(283.58, 49.41);
\end{scope}
\begin{scope}
\path[clip] (  0.00,  0.00) rectangle (289.08,180.67);
\definecolor{drawColor}{gray}{0.20}

\path[draw=drawColor,line width= 0.6pt,line join=round] ( 67.20, 46.66) --
	( 67.20, 49.41);

\path[draw=drawColor,line width= 0.6pt,line join=round] (111.20, 46.66) --
	(111.20, 49.41);

\path[draw=drawColor,line width= 0.6pt,line join=round] (155.20, 46.66) --
	(155.20, 49.41);

\path[draw=drawColor,line width= 0.6pt,line join=round] (199.20, 46.66) --
	(199.20, 49.41);

\path[draw=drawColor,line width= 0.6pt,line join=round] (243.20, 46.66) --
	(243.20, 49.41);
\end{scope}
\begin{scope}
\path[clip] (  0.00,  0.00) rectangle (289.08,180.67);
\definecolor{drawColor}{gray}{0.30}

\node[text=drawColor,anchor=base,inner sep=0pt, outer sep=0pt, scale=  1.90] at ( 67.20, 31.37) {1.50};

\node[text=drawColor,anchor=base,inner sep=0pt, outer sep=0pt, scale=  1.90] at (111.20, 31.37) {1.52};

\node[text=drawColor,anchor=base,inner sep=0pt, outer sep=0pt, scale=  1.90] at (155.20, 31.37) {1.54};

\node[text=drawColor,anchor=base,inner sep=0pt, outer sep=0pt, scale=  1.90] at (199.20, 31.37) {1.56};

\node[text=drawColor,anchor=base,inner sep=0pt, outer sep=0pt, scale=  1.90] at (243.20, 31.37) {1.58};
\end{scope}
\begin{scope}
\path[clip] (  0.00,  0.00) rectangle (289.08,180.67);
\definecolor{drawColor}{RGB}{0,0,0}

\node[text=drawColor,anchor=base,inner sep=0pt, outer sep=0pt, scale=  2.20] at (170.24,  9.78) {mean $\mathbb E[T]$};
\end{scope}
\begin{scope}
\path[clip] (  0.00,  0.00) rectangle (289.08,180.67);
\definecolor{drawColor}{RGB}{0,0,0}

\node[text=drawColor,rotate= 90.00,anchor=base,inner sep=0pt, outer sep=0pt, scale=  2.20] at ( 20.65,112.29) {mean runtime};
\end{scope}
\end{tikzpicture}}     &
    \scalebox{.5}{
\begin{tikzpicture}[x=1pt,y=1pt]
\definecolor{fillColor}{RGB}{255,255,255}
\path[use as bounding box,fill=fillColor,fill opacity=0.00] (0,0) rectangle (289.08,180.67);
\begin{scope}
\path[clip] (  0.00,  0.00) rectangle (289.08,180.67);
\definecolor{drawColor}{RGB}{255,255,255}
\definecolor{fillColor}{RGB}{255,255,255}

\path[draw=drawColor,line width= 0.6pt,line join=round,line cap=round,fill=fillColor] (  0.00,  0.00) rectangle (289.08,180.68);
\end{scope}
\begin{scope}
\path[clip] ( 56.90, 49.41) rectangle (283.58,175.17);
\definecolor{drawColor}{RGB}{255,255,255}

\path[draw=drawColor,line width= 0.6pt,line join=round] ( 56.90, 49.41) --
	(283.58, 49.41);

\path[draw=drawColor,line width= 0.6pt,line join=round] ( 56.90, 79.47) --
	(283.58, 79.47);

\path[draw=drawColor,line width= 0.6pt,line join=round] ( 56.90,109.53) --
	(283.58,109.53);

\path[draw=drawColor,line width= 0.6pt,line join=round] ( 56.90,139.60) --
	(283.58,139.60);

\path[draw=drawColor,line width= 0.6pt,line join=round] ( 56.90,169.66) --
	(283.58,169.66);

\path[draw=drawColor,line width= 0.6pt,line join=round] ( 67.20, 49.41) --
	( 67.20,175.17);

\path[draw=drawColor,line width= 0.6pt,line join=round] (135.22, 49.41) --
	(135.22,175.17);

\path[draw=drawColor,line width= 0.6pt,line join=round] (203.23, 49.41) --
	(203.23,175.17);

\path[draw=drawColor,line width= 0.6pt,line join=round] (271.25, 49.41) --
	(271.25,175.17);
\definecolor{fillColor}{RGB}{0,0,255}

\path[fill=fillColor] (243.53, 68.20) circle (  9.99);
\definecolor{iidjiq}{RGB}{60,180,60}

\path[draw=iidjiq,line width= 0.4pt,line join=round,line cap=round,fill=iidjiq] (140.47, 61.52) --
	(153.93, 84.84) --
	(127.01, 84.84) --
	cycle;
\definecolor{fillColor}{RGB}{255,0,0}

\path[fill=fillColor] ( 86.45,107.51) --
	( 96.44,117.50) --
	(106.44,107.51) --
	( 96.44, 97.51) --
	cycle;
\definecolor{srcjiq}{RGB}{180,180,0}

\path[fill=srcjiq] (273.28, 91.11) --
	(286.74, 67.79) --
	(259.82, 67.79) --
	cycle;
\definecolor{fillColor}{RGB}{160,32,240}

\path[fill=fillColor] (123.49, 72.11) --
	(143.48, 72.11) --
	(143.48, 92.10) --
	(123.49, 92.10) --
	cycle;
\definecolor{genjiqind}{RGB}{0,0,0}
\node[star,star points=9, fill=genjiqind, scale=  1.70] at (89.93, 142.45){};

\end{scope}
\begin{scope}
\path[clip] (  0.00,  0.00) rectangle (289.08,180.67);
\definecolor{drawColor}{RGB}{0,0,0}

\path[draw=drawColor,line width= 0.6pt,line join=round] ( 56.90, 49.41) --
	( 56.90,175.17);
\end{scope}
\begin{scope}
\path[clip] (  0.00,  0.00) rectangle (289.08,180.67);
\definecolor{drawColor}{gray}{0.30}

\node[text=drawColor,anchor=base east,inner sep=0pt, outer sep=0pt, scale=  1.90] at ( 51.95, 42.86) {0.0};

\node[text=drawColor,anchor=base east,inner sep=0pt, outer sep=0pt, scale=  1.90] at ( 51.95, 72.93) {0.4};

\node[text=drawColor,anchor=base east,inner sep=0pt, outer sep=0pt, scale=  1.90] at ( 51.95,102.99) {0.8};

\node[text=drawColor,anchor=base east,inner sep=0pt, outer sep=0pt, scale=  1.90] at ( 51.95,133.06) {1.2};

\node[text=drawColor,anchor=base east,inner sep=0pt, outer sep=0pt, scale=  1.90] at ( 51.95,163.12) {1.6};
\end{scope}
\begin{scope}
\path[clip] (  0.00,  0.00) rectangle (289.08,180.67);
\definecolor{drawColor}{gray}{0.20}

\path[draw=drawColor,line width= 0.6pt,line join=round] ( 54.15, 49.41) --
	( 56.90, 49.41);

\path[draw=drawColor,line width= 0.6pt,line join=round] ( 54.15, 79.47) --
	( 56.90, 79.47);

\path[draw=drawColor,line width= 0.6pt,line join=round] ( 54.15,109.53) --
	( 56.90,109.53);

\path[draw=drawColor,line width= 0.6pt,line join=round] ( 54.15,139.60) --
	( 56.90,139.60);

\path[draw=drawColor,line width= 0.6pt,line join=round] ( 54.15,169.66) --
	( 56.90,169.66);
\end{scope}
\begin{scope}
\path[clip] (  0.00,  0.00) rectangle (289.08,180.67);
\definecolor{drawColor}{RGB}{0,0,0}

\path[draw=drawColor,line width= 0.6pt,line join=round] ( 56.90, 49.41) --
	(283.58, 49.41);
\end{scope}
\begin{scope}
\path[clip] (  0.00,  0.00) rectangle (289.08,180.67);
\definecolor{drawColor}{gray}{0.20}

\path[draw=drawColor,line width= 0.6pt,line join=round] ( 67.20, 46.66) --
	( 67.20, 49.41);

\path[draw=drawColor,line width= 0.6pt,line join=round] (135.22, 46.66) --
	(135.22, 49.41);

\path[draw=drawColor,line width= 0.6pt,line join=round] (203.23, 46.66) --
	(203.23, 49.41);

\path[draw=drawColor,line width= 0.6pt,line join=round] (271.25, 46.66) --
	(271.25, 49.41);
\end{scope}
\begin{scope}
\path[clip] (  0.00,  0.00) rectangle (289.08,180.67);
\definecolor{drawColor}{gray}{0.30}

\node[text=drawColor,anchor=base,inner sep=0pt, outer sep=0pt, scale=  1.90] at ( 67.20, 31.37) {0.97};

\node[text=drawColor,anchor=base,inner sep=0pt, outer sep=0pt, scale=  1.90] at (135.22, 31.37) {0.99};

\node[text=drawColor,anchor=base,inner sep=0pt, outer sep=0pt, scale=  1.90] at (203.23, 31.37) {1.01};

\node[text=drawColor,anchor=base,inner sep=0pt, outer sep=0pt, scale=  1.90] at (271.25, 31.37) {1.03};
\end{scope}
\begin{scope}
\path[clip] (  0.00,  0.00) rectangle (289.08,180.67);
\definecolor{drawColor}{RGB}{0,0,0}

\node[text=drawColor,anchor=base,inner sep=0pt, outer sep=0pt, scale=  2.20] at (170.24,  9.78) {median ET};
\end{scope}
\begin{scope}
\path[clip] (  0.00,  0.00) rectangle (289.08,180.67);
\definecolor{drawColor}{RGB}{0,0,0}

\node[text=drawColor,rotate= 90.00,anchor=base,inner sep=0pt, outer sep=0pt, scale=  2.00] at ( 20.65,112.29) {median $\mathbb E[T]$};
\end{scope}
\end{tikzpicture}}     &
    \scalebox{.6}{
\begin{tikzpicture}[x=1pt,y=1pt]


\begin{scope}
\path[clip] (  0.00,  0.00) rectangle (300,105);

\definecolor{indjiq}{RGB}{255,0,0}
\path[fill=indjiq] (50,10+20) --
	(57,17+20) --
	(50,24+20) --
	(43,17+20) --
	cycle;
\node[anchor=west,inner sep=0pt, outer sep=0pt, scale=  1.40] at (65,17+20) {{\footnotesize  $\ipoptd{\indq}$}};

\definecolor{brjiq}{RGB}{0,0,255}
\path[fill=brjiq] (175,17+20) circle (  7.5);
\node[anchor=west,inner sep=0pt, outer sep=0pt, scale=  1.40] at (190,17+20) {{\footnotesize $\ipoptd{\brq}$}};

\definecolor{genjiqwoind}{RGB}{160,32,240}
\path[fill=genjiqwoind] (43,40+20) --
	(57,40+20) --
	(57,54+20) --
	(43,54+20) --
	cycle;
\node[anchor=west,inner sep=0pt, outer sep=0pt, scale=  1.40] at (65,47+20) {{\footnotesize $\ipoptd{\genq}$}};

\definecolor{srcjiq}{RGB}{180,180,0}
\path[fill=srcjiq] (175,54+20) --
	(182,40+20) --
	(168,40+20) --
	cycle;
\node[anchor=west,inner sep=0pt, outer sep=0pt, scale=  1.40] at (190,47+20) {{\footnotesize  $\ipoptd{\srcq}$}};

\definecolor{genjiqind}{RGB}{0,0,0}
\node[star,star points=9, fill=genjiqind, scale=  1.20] at (50, 77+20){};
\node[anchor=west,inner sep=0pt, outer sep=0pt, scale=  1.40] at (65,77+20) {{\footnotesize $\genseed$}};

\definecolor{iidjiq}{RGB}{60,180,60}
\path[draw=iidjiq,line width= 0.4pt,line join=round,line cap=round,fill=iidjiq] (168,84+20) --
	(182,84+20) --
	(175,70+20) --
	cycle;
\node[anchor=west,inner sep=0pt, outer sep=0pt, scale=  1.40] at (190,77+20) {{\footnotesize $\ipoptd{\iidq}$}};

\end{scope}
\end{tikzpicture}

} \\
    \end{tabular}\\
    \caption{Plots of the (mean $\mathbb E[T]$, mean runtime) pairs (left) and (median $\mathbb E[T]$, median runtime) pairs (right) calculated across all parameter settings defined in Section~\ref{sec:parameters} for six dispatching policies.}
    \label{fig:mean}
\end{figure}

Both the mean and the median results indicate that there is a tradeoff between $\mathbb E[T]$ and runtime: families that require a longer runtime to solve the optimization problem tend to yield lower $\mathbb E[T]$ values.
Note, however, that the trends exhibited in Figure~\ref{fig:mean} do not imply that the families have the same ordering with respect to $\mathbb E[T]$ and runtime for any \emph{specific} parameter setting; indeed, some trends suggested by Figure~\ref{fig:mean}(left) are reversed in Figure~\ref{fig:mean}(right).
For example, while $\ipoptd\iidq$ appears to dominate $\ipoptd\genq$ with respect to both mean measures, $\ipoptd\iidq$ has a higher (i.e., worse) median $\mathbb E[T]$ value than $\ipoptd\genq$.

Overall $\ipoptd\brq$ and $\ipoptd\srcq$ feature the lowest runtimes but at the expense of the worst performance (i.e., they have the highest $\mathbb E[T]$ values).  The fast runtime of $\ipoptd\brq$ can be attributed to the fact that it need only optimize over assignment; despite arising from the ``smallest'' of the optimization problems in many respects (see Table~2), $\ipoptd\srcq$ features a higher runtime than $\ipoptd\brq$.  Meanwhile, $\ipoptd\iidq$ and $\ipoptd\genq$ offer an improvement in performance at the cost of additional runtime.  Surprisingly, $\ipoptd\indq$ has a longer runtime (and as previously discussed, better performance than) $\ipoptd\genq$ despite arising from solving a problem of a smaller (nominal) size.  An examination of the optimization problem associated with $\DP\indq\cida$, as presented in Appendix C, provides a potential explanation for the exceptionally long runtimes associated with $\ipoptd\indq$: the constraints in this optimization problem with $\lb$ on the left-hand side are very complicated.  As previously noted, $\genseed$ achieves the best $\mathbb E[T]$ by building off of the strong performance of $\ipoptd\indq$.  Of course, this comes at a significant runtime expense, as one must now solve two optimization problems.

Throughout the entire parameter set, all of the optimization problems ran in well under one minute; the mean and median runtime associated with each of the dispatching policies examined was under 2 seconds.
In practice, these differences in runtimes are small enough that they would likely not be a significant factor, as this optimization would only need to be performed once to configure the system.
Thus, while the tradeoff between $\mathbb E[T]$ and runtime is of theoretical interest, $\genseed$ represents the best practical choice of dispatching policy among those studied here when achieving low $\mathbb E[T]$ is the foremost goal.
On the other hand, if the simplicity or interpretability of the policy is of value to the system designer, $\ipoptd\iidq$ provides a reasonable alternative to $\genseed$.

\begin{figure}
    \centering
    \scalebox{.7}{
\begin{tikzpicture}[x=1pt,y=1pt]
\definecolor{fillColor}{RGB}{255,255,255}
\path[use as bounding box,fill=fillColor,fill opacity=0.00] (0,0) rectangle (578.16,216.81);
\begin{scope}
\path[clip] (  0.00,  0.00) rectangle (578.16,216.81);
\definecolor{drawColor}{RGB}{255,255,255}
\definecolor{fillColor}{RGB}{255,255,255}

\path[draw=drawColor,line width= 0.6pt,line join=round,line cap=round,fill=fillColor] (  0.00,  0.00) rectangle (578.16,216.81);
\end{scope}
\begin{scope}
\path[clip] ( 83.08, 42.34) rectangle (444.21,211.31);
\definecolor{drawColor}{RGB}{255,255,255}

\path[draw=drawColor,line width= 0.6pt,line join=round] ( 83.08, 42.34) --
	(444.21, 42.34);

\path[draw=drawColor,line width= 0.6pt,line join=round] ( 83.08, 85.47) --
	(444.21, 85.47);

\path[draw=drawColor,line width= 0.6pt,line join=round] ( 83.08,128.60) --
	(444.21,128.60);

\path[draw=drawColor,line width= 0.6pt,line join=round] ( 83.08,171.73) --
	(444.21,171.73);

\path[draw=drawColor,line width= 0.6pt,line join=round] ( 83.08, 42.34) --
	( 83.08,211.31);

\path[draw=drawColor,line width= 0.6pt,line join=round] (175.20, 42.34) --
	(175.20,211.31);

\path[draw=drawColor,line width= 0.6pt,line join=round] (267.33, 42.34) --
	(267.33,211.31);

\path[draw=drawColor,line width= 0.6pt,line join=round] (359.46, 42.34) --
	(359.46,211.31);
\definecolor{indjiq}{RGB}{255,0,0}

\path[draw=indjiq,line width= 2.0pt,dash pattern=on 4pt off 4pt ,line join=round] ( 90.45, 42.34) --
	( 97.82, 42.34) --
	(105.19, 42.34) --
	(112.56, 42.34) --
	(119.93, 42.34) --
	(142.04, 42.34) --
	(149.41, 42.34) --
	(164.15, 42.34) --
	(171.52, 42.34) --
	(178.89, 42.34) --
	(186.26, 42.34) --
	(201.00, 42.34) --
	(208.37, 42.34) --
	(223.11, 42.34) --
	(230.48, 42.89) --
	(245.22, 44.29) --
	(252.59, 44.98) --
	(267.33, 46.11) --
	(274.70, 46.55) --
	(282.07, 46.92) --
	(289.44, 42.34) --
	(304.18, 42.34) --
	(311.55, 42.34) --
	(318.92, 42.34) --
	(326.29, 42.34) --
	(333.66, 42.34) --
	(355.77, 42.34) --
	(363.14, 42.34) --
	(370.51, 42.34) --
	(377.88, 42.34) --
	(399.99, 42.34) --
	(407.36, 42.34) --
	(414.73, 42.34) --
	(422.10, 42.34) --
	(444.21, 43.10);
\definecolor{iidjiq}{RGB}{60,180,60}

\path[draw=iidjiq,line width= 2.0pt,dash pattern=on 2pt off 2pt ,line join=round] ( 90.45, 42.34) --
	( 97.82, 42.34) --
	(105.19, 42.34) --
	(112.56, 42.34) --
	(119.93, 42.34) --
	(142.04, 42.34) --
	(149.41, 42.34) --
	(164.15, 42.34) --
	(171.52, 42.34) --
	(178.89, 42.34) --
	(186.26, 42.34) --
	(201.00, 42.34) --
	(208.37, 42.34) --
	(223.11, 42.42) --
	(230.48, 42.89) --
	(245.22, 44.29) --
	(252.59, 44.98) --
	(267.33, 46.11) --
	(274.70, 46.55) --
	(282.07, 46.92) --
	(289.44, 47.32) --
	(304.18, 48.91) --
	(311.55, 50.04) --
	(318.92, 51.33) --
	(326.29, 52.77) --
	(333.66, 54.33) --
	(355.77, 55.93) --
	(363.14, 55.22) --
	(370.51, 54.68) --
	(377.88, 54.05) --
	(399.99, 50.23) --
	(407.36, 48.91) --
	(414.73, 47.65) --
	(422.10, 46.45) --
	(444.21, 43.10);
\definecolor{brjiq}{RGB}{0,0,255}

\path[draw=brjiq,line width= 2.0pt,dash pattern=on 1pt off 3pt on 4pt off 3pt ,line join=round] ( 90.45,113.55) --
	( 97.82,120.20) --
	(105.19,127.41) --
	(112.56,135.08) --
	(119.93,143.14) --
	(142.04,168.84) --
	(149.41,177.66) --
	(164.15,195.23) --
	(171.52,203.06) --
	(178.89,208.54) --
	(186.26,211.31) --
	(201.00,209.92) --
	(208.37,205.83) --
	(223.11,190.30) --
	(230.48,183.19) --
	(245.22,174.53) --
	(252.59,172.41) --
	(267.33,171.23) --
	(274.70,167.25) --
	(282.07,155.97) --
	(289.44,138.61) --
	(304.18,107.14) --
	(311.55, 92.81) --
	(318.92, 80.71) --
	(326.29, 71.70) --
	(333.66, 65.50) --
	(355.77, 58.06) --
	(363.14, 58.10) --
	(370.51, 60.41) --
	(377.88, 65.31) --
	(399.99, 71.86) --
	(407.36, 73.09) --
	(414.73, 74.27) --
	(422.10, 75.37) --
	(444.21, 77.98);
\definecolor{srcjiq}{RGB}{180,180,0}

\path[draw=srcjiq,line width= 2.0pt,line join=round] ( 90.45, 42.34) --
	( 97.82, 42.34) --
	(105.19, 42.34) --
	(112.56, 42.34) --
	(119.93, 42.34) --
	(142.04, 42.34) --
	(149.41, 42.34) --
	(164.15, 42.34) --
	(171.52, 42.34) --
	(178.89, 42.34) --
	(186.26, 42.34) --
	(201.00, 42.34) --
	(208.37, 42.34) --
	(223.11, 43.05) --
	(230.48, 48.08) --
	(245.22, 62.88) --
	(252.59, 66.37) --
	(267.33, 71.17) --
	(274.70, 73.24) --
	(282.07, 75.10) --
	(289.44, 78.23) --
	(304.18, 92.73) --
	(311.55, 97.23) --
	(318.92,101.41) --
	(326.29,105.91) --
	(333.66,110.94) --
	(355.77,123.65) --
	(363.14,126.09) --
	(370.51,127.83) --
	(377.88,127.91) --
	(399.99,113.86) --
	(407.36,108.38) --
	(414.73,103.04) --
	(422.10, 97.85) --
	(444.21, 83.16);
\end{scope}
\begin{scope}
\path[clip] (  0.00,  0.00) rectangle (578.16,216.81);
\definecolor{drawColor}{RGB}{0,0,0}

\path[draw=drawColor,line width= 0.6pt,line join=round] ( 83.08, 42.34) --
	( 83.08,211.31);
\end{scope}
\begin{scope}
\path[clip] (  0.00,  0.00) rectangle (578.16,216.81);
\definecolor{drawColor}{gray}{0.30}

\node[text=drawColor,anchor=base east,inner sep=0pt, outer sep=0pt, scale=  1.70] at ( 78.13, 36.49) {1.00};

\node[text=drawColor,anchor=base east,inner sep=0pt, outer sep=0pt, scale=  1.70] at ( 78.13, 79.62) {1.05};

\node[text=drawColor,anchor=base east,inner sep=0pt, outer sep=0pt, scale=  1.70] at ( 78.13,122.75) {1.10};

\node[text=drawColor,anchor=base east,inner sep=0pt, outer sep=0pt, scale=  1.70] at ( 78.13,165.87) {1.15};
\end{scope}
\begin{scope}
\path[clip] (  0.00,  0.00) rectangle (578.16,216.81);
\definecolor{drawColor}{gray}{0.20}

\path[draw=drawColor,line width= 0.6pt,line join=round] ( 80.33, 42.34) --
	( 83.08, 42.34);

\path[draw=drawColor,line width= 0.6pt,line join=round] ( 80.33, 85.47) --
	( 83.08, 85.47);

\path[draw=drawColor,line width= 0.6pt,line join=round] ( 80.33,128.60) --
	( 83.08,128.60);

\path[draw=drawColor,line width= 0.6pt,line join=round] ( 80.33,171.73) --
	( 83.08,171.73);
\end{scope}
\begin{scope}
\path[clip] (  0.00,  0.00) rectangle (578.16,216.81);
\definecolor{drawColor}{RGB}{0,0,0}

\path[draw=drawColor,line width= 0.6pt,line join=round] ( 83.08, 42.34) --
	(444.21, 42.34);
\end{scope}
\begin{scope}
\path[clip] (  0.00,  0.00) rectangle (578.16,216.81);
\definecolor{drawColor}{gray}{0.20}

\path[draw=drawColor,line width= 0.6pt,line join=round] ( 83.08, 39.59) --
	( 83.08, 42.34);

\path[draw=drawColor,line width= 0.6pt,line join=round] (175.20, 39.59) --
	(175.20, 42.34);

\path[draw=drawColor,line width= 0.6pt,line join=round] (267.33, 39.59) --
	(267.33, 42.34);

\path[draw=drawColor,line width= 0.6pt,line join=round] (359.46, 39.59) --
	(359.46, 42.34);
\end{scope}
\begin{scope}
\path[clip] (  0.00,  0.00) rectangle (578.16,216.81);
\definecolor{drawColor}{gray}{0.30}

\node[text=drawColor,anchor=base,inner sep=0pt, outer sep=0pt, scale=  1.70] at ( 83.08, 25.68) {0};

\node[text=drawColor,anchor=base,inner sep=0pt, outer sep=0pt, scale=  1.70] at (175.20, 25.68) {0.25};

\node[text=drawColor,anchor=base,inner sep=0pt, outer sep=0pt, scale=  1.70] at (267.33, 25.68) {0.50};

\node[text=drawColor,anchor=base,inner sep=0pt, outer sep=0pt, scale=  1.70] at (359.46, 25.68) {0.75};
\end{scope}
\begin{scope}
\path[clip] (  0.00,  0.00) rectangle (578.16,216.81);
\definecolor{drawColor}{RGB}{0,0,0}

\node[text=drawColor,anchor=base,inner sep=0pt, outer sep=0pt, scale=  1.60] at (263.64,  8.61) {$\lambda$};
\end{scope}
\begin{scope}
\path[clip] (  0.00,  0.00) rectangle (578.16,216.81);
\definecolor{drawColor}{RGB}{0,0,0}

\node[text=drawColor,rotate= 90.00,anchor=base west,inner sep=0pt, outer sep=0pt, scale=  1.60] at ( 16.52, 50) {{\footnotesize $\mathbb E[T]$ normalized by that}};

\node[text=drawColor,rotate= 90.00,anchor=base west,inner sep=0pt, outer sep=0pt, scale=  1.60] at ( 38.94, 65) {{\footnotesize of $\genseed$}};
\end{scope}

\begin{scope}
\path[clip] (  0.00,  0.00) rectangle (650,400.67);
\definecolor{indjiq}{RGB}{255,0,0}
\path[draw=indjiq,line width= 2.0pt,dash pattern=on 4pt off 4pt ,line join=round] ( 250+200,260+3+10-70-20) -- ( 270+200+10,260+3+10-70-20);
\node[anchor=west,inner sep=0pt, outer sep=0pt, scale= 1.40] at (280+210,260+3+10-70-20) {{\footnotesize $\ipoptd{\indq}$}};

\definecolor{iidjiq}{RGB}{60,180,60}
\path[draw=iidjiq,line width= 2.0pt,dash pattern=on 2pt off 2pt ,line join=round] ( 250+200,240+3+10-75-20-5) -- ( 270+210,240+3+10-75-20-5);
\node[anchor=west,inner sep=0pt, outer sep=0pt, scale= 1.40] at (280+210,240+3+10-75-20-5) {{\footnotesize $\ipoptd{\iidq}$}};

\definecolor{srcjiq}{RGB}{180,180,0}
\path[draw=srcjiq,line width= 2.0pt,line join=round] ( ( 250+200,220+3+10-80-20-10) -- ( 270+210,220+3+10-80-20-10);
\node[anchor=west,inner sep=0pt, outer sep=0pt, scale= 1.40] at (280+210,220+3+10-80-20-10) {{\footnotesize $\ipoptd{\srcq}$}};

\definecolor{brjiq}{RGB}{0,0,255}
\path[draw=brjiq,line width= 2.0pt,dash pattern=on 1pt off 3pt on 4pt off 3pt ,line join=round] ( 250+200,200+3+10-85-20-15) -- ( 270+210,200+3+10-85-20-15);
\node[anchor=west,inner sep=0pt, outer sep=0pt, scale= 1.40] at (280+210,200+3+10-85-20-15) {{\footnotesize $\ipoptd{\brq}$}};
\end{scope}

\end{tikzpicture}}
    \caption{$\mathbb E[T]$ relative to that of $\genseed$ (i.e., $\mathbb E[T]^{\mathsf{DP}}/\mathbb E[T]^{\genseed}$) as a function of $\lambda$ for the parameter settings where $s=d=3$, $\lambda$ varies over $\{0.02,0.04,\ldots,0.98\}$,  $(q_1,q_2,q_3)=(1/3,1/6,1/2)$ and $(R_1,R_2)=(5,2)$, yielding $(\mu_1,\mu_2,\mu_3)=(2,4/5,2/5)$, for the dispatching policies $\mathsf{DP}\in\{\ipoptd{\indq},\ipoptd{\iidq},\ipoptd{\srcq},\ipoptd{\brq}\}$.}
    \label{fig:lambda-plot}
\end{figure}
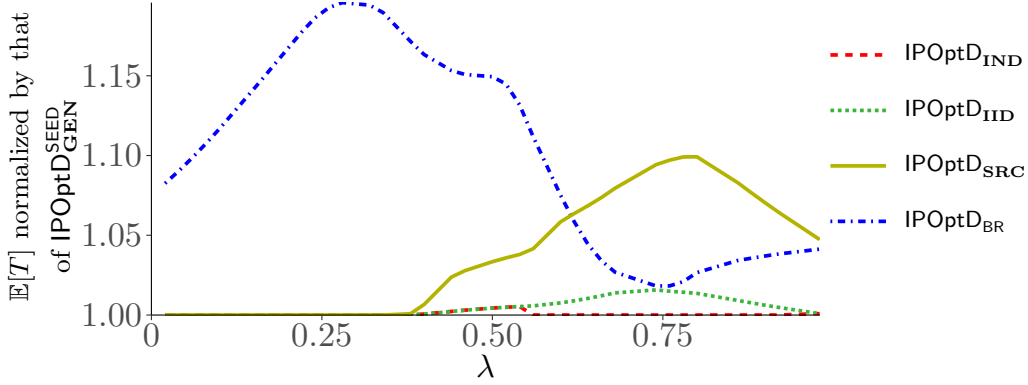

The results in Figure~\ref{fig:mean} were aggregated across the entire space of parameter settings; in Figure~\ref{fig:lambda-plot} we instead present results for a collection of settings where all parameters are fixed except for $\lambda$. This allows us to provide a \emph{direct} comparison of how our dispatching policies perform with respect to their $\mathbb E[T]$ values across the spectrum of arrival rates.
We first observe that $\ipoptd\srcq$ and $\ipoptd\brq$ each exhibit their best performance in different ranges of $\lambda$ values, which provides some insight into the reversed relationship these policies exhibit when comparing their mean and median $\mathbb E[T]$ values over the entire parameter set.  That said, these two policies perform considerably worse than the other policies examined. Both $\ipoptd\iidq$ and $\ipoptd\indq$ achieve performance comparable to $\genseed$, with $\ipoptd\indq$ indistinguishable from $\genseed$ at all but a small range of $\lambda$ values.


\subsection{The ``optimal'' dispatching policies found by IPOPT}

In the previous subsection we have evaluated the \emph{performance} of the dispatching policies resulting from the optimal solutions found by IPOPT when given the optimization problems associated with the $\DP{\qrf}{\cida}$ family of dispatching rules under various querying rule families $\qrf$.  In this subsection, we turn to numerically studying the \emph{solutions} (as found by IPOPT) themselves; that is, we study the best policies found by IPOPT across our parameter settings, although to facilitate comprehensible visualizations, we restrict attention to the setting where $s=d=2$.  We also restrict attention to those (IPOPT-determined) dispatching policy families that performed best based on the study from the previous subsection: $\ipoptd\iidq$, $\ipoptd\indq$, $\ipoptd\genq$, and $\genseed$.  We caution that the results presented here may reveal more about the idiosyncrasies of IPOPT than they do about the ``true optimal'' policies belong to the dispatching policy families of interest.

\begin{figure}
    \begin{subfigure}[t]{.8\textwidth}
        \centering
        \scalebox{.5}{
\begin{tikzpicture}[x=1pt,y=1pt]
\definecolor{fillColor}{RGB}{255,255,255}
\path[use as bounding box,fill=fillColor,fill opacity=0.00] (0,0) rectangle (867.24, 72.27);
\begin{scope}
\path[clip] (  0.00,  0.00) rectangle (867.24, 72.27);
\definecolor{fillColor}{RGB}{255,255,255}

\path[fill=fillColor] (168.17,  4.71) rectangle (699.07, 67.56);
\end{scope}
\begin{scope}
\path[clip] (  0.00,  0.00) rectangle (867.24, 72.27);
\node[inner sep=0pt,outer sep=0pt,anchor=south west,rotate=  0.00] at (266.78,  33.60) {
	\pgfimage[width=426.79pt,height= 28.45pt,interpolate=true]{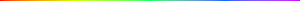}};
\end{scope}
\begin{scope}
\path[clip] (  0.00,  0.00) rectangle (867.24, 72.27);
\definecolor{drawColor}{RGB}{0,0,0}

\node[text=drawColor,anchor=base,inner sep=0pt, outer sep=0pt, scale=  1.90] at (267.49, 14.71) {0.00};

\node[text=drawColor,anchor=base,inner sep=0pt, outer sep=0pt, scale=  1.90] at (373.83, 14.71) {0.25};

\node[text=drawColor,anchor=base,inner sep=0pt, outer sep=0pt, scale=  1.90] at (480.17, 14.71) {0.50};

\node[text=drawColor,anchor=base,inner sep=0pt, outer sep=0pt, scale=  1.90] at (586.52, 14.71) {0.75};

\node[text=drawColor,anchor=base,inner sep=0pt, outer sep=0pt, scale=  1.90] at (692.86, 14.71) {1.00};
\end{scope}
\begin{scope}
\path[clip] (  0.00,  0.00) rectangle (867.24, 72.27);
\definecolor{drawColor}{RGB}{0,0,0}

\node[text=drawColor,anchor=base west,inner sep=0pt, outer sep=0pt, scale=  1.10] at (100, 40) {\Huge $\alpha_1\left(\zerovec,\left(1,1\right)\right)$};
\end{scope}
\begin{scope}
\path[clip] (  0.00,  0.00) rectangle (867.24, 72.27);
\definecolor{drawColor}{RGB}{255,255,255}

\path[draw=drawColor,line width= 0.2pt,line join=round] (267.49, 33.60) -- (267.49, 39.29);

\path[draw=drawColor,line width= 0.2pt,line join=round] (373.83, 33.60) -- (373.83, 39.29);

\path[draw=drawColor,line width= 0.2pt,line join=round] (480.17, 33.60) -- (480.17, 39.29);

\path[draw=drawColor,line width= 0.2pt,line join=round] (586.52, 33.60) -- (586.52, 39.29);

\path[draw=drawColor,line width= 0.2pt,line join=round] (692.86, 33.60) -- (692.86, 39.29);

\path[draw=drawColor,line width= 0.2pt,line join=round] (267.49, 56.37) -- (267.49, 62.06);

\path[draw=drawColor,line width= 0.2pt,line join=round] (373.83, 56.37) -- (373.83, 62.06);

\path[draw=drawColor,line width= 0.2pt,line join=round] (480.17, 56.37) -- (480.17, 62.06);

\path[draw=drawColor,line width= 0.2pt,line join=round] (586.52, 56.37) -- (586.52, 62.06);

\path[draw=drawColor,line width= 0.2pt,line join=round] (692.86, 56.37) -- (692.86, 62.06);
\end{scope}
\end{tikzpicture}}
    \end{subfigure}
    \smallskip
    \vspace{-2.5em}
    \begin{subfigure}[t]{.35\textwidth}
        \centering
        \scalebox{.5}{\input{figs/TernaryAndPolicies/query_IIDCID}}
    \vspace*{-10mm}\caption{$\ipoptd\iidq$}
    \end{subfigure}
    \hspace{5em}
    \begin{subfigure}[t]{.35\textwidth}
        \centering
        \scalebox{.5}{\input{figs/TernaryAndPolicies/query_INDCID}}
    \vspace*{-10mm}\caption{$\ipoptd\indq$}
    \end{subfigure}
    
    \medskip
    \vspace{-1cm}
    
    \begin{subfigure}[t]{.35\textwidth}
        \centering
        \scalebox{.5}{\input{figs/TernaryAndPolicies/query_GENCID}}
    \vspace*{-10mm}\caption{$\ipoptd\genq$}
    \end{subfigure}
    \hspace{5em}
    \begin{subfigure}[t]{.35\textwidth}
        \centering
        \scalebox{.5}{\input{figs/TernaryAndPolicies/query_GENCIDseeded}}
    \vspace*{-10mm}\caption{$\genseed$}
    \end{subfigure}
    \caption{IPOPT optimal dispatching policies $s=d=2$.  The parameter $\alpha_1(2,(1,1))=0$ in all cases shown.}
    \label{fig:query_results}
\end{figure}

Figure~\ref{fig:query_results} shows four plots---one for each of the aforementioned families of dispatching policies.  In each plot the optimal policy associated with each parameter setting is denoted by a single point.  Each point's position in the ternary plot gives the values of $p(2,0)$, $p(1,1)$, and $p(0,2)$, which collectively describe the policy's querying rule; note that $p(2,0)+p(1,1)+p(0,2)=1$.  Meanwhile, the color or shading of each point denotes the $\alpha_1(\mathbf0,(1,1))$ parameter associated with the policy---in all cases plotted, this single parameter uniquely identifies the assignment rule, as in all such cases IPOPT reported $\alpha_1(2,(1,1))=0$, and all other assignment rule parameters can be computed from these two.  That is, we find that in all of the optimal policies reported by IPOPT, jobs are never assigned to busy class-$1$ servers when an idle class-$2$ server has been queried.  

\begin{remark}
In Figure~\ref{fig:query_results}, we see that the lowest values of $\alpha_1(\mathbf0,(1,1))$ are associated with those policies where  $p(1,1)=0$; however, such policies are precisely those where the choice of $\alpha_1(\mathbf0,(1,1))$ is immaterial; When we fix $p(1,1)=0$, $\mathbb E[T]$ is entirely insensitive to $\alpha_1(\mathbf0,(1,1))$, as it doesn't matter how we assign jobs under the query mix $(1,1)$ if the probability of querying according to such a mix is set to zero.
\end{remark}

Upon looking at Figure~\ref{fig:query_results} we immediately observe the following: (i) all $\ipoptd\iidq$ policies lie on a ``curve'' on the ternary plot, (ii) all $\ipoptd\indq$ policies lie on either a curve or satisfy at least one of the $p(2,0)=0$ or $p(0,2)=0$ line segments, (iii) all $\ipoptd\genq$ policies line on at least one of the $p(2,0)=0$, $p(1,1)=0$, or $p(0,2)=0$ line segments, and (iv) the $\genseedq$ policies exhibit qualitatively similar behavior to that associated with $\ipoptd\indq$, although far fewer of the points lie on a curve.

Closer inspect reveals that all of the curves alluded to above are indeed the same.  In fact, this curve is defined by \[\left\{(p(2,0),p(1,1),p(0,2))=\left(x^2,2x(1-x),(1-x)^2\right)\colon x\in[0,1]\right\},\] which is precisely the set of querying rules comprising $\iidq$ when $s=d=2$ (i.e., this is the the feasible set of querying rules for the optimization problem associated with $\DP{\iidq}{\cida}$ when $s=d=2$).  Meanwhile, with some work, one can show that the set of querying rules comprised by $\indq$ when $s=d=2$ corresponds to the region bounded by the ``$\iidq$ curve'' and the lines $p(2,0)=0$ and $p(0,2)=0$ (inclusive). We find that \emph{every} querying rule reported as optimal by IPOPT that is contained within $\indq$---which includes all of the querying rules of the $\genseedq$ policies---is specifically contained within the \emph{boundary} of $\indq$.

The only dispatching policies reported by IPOPT that do not use $\indq$ querying are a subset of the $\ipoptq{\genq}$ policies where $p(1,1)=0$; in fact, such policies use $\srcq$ querying.  Meanwhile, the $\genseedq$ policies are always within $\indq$: it appears that seeding IPOPT with the ``$\indq$ solution'' when giving it a ``$\genq$ problem'' allows IPOPT to avoid finding dispatching policies using $\srcq$ querying in favor of those using $\indq$ querying, ultimately yielding better performance.  Moreover, we observe that very few $\genseed$ policies actually lie on the curve (i.e., very few such policies are in $\iidq$); the optimal policies tend to be those where either $p(2,0)=0$ or $p(0,2)=0$.  Such querying rules are precisely those that are either deterministic, or ``semi-deterministic'' in the sense that at least one server of a chosen (fixed) class $i\in\{1,2\}$, and determines the class for the remaining query randomly.

We leave it to future work to determine whether, how, and to what extent these observations can (i) yield results concerning ``true optimal'' policies and (ii) be generalized to the cases where $s>2$ and/or $d>2$.

\begin{remark}
While we have avoided plotting results for $\ipoptd{\brq}$ in the interest of brevity, we note here that we find two kinds of solutions associated with the optimization problem for $\ipoptd{\brq}$: those where $\alpha_1(2,(1,1))=0$ as in the case of the policies discussed above, and those where $\alpha_1(\mathbf0,(1,1))=1$, while $\alpha_1(2,(1,1))>0$.  The latter policies are precisely those where jobs are never assigned to a class-$2$ server when a class-$1$ server has been queried except---and only sometimes---when said class-$2$ server is idle and the class-$1$ server is busy.   We conjecture that the need to consider such policies under $\brq$ is a result of the fact that $\brq$ prohibits any optimization associated with the querying rule.
\end{remark}


\subsection{Performance under the $\detq$ querying rule family}

\begin{figure}
    \centering
    \scalebox{.5}{
\begin{tikzpicture}[x=1pt,y=1pt]
\definecolor{fillColor}{RGB}{255,255,255}
\path[use as bounding box,fill=fillColor,fill opacity=0.00] (0,0) rectangle (505.89,216.81);
\begin{scope}
\path[clip] ( 0.00, 0.00) rectangle (505.89,216.81);
\definecolor{drawColor}{RGB}{255,255,255}
\definecolor{fillColor}{RGB}{255,255,255}

\path[draw=drawColor,line width= 0.6pt,line join=round,line cap=round,fill=fillColor] ( 0.00, 0.00) rectangle (505.89,216.81);
\end{scope}
\begin{scope}
\path[clip] ( 78.58, 44.11) rectangle (500.39,211.31);
\definecolor{drawColor}{RGB}{255,255,255}

\path[draw=drawColor,line width= 0.6pt,line join=round] ( 78.58, 73.42) --
(500.39, 73.42);

\path[draw=drawColor,line width= 0.6pt,line join=round] ( 78.58,116.85) --
(500.39,116.85);

\path[draw=drawColor,line width= 0.6pt,line join=round] ( 78.58,160.28) --
(500.39,160.28);

\path[draw=drawColor,line width= 0.6pt,line join=round] ( 78.58,203.71) --
(500.39,203.71);

\path[draw=drawColor,line width= 0.6pt,line join=round] ( 78.58, 44.11) --
( 78.58,211.31);

\path[draw=drawColor,line width= 0.6pt,line join=round] (186.19, 44.11) --
(186.19,211.31);

\path[draw=drawColor,line width= 0.6pt,line join=round] (293.79, 44.11) --
(293.79,211.31);

\path[draw=drawColor,line width= 0.6pt,line join=round] (401.39, 44.11) --
(401.39,211.31);
\definecolor{drawColor}{RGB}{0,0,0}

\path[draw=drawColor,line width= 0.6pt,line join=round] ( 78.58, 73.42) -- (500.39, 73.42);
\definecolor{detjiq22}{RGB}{255,0,0}

\path[draw=detjiq22,line width= 2.3pt,dash pattern=on 7pt off 3pt ,line join=round] ( 87.19, 73.52) --
( 95.80, 73.81) --
(104.41, 74.29) --
(113.02, 74.96) --
(121.63, 75.82) --
(130.23, 76.87) --
(138.84, 78.12) --
(147.45, 79.55) --
(156.06, 81.18) --
(164.67, 83.00) --
(173.28, 85.01) --
(181.88, 87.21) --
(190.49, 89.61) --
(199.10, 92.19) --
(207.71, 94.97) --
(216.32, 97.94) --
(224.93,101.10) --
(233.53,104.45) --
(242.14,107.99) --
(250.75,111.73) --
(259.36,115.65) --
(267.97,119.77) --
(276.57,124.08) --
(285.18,128.58) --
(293.79,133.27) --
(302.40,138.16) --
(311.01,143.23) --
(319.62,148.50) --
(328.22,153.94) --
(336.83,158.42) --
(345.44,161.17) --
(354.05,162.12) --
(362.66,162.21) --
(371.27,162.01) --
(379.87,161.44) --
(388.48,160.61) --
(397.09,160.99) --
(405.70,162.59) --
(414.31,165.08) --
(422.92,168.22) --
(431.52,171.83) --
(440.13,175.74) --
(448.74,179.81) --
(457.35,183.83) --
(465.96,187.50) --
(474.57,191.00) --
(483.17,195.27) --
(491.78,200.86) --
(500.39,206.85);
\definecolor{detjsq22}{RGB}{0,0,255}

\path[draw=detjsq22,line width= 1.5pt,line join=round] ( 87.19, 73.52) --
( 95.80, 73.81) --
(104.41, 74.28) --
(113.02, 74.94) --
(121.63, 75.79) --
(130.23, 76.82) --
(138.84, 78.01) --
(147.45, 79.38) --
(156.06, 80.90) --
(164.67, 82.58) --
(173.28, 84.39) --
(181.88, 86.34) --
(190.49, 88.40) --
(199.10, 90.57) --
(207.71, 92.83) --
(216.32, 95.17) --
(224.93, 97.56) --
(233.53,100.00) --
(242.14,102.47) --
(250.75,104.93) --
(259.36,107.38) --
(267.97,109.79) --
(276.57,112.12) --
(285.18,114.35) --
(293.79,116.46) --
(302.40,118.39) --
(311.01,120.11) --
(319.62,121.58) --
(328.22,122.75) --
(336.83,123.55) --
(345.44,123.92) --
(354.05,123.79) --
(362.66,123.06) --
(371.27,121.66) --
(379.87,119.45) --
(388.48,116.14) --
(397.09,111.86) --
(405.70,106.48) --
(414.31,100.45) --
(422.92, 93.79) --
(431.52, 86.15) --
(440.13, 77.16) --
(448.74, 66.38) --
(457.35, 53.32) --
(465.96, 37.30) --
(474.57, 17.77) --
(481.02, 0.00);
\definecolor{detjiq}{RGB}{60,180,60}

\path[draw=detjiq,line width= 2.3pt,dash pattern=on 1pt off 3pt on 4pt off 3pt ,line join=round] ( 87.19, 73.42) --
( 95.80, 73.42) --
(104.41, 73.42) --
(113.02, 73.42) --
(121.63, 73.42) --
(130.23, 73.42) --
(138.84, 73.42) --
(147.45, 73.42) --
(156.06, 73.42) --
(164.67, 73.42) --
(173.28, 73.42) --
(181.88, 73.42) --
(190.49, 73.42) --
(199.10, 73.42) --
(207.71, 73.42) --
(216.32, 73.42) --
(224.93, 73.42) --
(233.53, 73.42) --
(242.14, 73.42) --
(250.75, 73.42) --
(259.36, 73.42) --
(267.97, 73.42) --
(276.57, 73.42) --
(285.18, 73.42) --
(293.79, 73.42) --
(302.40, 73.42) --
(311.01, 73.42) --
(319.62, 73.42) --
(328.22, 73.42) --
(336.83, 73.42) --
(345.44, 73.42) --
(354.05, 73.42) --
(362.66, 73.42) --
(371.27, 73.42) --
(379.87, 73.42) --
(388.48, 73.57) --
(397.09, 75.03) --
(405.70, 77.99) --
(414.31, 82.44) --
(422.92, 88.61) --
(431.52, 96.96) --
(440.13, 97.82) --
(448.74, 97.96) --
(457.35, 98.01) --
(465.96, 97.81) --
(474.57, 97.58) --
(483.17, 98.23) --
(491.78, 99.65) --
(500.39,101.12);
\end{scope}
\begin{scope}
\path[clip] ( 0.00, 0.00) rectangle (505.89,216.81);
\definecolor{drawColor}{RGB}{0,0,0}

\path[draw=drawColor,line width= 0.6pt,line join=round] ( 78.58, 44.11) --
( 78.58,211.31);
\end{scope}
\begin{scope}
\path[clip] ( 0.00, 0.00) rectangle (505.89,216.81);
\definecolor{drawColor}{gray}{0.30}

\node[text=drawColor,anchor=base east,inner sep=0pt, outer sep=0pt, scale= 1.70] at ( 73.63, 67.57) {1.0};

\node[text=drawColor,anchor=base east,inner sep=0pt, outer sep=0pt, scale= 1.70] at ( 73.63,111.00) {1.2};

\node[text=drawColor,anchor=base east,inner sep=0pt, outer sep=0pt, scale= 1.70] at ( 73.63,154.43) {1.4};

\node[text=drawColor,anchor=base east,inner sep=0pt, outer sep=0pt, scale= 1.70] at ( 73.63,197.86) {1.6};
\end{scope}
\begin{scope}
\path[clip] ( 0.00, 0.00) rectangle (505.89,216.81);
\definecolor{drawColor}{gray}{0.20}

\path[draw=drawColor,line width= 0.6pt,line join=round] ( 75.83, 73.42) --
( 78.58, 73.42);

\path[draw=drawColor,line width= 0.6pt,line join=round] ( 75.83,116.85) --
( 78.58,116.85);

\path[draw=drawColor,line width= 0.6pt,line join=round] ( 75.83,160.28) --
( 78.58,160.28);

\path[draw=drawColor,line width= 0.6pt,line join=round] ( 75.83,203.71) --
( 78.58,203.71);
\end{scope}
\begin{scope}
\path[clip] ( 0.00, 0.00) rectangle (505.89,216.81);
\definecolor{drawColor}{RGB}{0,0,0}

\path[draw=drawColor,line width= 0.6pt,line join=round] ( 78.58, 44.11) --
(500.39, 44.11);
\end{scope}
\begin{scope}
\path[clip] ( 0.00, 0.00) rectangle (505.89,216.81);
\definecolor{drawColor}{gray}{0.20}

\path[draw=drawColor,line width= 0.6pt,line join=round] ( 78.58, 41.36) --
( 78.58, 44.11);

\path[draw=drawColor,line width= 0.6pt,line join=round] (186.19, 41.36) --
(186.19, 44.11);

\path[draw=drawColor,line width= 0.6pt,line join=round] (293.79, 41.36) --
(293.79, 44.11);

\path[draw=drawColor,line width= 0.6pt,line join=round] (401.39, 41.36) --
(401.39, 44.11);
\end{scope}
\begin{scope}
\path[clip] ( 0.00, 0.00) rectangle (505.89,216.81);
\definecolor{drawColor}{gray}{0.30}

\node[text=drawColor,anchor=base,inner sep=0pt, outer sep=0pt, scale= 1.70] at ( 78.58, 27.45) {0.00};

\node[text=drawColor,anchor=base,inner sep=0pt, outer sep=0pt, scale= 1.70] at (186.19, 27.45) {0.25};

\node[text=drawColor,anchor=base,inner sep=0pt, outer sep=0pt, scale= 1.70] at (293.79, 27.45) {0.50};

\node[text=drawColor,anchor=base,inner sep=0pt, outer sep=0pt, scale= 1.70] at (401.39, 27.45) {0.75};
\end{scope}
\begin{scope}
\path[clip] ( 0.00, 0.00) rectangle (505.89,216.81);
\definecolor{drawColor}{RGB}{0,0,0}

\node[text=drawColor,anchor=base,inner sep=0pt, outer sep=0pt, scale= 1.80] at (289.49, 9.00) {$\lambda$};
\end{scope}
\begin{scope}
\path[clip] ( 0.00, 0.00) rectangle (505.89,216.81);
\definecolor{drawColor}{RGB}{0,0,0}

\node[text=drawColor,rotate= 90.00,anchor=base west,inner sep=0pt, outer sep=0pt, scale= 1.80] at ( 17.90, 60.00) {{\footnotesize $\mathbb E[T]$ normalized by that }};

\node[text=drawColor,rotate= 90.00,anchor=base west,inner sep=0pt, outer sep=0pt, scale= 1.80] at ( 42.17, 75.04) {{\footnotesize of $\genseed$}};
\end{scope}

\begin{scope}
\path[clip] (  0.00,  0.00) rectangle (361.35,400.67);

\definecolor{detjiq22}{RGB}{255,0,0}
\path[draw=detjiq22,line width= 2.3pt,dash pattern=on 7pt off 3pt ,line join=round] ( 70+25,170+23) -- ( 100+25,170+23);
\node[anchor=west,inner sep=0pt, outer sep=0pt, scale= 1.70] at (120+15,170+23) {{\footnotesize $\jiq(2,2)$}};

\definecolor{detjsq22}{RGB}{0,0,255}
\path[draw=detjsq22,line width= 1.5pt,line join=round] ( 70+25,150+23) -- ( 100+25,150+23);
\node[anchor=west,inner sep=0pt, outer sep=0pt, scale= 1.70] at (120+15,150+23) {{\footnotesize $\jsq(2,2)$}};

\definecolor{detjiq}{RGB}{60,180,60}
\path[draw=detjiq,line width= 2.3pt,dash pattern=on 1pt off 3pt on 4pt off 3pt ,line join=round] ( 70+25,130+23) -- ( 100+25,130+23);
\node[anchor=west,inner sep=0pt, outer sep=0pt, scale= 1.70] at (120+15,130+23) {{\footnotesize $\ipoptd{\detq}$}};

\end{scope}
\end{tikzpicture}}
    \caption{$\mathbb E[T]$ relative to that of $\genseed$ (i.e., $\mathbb E[T]^{\mathsf{DP}}/\mathbb E[T]^{\genseed}$) as a function of $\lambda$ for the parameter settings where $s=2$ and $d=4$, $\lambda$ varies over $\{0.02,0.04,\ldots,0.98\}$,  $(q_1,q_2)=(4/5,1/5)$, and $R_1=5$, yielding $(\mu_1,\mu_2)=(25/21,5/21)$, for the dispatching policies $\mathsf{DP}\in\{\jiq(2,2),\jsq(2,2),\ipoptd\detq\}$.}
    \label{fig:peva}
\end{figure}
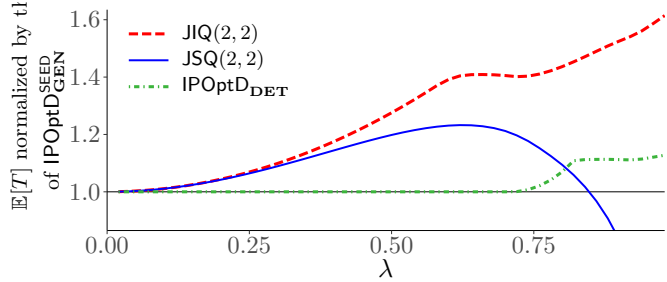

Finally, we study $\detq$ in the case where $s=2$ and $d=4$ by comparing the performance of $\ipoptd\detq$ to that of the $\mathsf{JIQ}(2,2)$\label{jiqdfds} and $\mathsf{JSQ}(2,2)$\label{jsqdfds} dispatching policies studied in~\cite{gardner2020scalable}.
We consider $\mathbb E[T]$ under these three policies, normalized to that under $\genseed$, for all 12 combinations of $R_1$ and $(q_1,q_2)$ that are studied in Section 5.2 of \cite{gardner2020scalable}.
All of these 12 combinations yield similar insights; Figure~\ref{fig:peva} shows results for the particular setting in which $R_1 = 5$ and $(q_1,q_2)=(4/5,1/5)$ (note that this parameterization is not a member of the space discussed in Section~\ref{sec:parameters}, rather it is taken from \cite{gardner2020scalable}).
The $\mathsf{JIQ}(2,2)$ dispatching policy often performs considerably worse than $\ipoptd\detq$, the performance of which is indistinguishable from $\genseed$ except at the highest load values.
However, despite its generally strong performance, there is no advantage to using the $\ipoptd\detq$ policy instead of $\genseed$, as $\ipoptd\detq$ features a substantially higher mean runtime: across the parameter settings shown in Figure~\ref{fig:peva} (respectively, all of the parameter settings considered in Section 5.2 of \cite{gardner2020scalable}), the mean runtime of $\ipoptd\detq$ is more than 95\% higher (respectively, more than 20\% higher) than that of $\genseed$.
This high runtime is likely due to the fact that $\ipoptd\detq$ must solve six smaller subproblems, whereas $\genseed$ solves only two (larger) optimization problems. Meanwhile, the queue length-aware $\mathsf{JSQ}(2,2)$ policy tracks $\mathsf{JIQ}(2,2)$ at low load, but considerably outperforms the other policies, including $\genseed$, as load approaches 1.  This suggests that there is considerable value in investigating $\clda$-driven dispatching policies, which we explore in the next section.

\begin{remark}
The $\mathsf{JIQ}(2,2)$ dispatching policy from~\cite{gardner2020scalable} would be denoted in our framework as $\ipoptd{\mathsf{DQR}_{(2,2)}}\in \DP{\detq}{\cida}$, recalling that $\mathsf{DQR}_{\dvec}$denotes the querying rule that always queries so that $\Dvec=\dvec$. Note that while  $\jsq(2,2)\in\DP{\mathsf{DQR}_{(2,2)}}{\clda\backslash(\cida\cup\lda)}$ fits within our framework, $\jsq(2,2)$ does \emph{not} use what our framework would describe as $\jsq$ assignment, i.e., $\jsq(2,2)\neq\DP{\mathsf{DQR}_{(2,2)}}{\jsq}$.  Specifically, $\jsq(2,2)$ is a variant of $\jiq(2,2)$ that, given a set of queried servers, assigns an incoming job to the same \emph{class} that the job would be assigned to under $\jiq(2,2)$.  However, while $\jiq(2,2)$ ultimately assigns the job to a server chosen uniformly at random among those queried from the selected class,  $\jsq(2,2)$ assigns the incoming job to a server chosen uniformly at random among those queried servers \emph{with the shortest queue(s)} from the selected class.
\end{remark}

\section{The $\clda$ family of assignment rules}
\label{sec:construction}
In this section, we discuss assignment rules in $\clda$ beyond those in the $\cida$ family.  After presenting a general structure for $\clda$ assignment rules (Section \ref{sec:clda-formal}) and discussing the difficulty of analyzing dispatching policies using such assignment rules, we turn our attention to the development of heuristic $\clda$-driven dispatching policies (Section \ref{sec:clda-heuristic}).  Simulations suggest that our heuristic policies perform favorably relative to existing dispatching policies presented in the literature (Section \ref{sec:clda-simulation}).  These heuristic policies allow for length-aware assignment while leveraging our analysis of querying rules under $\cida$ (length-blind) assignment, as presented in the preceding sections.

\subsection{Formal presentation of the $\clda$ family of assignment rules}
\label{sec:clda-formal}

In this section, we present a generalization of the $\cida$ family of assignment rules to account for queue lengths (rather than just idle/busy statuses), resulting in the $\clda$ family of assignment rules. This family encompasses all static and symmetric assignment rules that can observe both the class (i.e., speed) of, and queue length at, each of the queried servers in assigning a newly arrived job.  Throughout this subsection, we introduce a variety of new notation.  To aid the reader, this new notation is summarized in Table \ref{tbl:in-text-notation}.

{
\begin{table}[h]
\begin{tabular}{>{\raggedright}p{0.13\textwidth}@{}>{\centering}p{0.05\textwidth}@{}p{0.7\textwidth}>{\raggedleft\arraybackslash}p{0.05\textwidth}}
    \toprule
    $\Ain{i}{n}$ r.b. $\ain{i}{n}$ &$\equiv$& Number of queried class-$i$ servers with a queue length of $n$ \dotfill & p. \pageref{Avecn}\\
    $\Avecn{n}$ r.b. $\avecn{n}$ &$\equiv$& $\{\Ain{0}{n}, \Ain{1}{n}, \ldots, \Ain{s}{n}\}$ r.b. $\{\ain{0}{n}, \ain{1}{n}, \ldots, \ain{s}{n}\}$ \dotfill & p. \pageref{Avecn}\\
    $\Avecv$ r.b. $\avecv$ &$\equiv$& $\{\Avecn0,\Avecn1,\ldots\}$ r.b. $\{\avecn0,\avecn1,\ldots\}$  \dotfill & p. \pageref{Aarrow}\\
    $\Acalv$ &$\equiv$& $\left\{\avecv\colon\sum_{i=1}^s\sum_{n=0}^\infty \ain{i}n=d\right\}$; Set of all possible realizations $\Avecv$\dotfill & p. \pageref{Acalv}\\
    $\alpha_i^{(n)}(\avecv)$ &$\equiv$& Probability that the job is assigned to a class-$i$ server with $n$ jobs when $\Avecv\equiv\avecv$ \dotfill & p. \pageref{alphana}\\
    $\alpha_i(\avecv)$ &$\equiv$& Probability that the job is assigned to a class-$i$ server with $n_i^*$ jobs when $\Avecv\equiv\avecv$\dotfill  & p. \pageref{alphai}\\
    $n_i^*$ &$\equiv$& $\min\left\{m\in\mathbb N\colon\ain{i}{m}>0\right\}$; Queue length of the shortest queue among queried class-$i$ servers \dotfill & p. \pageref{nistar}\\
    \bottomrule
\end{tabular}\\
\caption{Table of Notation for Section~\ref{sec:clda-formal} (``r.b.'' stands for ``realized by'')\label{tbl:in-text-notation}}
\end{table}
}

\begin{remark}
When we refer to the ``queue length'' of/at a server, we mean the number of jobs occupying  that server's subsystem: this includes all jobs currently being served by the server and all those waiting for service.
\end{remark}

Recall that our study of the $\cida$ family of assignment rules motivated us to encode the idle/busy statuses of the queried servers by the random vector $\Avec$, which takes on realizations of the form $\avec\equiv(a_1,\ldots,a_s)\in\Acal\equiv\{\avec\colon a_1+\cdots+a_s\le d\}$, where $a_i$ is the number of \emph{idle} class-$i$ servers among the $d_i$ queried. Analogously, in studying the $\clda$ family it will be helpful to encode the number of class-$i$ servers of \emph{each possible queue length} among the $d_i$ queried.  To this end, for each $n\in\mathbb N\equiv\{0,1,\ldots\}$, let $\Avecn{n}\equiv\left(\Ain1n,\Ain2n,\ldots,\Ain{s}n\right)$\label{Avecn} be a random vector taking on realizations of the form $\avecn{n}\equiv\left(\ain1n,\ain2n,\ldots,\ain{s}n\right)\in\Acal$ where $\Ain{i}n$\label{Ain} (respectively $\ain{i}{n}$\label{ain}) is the random variable (respectively the realization of the random variable) giving the number of queried class-$i$ servers with a queue length of $n$.  Three observations follow immediately from these definitions: (i) $\Avecn0=\Avec$, (ii) $\Avecn{n}\le\Dvec$ (element-wise) for all $n\in\mathbb N$, and (iii) $\sum_{n=0}^\infty \Avecn{n}=\Dvec$.

Now let $\vec{\Avec}\equiv\{\Avecn0,\Avecn1,\ldots\}$\label{Aarrow}, denote the realizations of this random object by $\avecv\equiv\{\avecn0,\avecn1,\ldots\}$, and denote the set of all such realizations by $\Acalv\equiv\left\{\avecv\colon\sum_{i=1}^s\sum_{n=0}^\infty \ain{i}n=d\right\}$.\label{Acalv}  Each realized aggregate query state can now be fully described by some $\avecv\in\Avecv$, allowing us to treat $\dvec$ as a derived quantity: $\dvec=\sum_{n=0}^\infty\avecn{n}$.

Formally, a $\clda$ assignment rule is uniquely specified by a family of functions $\alpha_i^{(n)}\colon\Acalv\to[0,1]$\label{alphana} parameterized by $(i,n)\in\Scal\times\mathbb N$, where $\alpha_i^{(n)}(\avecv)$ denotes the probability that a job seeing a query with aggregate state $\Avecv=\avecv$ is assigned to a class-$i$ server with a queue length of $n$.  Clearly, we must have $\alpha_i^{(n)}(\avecv)=0$ whenever $\ain{i}n=0$ and $\sum_{i=1}^s\sum_{n=0}^\infty\alpha_i^{(n)}(\avecv)=1$ for all $\avecv\in\Acalv$.

As an illustrative example, let us see how an assignment rule that opts to \emph{ignore} queue length information (apart from idleness information) can be implemented via such a family of functions.  Specifically, let us consider an assignment rule from the $\cida$ family (noting that such an assignment rule is also a member of the $\clda$ family, as $\cida\subseteq\clda$), defined (as in Section~\ref{sec:cida-formal}) by some family of functions $\alpha_i\colon\Scalb\times\Dcal\to[0,1]$ parameterized by $i\in\Scal$, where $\alpha_i(j,\dvec)$ denotes the probability that a job is assigned to a queried class-$i$ server, given that it sees $J=j$ as the class of the fastest idle server and a query mix $\Dvec=\dvec$.  In this case, letting $j\equiv\min\{\ell\colon\ain{\ell}{0}>0\}$ (where we again use the convention that $\min\emptyset\equiv s+1$) we can define $\alpha_i^{(n)}(\avecv)$ in terms of $\alpha_i(j,\dvec)$ as follows:
\begin{align}
    \alpha_i^{(0)}(\avecv)&=\begin{cases}\alpha_i(j,\dvec)&\mbox{if }i=j\\ 0&\mbox{otherwise}\end{cases}\label{eq:alpha-0}\\
    \alpha_i^{(n)}(\avecv)&=\begin{cases}\displaystyle{\frac{\ain{i}{n}\alpha_i(j,d)}{\sum_{m=1}^\infty\ain{i}{m}}}&\mbox{if }i
    \le j\\0&\mbox{otherwise}\end{cases}&(\forall n\ge1).\label{eq:alpha-n}
\end{align}
Equation~\eqref{eq:alpha-0} gives the probability of assigning the job to an idle class-$i$ server, and Equation~\eqref{eq:alpha-n} gives the probability of assigning the job to a busy class-$i$ server with $n$ jobs in its queue.  While the first equation is straightforward, the latter becomes clear by making the following observation: if we ignore queue lengths beyond idle/busy statuses, once we have chosen to assign the job to a busy class-$i$ server, we choose one such server at random, and hence, the job is sent to a class-$i$ server with a queue length of $n$ with probability $\ain{i}{n}\left/\sum_{m=1}^\infty\ain{i}m\right.$.

Now observe that if one opts to make full use of queue length information, then whenever one assigns a job to a class-$i$ server it is naturally favorable to assign the job to the class-$i$ server with the shortest queue among those queried. If we prune the space of $\clda$ assignment rules in this fashion (which would leave out the $\cida$ assignment rules), then we can instead uniquely specify assignment rules by a family of functions $\alpha_i\colon\Acalv\to[0,1]$ that are parameterized only by $i\in\Scal$ (rather than also being parameterized by $n\in\mathbb N$).  In this case, letting $n_i^*\equiv\min\left\{m\in\mathbb N\colon\ain{i}{m}>0\right\}$\label{nistar} (i.e., letting $n_i^*$ be the queue length of the shortest queue among queried class-$i$ servers) for each $i\in\Scal$ (with $n_i^*\equiv\infty$ whenever $d_i=0$), we can express (the original) $\alpha_i^{(n)}(\avecv)$ in terms of (the new) $\alpha_i(\avecv)$ as follows:\label{alphai}
\[
\alpha_i^{(n)}(\avecv)=\begin{cases}\alpha_i(\avecv)&\mbox{if }n=n_i^*\\
0&\mbox{otherwise}
\end{cases}.
\]

\subsection{Examples of $\clda$ assignment rules}
\label{sec:clda-ex}

Two examples of assignment rules in $\clda\backslash\cida$ that we can specify using families of functions $\alpha_i\colon\Acalv\to[0,1]$ (where we again use $n_i^*$ to denote the queue length of the shortest queue among queried class-$i$ servers) include $\jsq$ and $\sed$.

\begin{remark}
When we refer to $\jsq$ and $\sed$, we are referring to just the \emph{assignment rules}, rather than the traditional $\jsq$ and $\sed$ \emph{dispatching policies} studied in the literature in small scale settings, or the $\jsq$-$d$ and $\sed$-$d$ \emph{dispatching policies}, which in our framework are referred to as the $\DP{\uniq}{\jsq}$ and $\DP{\uniq}{\sed}$ dispatching policies respectively.  Moreover, note that $\jsq$ is actually a member of the $\lda$ family---a subfamily of $\clda$ that allows for leveraging queue length information, but is blind to server classes (i.e., speeds).
\end{remark}

We discuss these two rules (i.e., $\jsq$ and $\sed$) in greater detail---together with a third assignment rule in $\clda\backslash\cida$, \textsf{Shortest Expected Wait} ($\sew$), which we introduce here---below:
\begin{itemize}
    \item \textsf{Join the Shortest Queue} ($\jsq$)\label{jsq} is an individual assignment rule that is a member of the $\lda$ family (and therefore also the $\clda$ family) that assigns the job to a queried server (chosen uniformly at random) among those with the shortest queue (regardless of their class).  It is specified by 
    \[\alpha_i(\avecv)=\frac{\ain{i}{n_i^*} \prod_{\ell=1}^{s} I\{n_{i}^* \leq n_{\ell}^*\}}{\sum_{i'=1}^s\ain{i'}{n_{i'}^*}\prod_{\ell=1}^{s}I\left\{ n_{i'}^* \leq n_{\ell}^* \right\}}.\]
    \item \textsf{Shortest Expected Delay} ($\sed$)\label{sed} is an individual assignment rule that is a member of the $\clda$ family that assigns the job to a queried server (chosen uniformly at random) among those on which the job would complete soonest in expectation \emph{under the assumption of \emph{\textsf{First Come First Serve} ($\fcfs$)} scheduling} (regardless of the actual scheduling rule being implemented).  By observing that the expected delay experienced by a job (under $\fcfs$ scheduling) that is assigned to a class-$i$ server with $n$ other jobs already in its queue is $(n+1)/\mu_i$, we find that the $\sed$ assignment rule is specified by
    \[\alpha_i(\avecv)=\frac{\ain{i}{n_i^*}\prod_{\ell=1}^{s}I\left\{ \frac{n_i^*+1}{\mu_i} \leq \frac{n_{\ell}^*+1}{\mu_{\ell}} \right\}}{\sum_{i'=1}^s \ain{i'}{n_{i'}^*} \prod_{\ell=1}^{s}I\left\{\frac{n_{i'}^*+1}{\mu_{i'}} \leq \frac{n_{\ell}^*+1}{\mu_{\ell}}\right\}}.\]
    \item \textsf{Shortest Expected Wait} ($\sew$)\label{sew} is an individual assignment rule that is a member of the $\clda$ family that assigns the job to a queried server (chosen uniformly at random) among those on which the job would \emph{enter service} soonest in expectation \emph{under the assumption of \emph{\textsf{First Come First Serve} ($\fcfs$)} scheduling} (regardless of the actual scheduling rule being implemented).  Unlike $\sed$, $\sew$ does not account for the expected size of the arriving job, $1/\mu_i$.  By observing that the expected waiting time until entering service experienced by a job (under $\fcfs$ scheduling) that is assigned to a class-$i$ server with $n$ other jobs already in its queue is $n/\mu_i$, we find that the $\sew$ assignment rule is specified by
    \[\alpha_i(\avecv)=\frac{\ain{i}{n_i^*}\prod_{\ell=1}^{s}I\left\{ \frac{n_i^*}{\mu_i} \leq \frac{n_{\ell}^*}{\mu_{\ell}} \right\}}{\sum_{i'=1}^s \ain{i'}{n_{i'}^*} \prod_{\ell=1}^{s}I\left\{\frac{n_{i'}^*}{\mu_{i'}} \leq \frac{n_{\ell}^*}{\mu_{\ell}}\right\}}.\]
\end{itemize}
\begin{remark}
Our nomenclature is perhaps imperfect, as \emph{delay} is sometimes used to refer to time in queue, but here we are using \emph{delay} (in the name of $\sed$---which we inherit from the literature) to refer to the \emph{time in system} and \emph{wait} (in the name of $\sew$) to refer to the \emph{time in queue}.
\end{remark}

We also introduce a variant of each of the above policy that breaks ``ties'' in favor of faster servers: (i) $\jsq^\star$, (ii) $\sed^\star$, and (iii) $\sew^\star$ act like the (i) $\jsq$, (ii) $\sed$, and (iii) $\sew$ assignment rules, except jobs are always assigned to one the \emph{fastest} servers (chosen uniformly at random) among those queried servers that have the (i) shortest queue, (ii) the shortest expected delay (i.e., at which the job would experience the shortest expected response time), and (iii) the shortest wait (i.e., at which the job would experience the shortest time in queue), respectively.  Note that while $\jsq\in\lda$, $\jsq^\star\in\clda\backslash\lda$, as $\jsq^\star$ makes use of class information in breaking ties between queried servers with the same queue length.  These rules are specified by the following:\label{jsqstar}\label{sedstar}\label{sewstar}
\begin{align*}
      \alpha_i(\avecv)
    &=\prod_{\ell=1}^{i-1}I\{n_i^*<n_\ell^*\}\prod_{\ell=i+1}^s I\{n_i^*\le n_\ell\}.&(\jsq^\star)\\
    \alpha_i(\avecv)&=\prod_{\ell=1}^{i-1}I\left\{ \frac{n_i^*+1}{\mu_i} < \frac{n_{\ell}^*+1}{\mu_{\ell}}\right\}\prod_{\ell=i+1}^{s}I\left\{ \frac{n_i^*+1}{\mu_i} \le \frac{n_{\ell}^*+1}{\mu_{\ell}}    \right\}&(\sed^\star)\\
    \alpha_i(\avec)&=\prod_{\ell=1}^{i-1}I\left\{ \frac{n_i^*}{\mu_i} < \frac{n_{\ell}^*}{\mu_{\ell}}\right\}\prod_{\ell=i+1}^{s}I\left\{ \frac{n_i^*}{\mu_i} \le \frac{n_{\ell}^*}{\mu_{\ell}} \right\} &(\sew^\star).
\end{align*}

\subsection{A heuristic for finding strong dispatching policies using $\clda$ assignment}
\label{sec:clda-heuristic}

The analysis of general assignment rules in the $\clda$ family introduces intractability issues that we were able to avoid in our analysis of the $\cida$ family of assignment rules.  There are two key challenges for identifying strong dispatching policies with assignment rules in $\clda\backslash\cida$.  First, while the $\alpha_i$ functions designating the $\cida$ policies had a finite domain ($\Acal\times\Dcal$, and after subsequent pruning $\Scalb\times\Dcal$), those functions specifying assignment rules for $\clda$ policies---even with the pruning introduced in Section~\ref{sec:clda-formal}---have an infinite domain ($\Acalv$).  Hence, the $\clda$ assignment rules span an infinite dimensional space, unlike the finite-dimensional polytopes spanned by their $\cida$ counterparts (see Appendix D for details); the former generally precludes straightforward optimization, while the latter facilitates it.

The second challenge associated with identifying strong dispatching policies with assignment rules that take queue lengths into account is the lack of exact performance analysis for most dispatching policies in $\clda\backslash\cida$. Thus, even if we could solve an infinite-dimensional optimization problem (i.e., even if we could overcome the first challenge), it is impossible to formulate the objective function for such an optimization problem.  

We attempt to jointly overcome these challenges by populating a roster of heuristic dispatching policies designed based on the $\cida$-driven policies of Sections~\ref{sec:model}--\ref{sec:numerical}.
In the next subsection, we show (via simulation) that many of these policies perform well relative to the aforementioned $\cida$-driven policies.

We address the first challenge (i.e., the infinite dimensional space spanned by the $\clda$ assignment rules) by limiting ourselves to the example assignment rules discussed in Section~\ref{sec:clda-ex}: $\jsq$, $\sed$, $\sew$, $\jsq^\star$, $\sed^\star$, and $\sew^\star$. Note that these are \emph{individual} assignment rules, rather than assignment rule \emph{families}, which obviates the need for optimizing continuous probabilistic parameters.

One hopes that even without sophisticated fine-tuned probabilistic parameters, the greedy $\sed$ and $\sew$ assignment rules (with or without class-based tie-breaking) still manage to yield stronger performance than the length-blind $\cida$-driven dispatching policies---at least when paired with a judiciously chosen querying rule.  Meanwhile, by studying $\jsq$, we can assess the extent to which queue-length information can lead to strong performance even in the absence of heterogeneity-awareness in the assignment decision.

The second challenge (i.e., the lack of performance analysis as a basis for optimization), then, reduces to the problem of choosing a querying rule to use in conjunction with our six chosen assignment rules. We propose three ideas for choosing an appropriate querying rule---ultimately, each approach will add additional dispatching policies to our roster.

The first idea for choosing a querying rule is to use the same approach that we are taking on the assignment side.  That is, we can limit ourselves to one (or some small number of) individual querying rule(s).  Of the two specific individual querying rules discussed in this paper, $\uniq$ does not guarantee stability, while $\brq$ does (see Section~\ref{sec:stability} for details).  For this reason, we add the following six dispatching policies to our roster: $\DP\brq\jsq$, $\DP\brq\sed$, $\DP\brq\sew$, $\DP\brq{\jsq^\star}$, $\DP\brq{\sed^\star}$, and $\DP\brq{\sew^\star}$.

The remaining two ideas involve leveraging the diversity of querying rules available within the families studied throughout this paper, as it is unnecessarily restrictive to only consider dispatching policies that involve no optimization (i.e., that involve combining a specific individual querying rule with a specific individual assignment rule, as above).  The broadest querying rule family that we have studied is, of course, $\genq$; unfortunately, analyzing exact mean response times under, e.g.,  $\DP{\genq}{\jsq}$ and $\DP{\genq}{\sed}$ appears to be intractable. 

Our second idea presents one way to overcome this tractability limitation: we restrict attention to the $\srcq$ family of querying rules, where one selects a class at random (according to some fixed distribution) upon the arrival of each job and then queries $d$ servers of that class.  $\srcq$ querying eliminates the possibility of needing to make assignment decisions between servers running at different speeds, meaning that pairing $\srcq$ with any of our six individual assignment rules yields the same dispatching policy; we will refer to this single policy as $\DP\srcq\jsq$.
Furthermore, because assignment decisions are always made among servers of the same speed, the analysis of $\DP{\srcq}{\jsq}$ reduces to that of $s$ independent homogeneous systems under $\jsq$. This exact analysis allows us to use IPOPT to find the $\ipoptd{\DP{\srcq}{\jsq}}$ dispatching policy, which we add to our roster of dispatching policies.

\begin{remark}
 As noted above, the analysis of $\DP{\srcq}{\jsq}$ reduces to that of $s$ independent homogeneous systems under the $\DP{\uniq}{\jsq}$ dispatching policy (referred to in the literature as $\jsq$-$d$). The mean response time in such homogeneous systems was analyzed exactly in \cite{mitzenmacher2001power,vvedenskaya1996queueing}. We then rely on IPOPT to determine the ``optimal'' $\hat{p}(i)$ parameters for $i \in \Scal$ (i.e., the probability of querying each single class $i$, see Section~\ref{sec:opt-src}). We further note that $\DP{\srcq}{\jsq}$ was  previously studied in the case of $s=2$, under the \textsf{Processor Sharing} ($\mathsf{PS}$) scheduling discipline, in \cite{mukhopadhyay2016analysis}.
\end{remark}

Our third idea is to use a novel heuristic that leverages our previous study of $\DP\qrf\cida$ dispatching policies from Sections~\ref{sec:model}--\ref{sec:numerical}. Our heuristic constructs a dispatching policy by combining a individual querying rule found by IPOPT and any one of our six individual assignment rules.  Specifically, the heuristic uses the $\genseedq$ querying rule (i.e., the querying rule yielded by the IPOPT solution to the optimization problem associated with $\DP\genq\cida$, seeded with the IPOPT solution for $\DP\indq\cida$).
Note that our choice of querying rule (i.e., $\genseedq$) is not contingent on the choice of assignment rule, as tractability necessitates foregoing any kind of joint optimization.
To this end we complete our roster with the following six policies: $\DP{\genseedq}{\jsq}$, $\DP{\genseedq}{\sed}$, $\DP{\genseedq}{\sew}$,  $\DP{\genseedq}{\jsq^\star}$, $\DP{\genseedq}{\sed^\star}$, and $\DP{\genseedq}{\sew^\star}$.

\subsection{Simulation-driven performance evaluation of dispatching policies using $\clda$ assignment}
\label{sec:clda-simulation}

We simulate the $\DP\brq\jsq$, $\DP\brq\sed$, $\DP\genseedq\jsq$, and $\DP\genseedq\sed$ dispatching policies in a system with $k=3000$ servers under under the same collection of parameter settings studied in Figure~\ref{fig:lambda-plot} in Section~\ref{sec:graphs}.  We simulate 10\,000\,000 arrivals to the system and record the observed response time for each.  We then average these values (discarding the first 1\,000\,000 to allow the system to ``reach a steady state'' where the running average response time was observed to stabilize) to obtain a $\mathbb E[T]$ value under each policy at each value of $\lambda$. We omit results for $\lambda\in\{0.92,0.94,0.96,0.98\}$ as the observed variance of response times across successive runs exceeded 1\% of the mean in these cases. Running longer simulations with more arrivals could reduce the variance in these cases, but doing so would have been prohibitively expensive in terms of the simulation runtime.

\begin{figure}
    \centering
    \hspace{-0.85cm}
\resizebox{\textwidth}{!}{
\begin{tabular}{ccc}
\begin{tikzpicture}[x=1pt,y=1pt]


\definecolor{fillColor}{RGB}{255,255,255}
\path[use as bounding box,fill=fillColor,fill opacity=0.00] (0,0) rectangle (\legendxpos-100,\legendypos);
\begin{scope}

\path[brjsqstyle] ( 42.53,168.55) --
	( 46.33,168.44) --
	( 50.14,168.59) --
	( 53.94,168.51) --
	( 57.75,168.24) --
	( 61.55,167.78) --
	( 65.36,167.37) --
	( 69.16,166.73) --
	( 72.97,165.99) --
	( 76.77,164.98) --
	( 80.58,163.90) --
	( 84.39,162.19) --
	( 88.19,160.64) --
	( 92.00,158.58) --
	( 95.80,156.10) --
	( 99.61,153.41) --
	(103.41,150.09) --
	(107.22,146.46) --
	(111.02,142.15) --
	(114.83,138.51) --
	(118.63,134.79) --
	(122.44,131.34) --
	(126.24,127.88) --
	(130.05,124.30) --
	(133.85,120.50) --
	(137.66,116.61) --
	(141.46,112.28) --
	(145.27,107.66) --
	(149.07,103.71) --
	(152.88,100.23) --
	(156.68, 97.41) --
	(160.49, 94.76) --
	(164.29, 92.51) --
	(168.10, 90.26) --
	(171.91, 88.19) --
	(175.71, 85.76) --
	(179.52, 82.83) --
	(183.32, 79.90) --
	(187.13, 76.49) --
	(190.93, 72.75) --
	(194.74, 68.11) --
	(198.54, 63.12) --
	(202.35, 57.41) --
	(206.15, 51.11) --
	(209.96, 44.05);

\path[brjsqstarstyle] ( 42.53, 83.46) --
	( 46.33, 85.16) --
	( 50.14, 87.02) --
	( 53.94, 89.07) --
	( 57.75, 91.16) --
	( 61.55, 93.31) --
	( 65.36, 95.63) --
	( 69.16, 97.63) --
	( 72.97, 99.55) --
	( 76.77,101.60) --
	( 80.58,103.42) --
	( 84.39,104.83) --
	( 88.19,106.19) --
	( 92.00,107.23) --
	( 95.80,107.51) --
	( 99.61,107.60) --
	(103.41,107.18) --
	(107.22,106.47) --
	(111.02,104.95) --
	(114.83,103.41) --
	(118.63,102.01) --
	(122.44,100.42) --
	(126.24, 98.79) --
	(130.05, 96.87) --
	(133.85, 94.57) --
	(137.66, 92.14) --
	(141.46, 89.15) --
	(145.27, 85.95) --
	(149.07, 83.17) --
	(152.88, 80.88) --
	(156.68, 78.93) --
	(160.49, 77.46) --
	(164.29, 75.95) --
	(168.10, 74.56) --
	(171.91, 73.30) --
	(175.71, 71.68) --
	(179.52, 69.90) --
	(183.32, 67.81) --
	(187.13, 65.35) --
	(190.93, 62.44) --
	(194.74, 58.98) --
	(198.54, 54.63) --
	(202.35, 49.62) --
	(206.15, 44.44) --
	(209.96, 38.59);

\path[genjsqstyle] ( 42.53, 74.75) --
	( 46.33, 74.80) --
	( 50.14, 74.85) --
	( 53.94, 74.74) --
	( 57.75, 74.77) --
	( 61.55, 74.79) --
	( 65.36, 74.81) --
	( 69.16, 74.84) --
	( 72.97, 74.77) --
	( 76.77, 74.75) --
	( 80.58, 74.66) --
	( 84.39, 74.53) --
	( 88.19, 74.45) --
	( 92.00, 74.27) --
	( 95.80, 74.11) --
	( 99.61, 73.72) --
	(103.41, 73.10) --
	(107.22, 72.54) --
	(111.02, 73.22) --
	(114.83, 74.67) --
	(118.63, 75.78) --
	(122.44, 76.37) --
	(126.24, 76.80) --
	(130.05, 76.69) --
	(133.85, 76.36) --
	(137.66, 75.84) --
	(141.46, 74.79) --
	(145.27, 75.31) --
	(149.07, 77.00) --
	(152.88, 78.67) --
	(156.68, 80.27) --
	(160.49, 81.92) --
	(164.29, 83.07) --
	(168.10, 84.10) --
	(171.91, 84.88) --
	(175.71, 83.19) --
	(179.52, 81.77) --
	(183.32, 80.75) --
	(187.13, 79.12) --
	(190.93, 76.79) --
	(194.74, 73.27) --
	(198.54, 67.38) --
	(202.35, 63.19) --
	(206.15, 56.66);

\path[genjsqstarstyle] ( 42.53, 74.82) --
	( 46.33, 74.79) --
	( 50.14, 74.79) --
	( 53.94, 74.77) --
	( 57.75, 74.79) --
	( 61.55, 74.80) --
	( 65.36, 74.80) --
	( 69.16, 74.77) --
	( 72.97, 74.79) --
	( 76.77, 74.72) --
	( 80.58, 74.59) --
	( 84.39, 74.57) --
	( 88.19, 74.42) --
	( 92.00, 74.30) --
	( 95.80, 73.98) --
	( 99.61, 73.62) --
	(103.41, 73.17) --
	(107.22, 72.50) --
	(111.02, 72.27) --
	(114.83, 72.02) --
	(118.63, 71.90) --
	(122.44, 71.60) --
	(126.24, 71.16) --
	(130.05, 70.61) --
	(133.85, 69.96) --
	(137.66, 69.00) --
	(141.46, 67.92) --
	(145.27, 67.81) --
	(149.07, 68.35) --
	(152.88, 68.81) --
	(156.68, 69.26) --
	(160.49, 69.87) --
	(164.29, 70.11) --
	(168.10, 70.48) --
	(171.91, 70.76) --
	(175.71, 69.71) --
	(179.52, 68.84) --
	(183.32, 67.97) --
	(187.13, 66.66) --
	(190.93, 64.76) --
	(194.74, 61.61) --
	(198.54, 57.13) --
	(202.35, 53.02) --
	(206.15, 47.79);
\end{scope}
\begin{scope}
\path[draw=gray,line width=0.4*\lthickness pt,dash pattern=] (38.72, 30.69) -- (38.72, 175);
\node[text=black,anchor=base east,inner sep=0pt, outer sep=0pt, scale= \labscal] at ( 32, 71.75) {1.0};

\node[text=black,anchor=base east,inner sep=0pt, outer sep=0pt, scale= \labscal] at ( 32,118.66) {1.5};

\node[text=black,anchor=base east,inner sep=0pt, outer sep=0pt, scale= \labscal] at ( 32,165.58) {2.0};
\end{scope}
\begin{scope}
\path[draw=black,line width= 0.6*\lthickness pt,line join=round] ( 35, 74.78) --
	( 38.72, 74.78);

\path[draw=black,line width= 0.6*\lthickness pt,line join=round] ( 35,121.70) --
	( 38.72,121.70);

\path[draw=black,line width= 0.6*\lthickness pt,line join=round] ( 35,168.61) --
	( 38.72,168.61);
\end{scope}
\path[draw=gray,line width=0.4*\lthickness pt,dash pattern=] (34.16, 74.78) -- (230, 74.78);
\begin{scope}
\path[draw=black,line width= 0.6*\lthickness pt,line join=round] ( 38.72, 27.94) --
	( 38.72, 30.69);

\path[draw=black,line width= 0.6*\lthickness pt,line join=round] ( 86.29, 27.94) --
	( 86.29, 30.69);

\path[draw=black,line width= 0.6*\lthickness pt,line join=round] (133.85, 27.94) --
	(133.85, 30.69);

\path[draw=black,line width= 0.6*\lthickness pt,line join=round] (181.42, 27.94) --
	(181.42, 30.69);
\end{scope}
\begin{scope}
\node[text=black,anchor=base,inner sep=0pt, outer sep=0pt, scale= \labscal] at ( 38.72, 19.68) {0.00};

\node[text=black,anchor=base,inner sep=0pt, outer sep=0pt, scale= \labscal] at ( 86.29, 19.68) {0.25};

\node[text=black,anchor=base,inner sep=0pt, outer sep=0pt, scale= \labscal] at (133.85, 19.68) {0.50};

\node[text=black,anchor=base,inner sep=0pt, outer sep=0pt, scale= \labscal] at (181.42, 19.68) {0.75};
\end{scope}
\begin{scope}
\node[text=black,anchor=base,inner sep=0pt, outer sep=0pt, scale=  1.5*\labscal] at (110.07,  180) {$\jsq$ Policies};
\node[text=black,anchor=base,inner sep=0pt, outer sep=0pt, scale=  \labscal] at (110.07,  8.22) {$\lambda$};
\node[text=black,rotate= 90.00,anchor=base,inner sep=0pt, outer sep=0pt, scale= \labscal] at ( 12.0,0.5*\legendypos) {Normalized $\mathbb E[T]$};
\end{scope}

\end{tikzpicture} & 
\begin{tikzpicture}[x=1pt,y=1pt]


\definecolor{fillColor}{RGB}{255,255,255}
\path[use as bounding box,fill=fillColor,fill opacity=0.00] (0,0) rectangle (\legendxpos-100,\legendypos);
\begin{scope}

\path[brsewstyle] ( 42.39,168.51) --
	( 46.14,168.50) --
	( 49.88,168.50) --
	( 53.63,168.20) --
	( 57.37,168.10) --
	( 61.12,167.59) --
	( 64.86,167.03) --
	( 68.61,166.38) --
	( 72.35,165.29) --
	( 76.10,164.02) --
	( 79.84,162.57) --
	( 83.59,160.91) --
	( 87.33,158.85) --
	( 91.07,156.42) --
	( 94.82,153.52) --
	( 98.56,150.18) --
	(102.31,146.80) --
	(106.05,142.77) --
	(109.80,138.20) --
	(113.54,134.01) --
	(117.29,129.92) --
	(121.03,125.87) --
	(124.78,121.93) --
	(128.52,117.97) --
	(132.27,113.85) --
	(136.01,109.50) --
	(139.75,104.78) --
	(143.50,100.05) --
	(147.24, 95.83) --
	(150.99, 92.25) --
	(154.73, 89.07) --
	(158.48, 86.15) --
	(162.22, 83.60) --
	(165.97, 81.16) --
	(169.71, 78.75) --
	(173.46, 76.43) --
	(177.20, 73.78) --
	(180.95, 71.03) --
	(184.69, 68.13) --
	(188.43, 64.86) --
	(192.18, 61.49) --
	(195.92, 57.82) --
	(199.67, 53.67) --
	(203.41, 49.66) --
	(207.16, 45.83);

\path[brsewstarstyle] ( 42.39, 89.14) --
	( 46.14, 90.76) --
	( 49.88, 92.52) --
	( 53.63, 94.36) --
	( 57.37, 96.35) --
	( 61.12, 98.28) --
	( 64.86,100.43) --
	( 68.61,102.38) --
	( 72.35,104.33) --
	( 76.10,106.05) --
	( 79.84,107.74) --
	( 83.59,109.07) --
	( 87.33,110.25) --
	( 91.07,111.11) --
	( 94.82,111.52) --
	( 98.56,111.76) --
	(102.31,111.43) --
	(106.05,110.37) --
	(109.80,109.20) --
	(113.54,107.64) --
	(117.29,106.25) --
	(121.03,104.72) --
	(124.78,103.02) --
	(128.52,100.95) --
	(132.27, 98.96) --
	(136.01, 96.38) --
	(139.75, 93.35) --
	(143.50, 90.12) --
	(147.24, 87.24) --
	(150.99, 84.78) --
	(154.73, 82.43) --
	(158.48, 80.50) --
	(162.22, 78.70) --
	(165.97, 77.01) --
	(169.71, 75.14) --
	(173.46, 73.27) --
	(177.20, 71.29) --
	(180.95, 69.05) --
	(184.69, 66.62) --
	(188.43, 63.66) --
	(192.18, 60.51) --
	(195.92, 56.87) --
	(199.67, 53.16) --
	(203.41, 49.59) --
	(207.16, 46.10);

\path[gensewstyle] ( 42.39, 81.08) --
	( 46.14, 80.99) --
	( 49.88, 81.05) --
	( 53.63, 81.05) --
	( 57.37, 81.02) --
	( 61.12, 81.05) --
	( 64.86, 81.06) --
	( 68.61, 80.98) --
	( 72.35, 80.99) --
	( 76.10, 80.97) --
	( 79.84, 80.94) --
	( 83.59, 80.85) --
	( 87.33, 80.71) --
	( 91.07, 80.61) --
	( 94.82, 80.39) --
	( 98.56, 79.97) --
	(102.31, 79.54) --
	(106.05, 78.94) --
	(109.80, 79.49) --
	(113.54, 80.55) --
	(117.29, 80.81) --
	(121.03, 80.78) --
	(124.78, 80.38) --
	(128.52, 79.59) --
	(132.27, 78.65) --
	(136.01, 77.36) --
	(139.75, 75.92) --
	(143.50, 76.26) --
	(147.24, 77.66) --
	(150.99, 78.40) --
	(154.73, 78.77) --
	(158.48, 78.76) --
	(162.22, 78.44) --
	(165.97, 77.79) --
	(169.71, 76.89) --
	(173.46, 75.02) --
	(177.20, 72.91) --
	(180.95, 71.04) --
	(184.69, 68.85) --
	(188.43, 66.05) --
	(192.18, 62.75) --
	(195.92, 58.51) --
	(199.67, 54.73) --
	(203.41, 50.45);

\path[gensewstarstyle] ( 42.39, 81.02) --
	( 46.14, 81.07) --
	( 49.88, 81.04) --
	( 53.63, 81.04) --
	( 57.37, 81.05) --
	( 61.12, 81.11) --
	( 64.86, 81.02) --
	( 68.61, 81.03) --
	( 72.35, 81.05) --
	( 76.10, 80.97) --
	( 79.84, 80.92) --
	( 83.59, 80.90) --
	( 87.33, 80.68) --
	( 91.07, 80.57) --
	( 94.82, 80.32) --
	( 98.56, 80.00) --
	(102.31, 79.59) --
	(106.05, 78.90) --
	(109.80, 78.65) --
	(113.54, 78.50) --
	(117.29, 78.32) --
	(121.03, 77.99) --
	(124.78, 77.54) --
	(128.52, 76.92) --
	(132.27, 76.13) --
	(136.01, 75.09) --
	(139.75, 73.68) --
	(143.50, 73.53) --
	(147.24, 73.93) --
	(150.99, 74.34) --
	(154.73, 74.53) --
	(158.48, 74.39) --
	(162.22, 74.23) --
	(165.97, 73.70) --
	(169.71, 73.15) --
	(173.46, 71.73) --
	(177.20, 70.15) --
	(180.95, 68.42) --
	(184.69, 66.59) --
	(188.43, 64.18) --
	(192.18, 61.14) --
	(195.92, 57.25) --
	(199.67, 53.58) --
	(203.41, 49.39);
\end{scope}
\begin{scope}
\path[draw=gray,line width=0.4*\lthickness pt,dash pattern=] (38.65, 30.69) -- (38.65, 175);
\node[text=black,anchor=base east,inner sep=0pt, outer sep=0pt, scale= \labscal] at ( 32, 34.22) {0.5};

\node[text=black,anchor=base east,inner sep=0pt, outer sep=0pt, scale= \labscal] at ( 32, 78.01) {1.0};

\node[text=black,anchor=base east,inner sep=0pt, outer sep=0pt, scale= \labscal] at ( 32,121.79) {1.5};

\node[text=black,anchor=base east,inner sep=0pt, outer sep=0pt, scale= \labscal] at ( 32,165.58) {2.0};
\end{scope}
\begin{scope}
\path[draw=black,line width= 0.6*\lthickness pt,line join=round] ( 35, 37.25) --
	( 38.65, 37.25);

\path[draw=black,line width= 0.6*\lthickness pt,line join=round] ( 35, 81.04) --
	( 38.65, 81.04);

\path[draw=black,line width= 0.6*\lthickness pt,line join=round] ( 35,124.82) --
	( 38.65,124.82);

\path[draw=black,line width= 0.6*\lthickness pt,line join=round] ( 35,168.61) --
	( 38.65,168.61);
\end{scope}
\path[draw=gray,line width=0.4*\lthickness pt,dash pattern=] (34.16, 81.04) -- (230, 81.04);
\begin{scope}
\path[draw=black,line width= 0.6*\lthickness pt,line join=round] ( 38.65, 27.94) --
	( 38.65, 30.69);

\path[draw=black,line width= 0.6*\lthickness pt,line join=round] ( 85.46, 27.94) --
	( 85.46, 30.69);

\path[draw=black,line width= 0.6*\lthickness pt,line join=round] (132.27, 27.94) --
	(132.27, 30.69);

\path[draw=black,line width= 0.6*\lthickness pt,line join=round] (179.07, 27.94) --
	(179.07, 30.69);
\end{scope}
\begin{scope}
\node[text=black,anchor=base,inner sep=0pt, outer sep=0pt, scale= \labscal] at ( 38.65, 19.68) {0.00};

\node[text=black,anchor=base,inner sep=0pt, outer sep=0pt, scale= \labscal] at ( 85.46, 19.68) {0.25};

\node[text=black,anchor=base,inner sep=0pt, outer sep=0pt, scale= \labscal] at (132.27, 19.68) {0.50};

\node[text=black,anchor=base,inner sep=0pt, outer sep=0pt, scale= \labscal] at (179.07, 19.68) {0.75};
\end{scope}
\begin{scope}
\node[text=black,anchor=base,inner sep=0pt, outer sep=0pt, scale=  1.5*\labscal] at (108.87,  180) {$\sew$ Policies};
\node[text=black,anchor=base,inner sep=0pt, outer sep=0pt, scale=  \labscal] at (108.87,  8.22) {$\lambda$};
\node[text=black,rotate= 90.00,anchor=base,inner sep=0pt, outer sep=0pt, scale= \labscal] at ( 12,0.5*\legendypos) {Normalized $\mathbb E[T]$};
\end{scope}

\end{tikzpicture} & \multirow{2}{*}[7cm]{\begin{tikzpicture}[x=1pt,y=1pt]
\begin{scope}
\renewcommand{\legendxpos}{200}
\renewcommand{\legendypos}{420}

\path[use as bounding box,fill=white,fill opacity=0.00] (0,0) rectangle (100,\legendypos);
\path[srcjsqstyle] (\legendxpos-200,0.2*\legendypos) -- (\legendxpos-170,0.2*\legendypos);
\node[text=black,anchor=base west,inner sep=0pt, outer sep=0pt, scale= \labscal] at (\legendxpos-160,0.2*\legendypos-2) {{$\ipoptd{\DP\srcq\jsq}$}};

\path[brjsqstyle] (\legendxpos-200,0.8*\legendypos) -- (\legendxpos-170,0.8*\legendypos);
\node[text=black,anchor=base west,inner sep=0pt, outer sep=0pt, scale= \labscal] at (\legendxpos-160,0.8*\legendypos-2) {{ $\DP{\brq}{\jsq}$}};
\path[brjsqstarstyle] (\legendxpos-200,0.75*\legendypos) -- (\legendxpos-170,0.75*\legendypos);
\node[text=black,anchor=base west,inner sep=0pt, outer sep=0pt, scale= \labscal] at (\legendxpos-160,0.75*\legendypos-2) {{$\DP{\brq}{\jsq^\star}$}};
\path[genjsqstyle] (\legendxpos-200,0.7*\legendypos) -- (\legendxpos-170,0.7*\legendypos);
\node[text=black,anchor=base west,inner sep=0pt, outer sep=0pt, scale= \labscal] at (\legendxpos-160,0.7*\legendypos-2) {{$\DP{\genseedq}{\jsq}$}};
\path[genjsqstarstyle] (\legendxpos-200,0.65*\legendypos) -- (\legendxpos-170,0.65*\legendypos);
\node[text=black,anchor=base west,inner sep=0pt, outer sep=0pt, scale= \labscal] at (\legendxpos-160,0.65*\legendypos-2) {{$\DP{\genseedq}{\jsq^\star}$}};

\path[brsedstyle] (\legendxpos-200,0.6*\legendypos) -- (\legendxpos-170,0.6*\legendypos);
\node[text=black,anchor=base west,inner sep=0pt, outer sep=0pt, scale= \labscal] at (\legendxpos-160,0.6*\legendypos-2) {{ $\DP{\brq}{\sed}$}};
\path[brsedstarstyle] (\legendxpos-200,0.55*\legendypos) -- (\legendxpos-170,0.55*\legendypos);
\node[text=black,anchor=base west,inner sep=0pt, outer sep=0pt, scale= \labscal] at (\legendxpos-160,0.55*\legendypos-2) {{$\DP{\brq}{\sed^\star}$}};
\path[gensedstyle] (\legendxpos-200,0.5*\legendypos) -- (\legendxpos-170,0.5*\legendypos);
\node[text=black,anchor=base west,inner sep=0pt, outer sep=0pt, scale= \labscal] at (\legendxpos-160,0.5*\legendypos-2) {{$\DP{\genseedq}{\sed}$}};
\path[gensedstarstyle] (\legendxpos-200,0.45*\legendypos) -- (\legendxpos-170,0.45*\legendypos);
\node[text=black,anchor=base west,inner sep=0pt, outer sep=0pt, scale= \labscal] at (\legendxpos-160,0.45*\legendypos-2) {{$\DP{\genseedq}{\sed^\star}$}};

\path[brsewstyle] (\legendxpos-200,0.4*\legendypos) -- (\legendxpos-170,0.4*\legendypos);
\node[text=black,anchor=base west,inner sep=0pt, outer sep=0pt, scale= \labscal] at (\legendxpos-160,0.4*\legendypos-2) {{ $\DP{\brq}{\sew}$}};
\path[brsewstarstyle] (\legendxpos-200,0.35*\legendypos) -- (\legendxpos-170,0.35*\legendypos);
\node[text=black,anchor=base west,inner sep=0pt, outer sep=0pt, scale= \labscal] at (\legendxpos-160,0.35*\legendypos-2) {{$\DP{\brq}{\sew^\star}$}};
\path[gensewstyle] (\legendxpos-200,0.3*\legendypos) -- (\legendxpos-170,0.3*\legendypos);
\node[text=black,anchor=base west,inner sep=0pt, outer sep=0pt, scale= \labscal] at (\legendxpos-160,0.3*\legendypos-2) {{$\DP{\genseedq}{\sew}$}};
\path[gensewstarstyle] (\legendxpos-200,0.25*\legendypos) -- (\legendxpos-170,0.25*\legendypos);
\node[text=black,anchor=base west,inner sep=0pt, outer sep=0pt, scale= \labscal] at (\legendxpos-160,0.25*\legendypos-2) {{$\DP{\genseedq}{\sew^\star}$}};
\end{scope}

\end{tikzpicture}} \\
\begin{tikzpicture}[x=1pt,y=1pt]


\definecolor{fillColor}{RGB}{255,255,255}
\path[use as bounding box,fill=fillColor,fill opacity=0.00] (0,0) rectangle (\legendxpos-100,\legendypos);
\begin{scope}

\path[brsedstyle] ( 43.54,121.00) --
	( 47.80,122.85) --
	( 52.07,124.52) --
	( 56.33,126.38) --
	( 60.60,128.29) --
	( 64.86,130.26) --
	( 69.13,132.23) --
	( 73.39,134.25) --
	( 77.66,135.95) --
	( 81.92,137.56) --
	( 86.19,138.88) --
	( 90.45,140.22) --
	( 94.71,141.12) --
	( 98.98,141.65) --
	(103.24,141.76) --
	(107.51,141.10) --
	(111.77,140.16) --
	(116.04,138.29) --
	(120.30,135.98) --
	(124.57,133.43) --
	(128.83,131.19) --
	(133.10,129.39) --
	(137.36,127.58) --
	(141.63,125.41) --
	(145.89,123.05) --
	(150.16,120.07) --
	(154.42,116.66) --
	(158.69,113.28) --
	(162.95,110.30) --
	(167.21,108.53) --
	(171.48,108.92) --
	(175.74,109.04) --
	(180.01,110.91) --
	(184.27,113.32) --
	(188.54,116.14) --
	(192.80,119.07) --
	(197.07,120.45) --
	(201.33,119.74) --
	(205.60,118.23) --
	(209.86,113.23) --
	(214.13,106.61) --
	(218.39, 95.68) --
	(222.66, 84.31) --
	(226.92, 69.81) --
	(231.19, 55.24);
\definecolor{drawColor}{RGB}{124,174,0}

\path[brsedstarstyle] ( 43.54,121.02) --
	( 47.80,122.77) --
	( 52.07,124.60) --
	( 56.33,126.38) --
	( 60.60,128.18) --
	( 64.86,130.38) --
	( 69.13,132.21) --
	( 73.39,134.21) --
	( 77.66,135.92) --
	( 81.92,137.67) --
	( 86.19,139.24) --
	( 90.45,140.25) --
	( 94.71,141.20) --
	( 98.98,141.52) --
	(103.24,141.93) --
	(107.51,141.27) --
	(111.77,140.13) --
	(116.04,138.29) --
	(120.30,135.58) --
	(124.57,133.71) --
	(128.83,131.40) --
	(133.10,129.84) --
	(137.36,127.79) --
	(141.63,125.87) --
	(145.89,123.41) --
	(150.16,120.53) --
	(154.42,117.15) --
	(158.69,113.57) --
	(162.95,110.90) --
	(167.21,110.45) --
	(171.48,110.46) --
	(175.74,112.25) --
	(180.01,115.34) --
	(184.27,120.07) --
	(188.54,126.15) --
	(192.80,133.00) --
	(197.07,139.16) --
	(201.33,142.73) --
	(205.60,142.31) --
	(209.86,137.18) --
	(214.13,127.80) --
	(218.39,115.46) --
	(222.66,100.34) --
	(226.92, 83.52) --
	(231.19, 66.25);
\definecolor{drawColor}{RGB}{0,191,196}

\path[gensedstyle] ( 43.54,102.92) --
	( 47.80,102.99) --
	( 52.07,102.88) --
	( 56.33,102.93) --
	( 60.60,102.88) --
	( 64.86,102.83) --
	( 69.13,102.95) --
	( 73.39,102.80) --
	( 77.66,103.08) --
	( 81.92,102.79) --
	( 86.19,102.62) --
	( 90.45,102.31) --
	( 94.71,102.10) --
	( 98.98,101.77) --
	(103.24,101.31) --
	(107.51,100.22) --
	(111.77, 99.26) --
	(116.04, 97.47) --
	(120.30, 96.96) --
	(124.57, 97.38) --
	(128.83, 97.86) --
	(133.10, 98.60) --
	(137.36, 99.00) --
	(141.63, 98.93) --
	(145.89, 97.99) --
	(150.16, 96.11) --
	(154.42, 93.35) --
	(158.69, 90.95) --
	(162.95, 90.10) --
	(167.21, 90.48) --
	(171.48, 92.81) --
	(175.74, 96.09) --
	(180.01,100.98) --
	(184.27,106.58) --
	(188.54,111.83) --
	(192.80,115.40) --
	(197.07,116.06) --
	(201.33,115.13) --
	(205.60,111.83) --
	(209.86,105.30) --
	(214.13, 94.89) --
	(218.39, 84.38) --
	(222.66, 71.53) --
	(226.92, 56.77) --
	(231.19, 40.53);
\definecolor{drawColor}{RGB}{199,124,255}

\path[gensedstarstyle] ( 43.54,102.88) --
	( 47.80,102.90) --
	( 52.07,103.00) --
	( 56.33,102.92) --
	( 60.60,103.03) --
	( 64.86,102.96) --
	( 69.13,102.89) --
	( 73.39,102.84) --
	( 77.66,102.78) --
	( 81.92,102.85) --
	( 86.19,102.63) --
	( 90.45,102.33) --
	( 94.71,102.14) --
	( 98.98,101.70) --
	(103.24,101.10) --
	(107.51,100.33) --
	(111.77, 99.06) --
	(116.04, 97.46) --
	(120.30, 96.77) --
	(124.57, 97.23) --
	(128.83, 98.21) --
	(133.10, 98.58) --
	(137.36, 98.98) --
	(141.63, 98.88) --
	(145.89, 97.99) --
	(150.16, 96.44) --
	(154.42, 93.60) --
	(158.69, 90.82) --
	(162.95, 90.11) --
	(167.21, 90.93) --
	(171.48, 92.82) --
	(175.74, 97.72) --
	(180.01,103.05) --
	(184.27,111.58) --
	(188.54,120.08) --
	(192.80,128.94) --
	(197.07,134.47) --
	(201.33,136.17) --
	(205.60,132.95) --
	(209.86,125.71) --
	(214.13,114.05) --
	(218.39,100.77) --
	(222.66, 86.31) --
	(226.92, 69.18) --
	(231.19, 51.07);
\end{scope}
\begin{scope}
\path[draw=gray,line width=0.4*\lthickness pt,dash pattern=] (38.69, 30.69) -- (38.69, 175);
\node[text=black,anchor=base east,inner sep=0pt, outer sep=0pt, scale=  \labscal] at ( 32, 56.12) {0.8};

\node[text=black,anchor=base east,inner sep=0pt, outer sep=0pt, scale=  \labscal] at ( 32, 99.90) {1.0};

\node[text=black,anchor=base east,inner sep=0pt, outer sep=0pt, scale=  \labscal] at ( 32,143.68) {1.2};
\end{scope}
\begin{scope}
\path[draw=black,line width= 0.6*\lthickness pt,line join=round] ( 35, 59.15) --
	( 38.69, 59.15);

\path[draw=black,line width= 0.6*\lthickness pt,line join=round] ( 35,102.93) --
	( 38.69,102.93);

\path[draw=black,line width= 0.6*\lthickness pt,line join=round] ( 35,146.72) --
	( 38.69,146.72);
\end{scope}
\path[draw=gray,line width=0.4*\lthickness pt,dash pattern=] (34.16,102.93) -- (230,102.93);
\begin{scope}
\path[draw=black,line width= 0.6*\lthickness pt,line join=round] ( 38.69, 27.94) --
	( 38.69, 30.69);

\path[draw=black,line width= 0.6*\lthickness pt,line join=round] ( 85.91, 27.94) --
	( 85.91, 30.69);

\path[draw=black,line width= 0.6*\lthickness pt,line join=round] (133.13, 27.94) --
	(133.13, 30.69);

\path[draw=black,line width= 0.6*\lthickness pt,line join=round] (180.34, 27.94) --
	(180.34, 30.69);
\end{scope}
\begin{scope}
\node[text=black,anchor=base,inner sep=0pt, outer sep=0pt, scale=  \labscal] at ( 38.69, 19.68) {0.00};

\node[text=black,anchor=base,inner sep=0pt, outer sep=0pt, scale=  \labscal] at ( 85.91, 19.68) {0.25};

\node[text=black,anchor=base,inner sep=0pt, outer sep=0pt, scale=  \labscal] at (133.13, 19.68) {0.50};

\node[text=black,anchor=base,inner sep=0pt, outer sep=0pt, scale=  \labscal] at (180.34, 19.68) {0.75};
\end{scope}
\begin{scope}
\node[text=black,anchor=base,inner sep=0pt, outer sep=0pt, scale=  1.5*\labscal] at (109.52,  180) {$\sed$ Policies};
\node[text=black,anchor=base,inner sep=0pt, outer sep=0pt, scale=  \labscal] at (109.52,  8.22) {$\lambda$};
\node[text=black,rotate= 90.00,anchor=base,inner sep=0pt, outer sep=0pt, scale= \labscal] at ( 12,0.5*\legendypos) {Normalized $\mathbb E[T]$};
\end{scope}

\end{tikzpicture} & 
\begin{tikzpicture}[x=1pt,y=1pt]


\definecolor{fillColor}{RGB}{255,255,255}
\path[use as bounding box,fill=fillColor,fill opacity=0.00] (0,0) rectangle (\legendxpos-100,\legendypos);
\begin{scope}

\path[genjsqstarstyle] ( 43.47,128.22) --
	( 47.70,128.12) --
	( 51.93,128.25) --
	( 56.16,128.19) --
	( 60.39,128.17) --
	( 64.63,128.34) --
	( 68.86,128.08) --
	( 73.09,128.16) --
	( 77.32,128.12) --
	( 81.55,128.11) --
	( 85.79,127.90) --
	( 90.02,127.68) --
	( 94.25,127.39) --
	( 98.48,127.04) --
	(102.71,126.56) --
	(106.94,125.88) --
	(111.18,124.53) --
	(115.41,123.42) --
	(119.64,122.55) --
	(123.87,122.25) --
	(128.10,122.03) --
	(132.34,121.47) --
	(136.57,120.71) --
	(140.80,119.38) --
	(145.03,118.06) --
	(149.26,116.13) --
	(153.50,113.59) --
	(157.73,113.34) --
	(161.96,114.63) --
	(166.19,115.58) --
	(170.42,116.41) --
	(174.65,117.40) --
	(178.89,118.44) --
	(183.12,118.95) --
	(187.35,119.32) --
	(191.58,119.54) --
	(195.81,118.65) --
	(200.05,116.56) --
	(204.28,113.69) --
	(208.51,109.03) --
	(212.74,104.25) --
	(216.97, 96.00) --
	(221.21, 85.51) --
	(225.44, 74.15) --
	(229.67, 60.18);

\path[gensedstyle] ( 43.47,128.18) --
	( 47.70,128.24) --
	( 51.93,128.15) --
	( 56.16,128.19) --
	( 60.39,128.14) --
	( 64.63,128.10) --
	( 68.86,128.21) --
	( 73.09,128.07) --
	( 77.32,128.33) --
	( 81.55,128.07) --
	( 85.79,127.90) --
	( 90.02,127.62) --
	( 94.25,127.43) --
	( 98.48,127.12) --
	(102.71,126.69) --
	(106.94,125.69) --
	(111.18,124.80) --
	(115.41,123.15) --
	(119.64,122.68) --
	(123.87,123.07) --
	(128.10,123.51) --
	(132.34,124.19) --
	(136.57,124.56) --
	(140.80,124.49) --
	(145.03,123.63) --
	(149.26,121.90) --
	(153.50,119.35) --
	(157.73,117.13) --
	(161.96,116.35) --
	(166.19,116.69) --
	(170.42,118.85) --
	(174.65,121.88) --
	(178.89,126.39) --
	(183.12,131.56) --
	(187.35,136.40) --
	(191.58,139.70) --
	(195.81,140.31) --
	(200.05,139.45) --
	(204.28,136.41) --
	(208.51,130.38) --
	(212.74,120.77) --
	(216.97,111.07) --
	(221.21, 99.20) --
	(225.44, 85.58) --
	(229.67, 70.59);

\path[gensewstarstyle] ( 43.47,128.33) --
	( 47.70,128.14) --
	( 51.93,128.17) --
	( 56.16,128.03) --
	( 60.39,128.23) --
	( 64.63,128.28) --
	( 68.86,128.19) --
	( 73.09,128.24) --
	( 77.32,127.97) --
	( 81.55,128.05) --
	( 85.79,127.94) --
	( 90.02,127.72) --
	( 94.25,127.58) --
	( 98.48,127.04) --
	(102.71,126.45) --
	(106.94,125.82) --
	(111.18,124.70) --
	(115.41,123.15) --
	(119.64,122.47) --
	(123.87,122.26) --
	(128.10,121.86) --
	(132.34,120.90) --
	(136.57,120.11) --
	(140.80,118.88) --
	(145.03,117.23) --
	(149.26,114.76) --
	(153.50,111.97) --
	(157.73,111.34) --
	(161.96,112.55) --
	(166.19,112.69) --
	(170.42,113.32) --
	(174.65,113.52) --
	(178.89,113.12) --
	(183.12,112.47) --
	(187.35,111.28) --
	(191.58,109.16) --
	(195.81,106.10) --
	(200.05,102.42) --
	(204.28, 97.53) --
	(208.51, 92.10) --
	(212.74, 85.70) --
	(216.97, 76.47) --
	(221.21, 67.04) --
	(225.44, 56.76) --
	(229.67, 45.19);

\path[srcjsqstyle] ( 43.47,128.19) --
	( 47.70,128.19) --
	( 51.93,128.19) --
	( 56.16,128.19) --
	( 60.39,128.19) --
	( 64.63,128.18) --
	( 68.86,128.17) --
	( 73.09,128.15) --
	( 77.32,128.11) --
	( 81.55,128.04) --
	( 85.79,127.93) --
	( 90.02,127.75) --
	( 94.25,127.48) --
	( 98.48,127.09) --
	(102.71,126.53) --
	(106.94,125.75) --
	(111.18,124.69) --
	(115.41,123.28) --
	(119.64,121.61) --
	(123.87,120.39) --
	(128.10,119.65) --
	(132.34,119.32) --
	(136.57,119.35) --
	(140.80,119.14) --
	(145.03,117.94) --
	(149.26,116.01) --
	(153.50,113.38) --
	(157.73,110.35) --
	(161.96,108.29) --
	(166.19,107.29) --
	(170.42,107.30) --
	(174.65,108.30) --
	(178.89,109.52) --
	(183.12,109.85) --
	(187.35,109.40) --
	(191.58,108.03) --
	(195.81,105.63) --
	(200.05,102.46) --
	(204.28, 98.41) --
	(208.51, 93.24) --
	(212.74, 86.63) --
	(216.97, 78.51) --
	(221.21, 69.36) --
	(225.44, 58.81) --
	(229.67, 46.56);
\end{scope}
\begin{scope}
\path[draw=gray,line width=0.4*\lthickness pt,dash pattern=] (39.23, 30.69) -- (39.23, 175);
\node[text=black,anchor=base east,inner sep=0pt, outer sep=0pt, scale= \labscal] at ( 32.5, 44.33) {0.6};

\node[text=black,anchor=base east,inner sep=0pt, outer sep=0pt, scale= \labscal] at ( 32.5, 84.74) {0.8};

\node[text=black,anchor=base east,inner sep=0pt, outer sep=0pt, scale= \labscal] at ( 32.5,125.16) {1.0};

\node[text=black,anchor=base east,inner sep=0pt, outer sep=0pt, scale= \labscal] at ( 32.5,165.58) {1.2};
\end{scope}
\begin{scope}
\path[draw=black,line width= 0.6*\lthickness pt,line join=round] ( 35.5, 47.36) --
	( 39.23, 47.36);

\path[draw=black,line width= 0.6*\lthickness pt,line join=round] ( 35.5, 87.77) --
	( 39.23, 87.77);

\path[draw=black,line width= 0.6*\lthickness pt,line join=round] ( 35.5,128.19) --
	( 39.23,128.19);

\path[draw=black,line width= 0.6*\lthickness pt,line join=round] ( 35.5,168.61) --
	( 39.23,168.61);
\end{scope}
\path[draw=gray,line width=0.4*\lthickness pt,dash pattern=] (31.41,128.19) -- (230,128.19);
\begin{scope}
\path[draw=black,line width= 0.6*\lthickness pt,line join=round] ( 39.23, 27.94) --
	( 39.23, 30.69);

\path[draw=black,line width= 0.6*\lthickness pt,line join=round] ( 92.13, 27.94) --
	( 92.13, 30.69);

\path[draw=black,line width= 0.6*\lthickness pt,line join=round] (145.03, 27.94) --
	(145.03, 30.69);

\path[draw=black,line width= 0.6*\lthickness pt,line join=round] (197.93, 27.94) --
	(197.93, 30.69);
\end{scope}
\begin{scope}
\node[text=black,anchor=base,inner sep=0pt, outer sep=0pt, scale= \labscal] at (39.24, 19.68) {0.00};

\node[text=black,anchor=base,inner sep=0pt, outer sep=0pt, scale= \labscal] at (92.14, 19.68) {0.25};

\node[text=black,anchor=base,inner sep=0pt, outer sep=0pt, scale= \labscal] at (145.03, 19.68) {0.50};

\node[text=black,anchor=base,inner sep=0pt, outer sep=0pt, scale= \labscal] at (197.93, 19.68) {0.75};
\end{scope}
\begin{scope}
\node[text=black,anchor=base,inner sep=0pt, outer sep=0pt, scale=  1.5*\labscal] at (118.59,  180) {Best Policies};
\node[text=black,anchor=base,inner sep=0pt, outer sep=0pt, scale=  \labscal] at (118.59,  8.22) {$\lambda$};
\node[text=black,rotate= 90.00,anchor=base,inner sep=0pt, outer sep=0pt, scale= \labscal] at ( 12,0.5*\legendypos) {Normalized $\mathbb E[T]$};
\end{scope}

\end{tikzpicture} 
\end{tabular}}
    \caption{$\mathbb E[T]$ relative to that of $\genseed$ (i.e., $\mathbb E[T]^{\mathsf{DP}}/\mathbb E[T]^{\genseed}$) as a function of $\lambda$ for the parameter settings where $s=d=3$, with $\lambda$ varying over $\{0.02,0.04,\ldots,0.90\}$,  $(q_1,q_2,q_3)=(1/3,1/6,1/2)$ and $(R_1,R_2)=(5,2)$, yielding $(\mu_1,\mu_2,\mu_3)=(2,4/5,2/5)$, for the dispatching policies $\mathsf{DP}$ constructed from the $\brq$ and $\genseedq$ querying rules when paired with the $\jsq$ and $\jsq^\star$ assignment rules (upper left), the $\sed$ and $\sed^\star$ assignment rules (lower left), and the $\sew$ and $\sew^\star$ assignment rules (upper right).  In the lower right, we compare the ``best'' policy (obtaining the lowest $\mathbb E[T]$ value on the majority of $\lambda$ values) from each of the other three subfigures: $\DP{\genseedq}{\jsq^\star}$, $\DP{\genseedq}{\sew^\star}$, and $\DP{\genseedq}{\sed}$.   We also include the $\ipoptd{\DP\srcq\jsq}$ dispatching policy in the lower right subfigure.   All $\mathbb E[T]$ values were obtained through simulation except for those associated with $\ipoptd{\DP\srcq\jsq}$ and the normalizing policy, $\genseedq$.  Note that not all subfigures use the same scale for the vertical axis.}
    \label{fig:mixandmatch}
\end{figure}
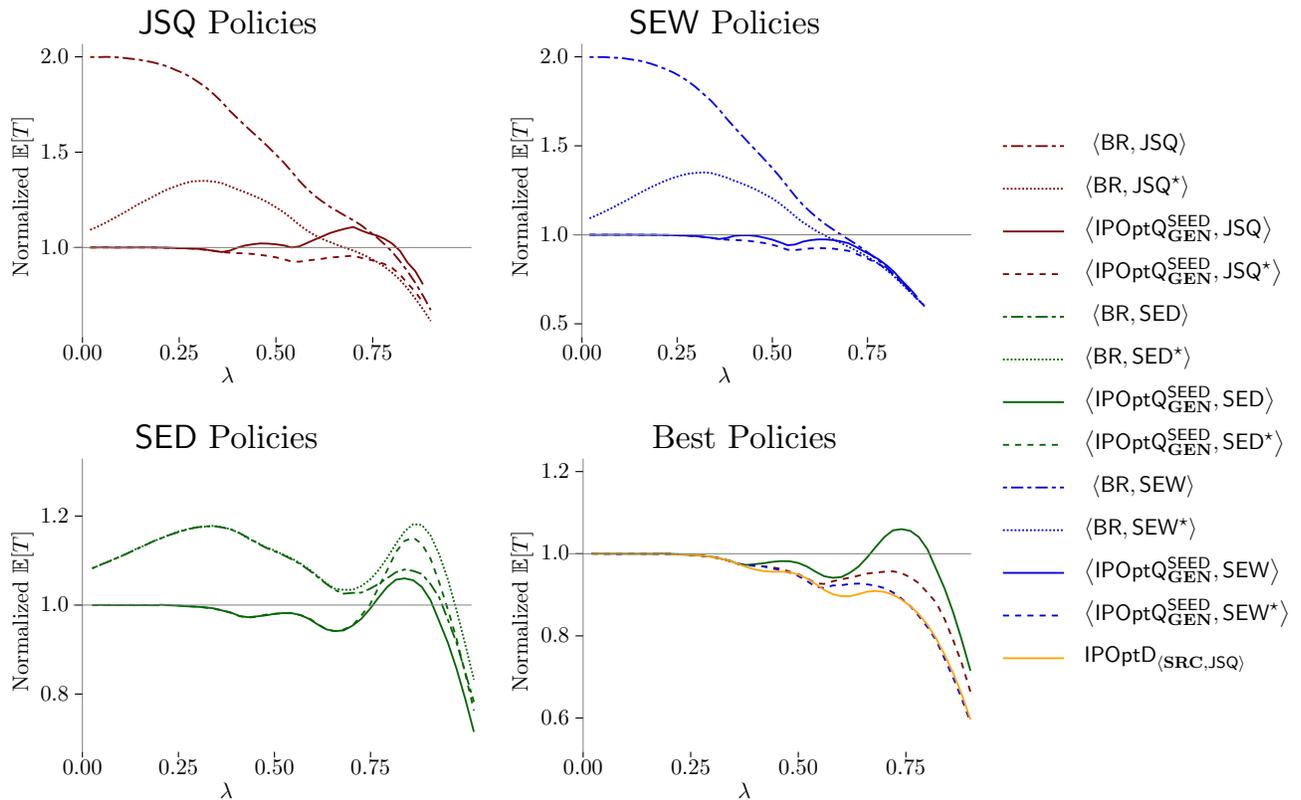

In Figure~\ref{fig:mixandmatch}, we plot the simulated $\mathbb E[T]$ of each of the above dispatching policies---as well as the computational (non-simulated, based on the assumption where $k\to\infty$) results for $\ipoptd{\iidq}$ and $\ipoptd{\DP\srcq\jsq}$---normalized by the $\mathbb E[T]$ value $\genseed \equiv \DP\genseedq\genseeda$ as a function of $\lambda$.  We examined a number of other parameter settings and chose this parameter setting in order to make the trends more salient, although qualitatively similar results are exhibited across most of the parameter settings observed.

We observe that at low values of $\lambda$, the $\brq$-driven policies perform poorly, because they occasionally query no servers of the fastest class, even though under such light traffic one would like to discard all but the fastest servers.  These policies continue to be the worst performers---as, in addition to using slow servers, queues begin to build up at these servers---until a certain point where the gap between these policies and the others begins to close. Meanwhile, in this low-$\lambda$ regime, all of the other policies (including the normalizing policy, $\genseed$, which does not make use of queue length information) perform near-identically, because all of them query essentially only the fastest servers, and most of these servers are idle, rendering the assignment rule immaterial.
At higher $\lambda$, we enter an ``assignment-driven regime,'' where all of our $\clda$-based policies outperform the $\cida$-based $\genseed$ policy.
We call this an ``assignment-driven regime'' because the performance of the policies become differentiated from one another primarily on the basis of their assignment rules.  That is, even though, e.g., $\genseed\equiv\DP{\genseedq}{\genseeda}$ and $\DP{\genseedq}{\sed}$ use the same querying rule---which is optimized for use with $\cida$ assignment---the latter achieves better performance because the advantage of $\clda$-based assignment outweighs the benefit of jointly optimizing the querying and assignment rules.  The existence of such a regime is a result of the fact that, in heavy traffic, all querying rules that maintain the system's stability must result in similar $\lb$ values, meaning that dispatching policies that stabilize the system are distinguished from one another primarily in terms of their assignment rules.

Two of the dispatching policies under consideration emerge as the consistently strongest performers: $\DP\genseedq{\sew^\star}$ and $\ipoptd{\DP\srcq\jsq}$.  It may appear surprising that $\DP{\genseedq}{\sew^\star}$ consistently outperforms its counterparts that make use of $\sed$ or $\sed^\star$.  It turns out that (assuming a judiciously chosen querying rule) it is crucial to make use of idle queried servers whenever possible; unlike $\sed$ and $\sed^\star$, $\sew$ and $\sew^\star$ never send jobs to busy servers when an idle server has been queried. The strong performance of $\DP\genseedq{\sew^\star}$ also highlights the value of our analysis of $\cida$-based dispatching policies: even though such length-unaware policies do not themselves necessarily achieve the best performance (especially at high $\lambda$), we see that the $\cida$-based optimization of the querying rule allows for the development of considerably stronger $\clda$-based policies. Such policies are likely to be difficult to discover using, e.g., a grid search approach. 

\begin{remark}
In fact, the best-performing $\DP\qr{\sew^\star}$ policies found by a simulation-driven grid search (over $\qr\in\genq$) performed no better than $\DP\genseedq{\sew^\star}$. We performed this grid search to validate the performance of our heuristic policies, but, even with a fairly coarse search, this took on the order of an hour for a single value of $\lambda$, while one can obtain a better performing $\genseedq$ policy in mere seconds by leveraging on our optimization problems rather than simulations. Our experience leads us to conclude that the high dimensionality of the $\genq$ family renders such searches poorly-suited for practice.
\end{remark}

Meanwhile, $\ipoptd{\DP\srcq\jsq}$ also exhibits consistently strong performance across the range of $\lambda$ values.  We can attribute its excellent performance to the fact that---unlike the other $\clda$-based policies under consideration---$\ipoptd{\DP\srcq\jsq}$ features a querying rule that is optimized for use with its own assignment rule, rather than for use with a $\cida$ assignment rule.

\section{Conclusion}
\label{sec:conclusion}
This paper provides a comprehensive framework for dispatching in scalable systems in the presence of heterogeneous servers, by examining two separate components of the dispatching policy: the querying rule and the assignment rule.  We highlight tradeoffs associated with the choice of each rule: less restrictive families of querying rules allow for lower mean response times at the cost of increased solution runtime.  Meanwhile, some assignment rules lend themselves to tractable analysis, while others boast better performance (insofar as observed from simulations).  Moreover, at some system loads, both the querying and assignment decision can be crucial, while at more extreme loads, one decision plays a more dominant role over the other (subject to stability constraints).

Our framework illuminates several potentially fruitful areas of future work.  First, this paper restricts attention to symmetric and static querying and assignment rules.  There has been little study of asymmetric rules (of either kind) in the literature when all jobs are \emph{ex ante} identical (i.e., when dispatching is size blind and all jobs are equally important).  Yet we believe that the explicit and separate consideration of---and study of the interaction between---querying and assignment rules suggests how asymmetry might be exploited to develop superior dispatching policies even when jobs are \emph{ex ante} identical.  A judiciously chosen asymmetric assignment rule may be able to synergistically exploit the asymmetry introduced by the querying rule.  Meanwhile, future research could allow for dynamic, rather than merely static, querying and/or assignment rules, permitting the incorporation of round-robin-like dispatching decisions into our framework, which would necessitate novel analysis.
Another direction for future work involves generalizing our framework to heterogeneous systems with multiple dispatchers, as considered in~\cite{stolyar2017pull,zhou2021asymptotically}. Such a generalization likely would require a different approach for selecting policy parameters, as each dispatcher possesses only a partial view of the system's arrival process.

While this paper presents a comprehensive examination of querying rules within the space restricted by the aforementioned assumptions, the bulk of our analysis focused on the $\cida$ family, where assignment rules eschew making decisions on the basis of detailed queue length information in favor of idleness information.  The performance analysis of even $\clda$-based dispatching policies remains an open problem, and while the explicit analysis of the set of all $\clda$ assignment rules (in conjunction with querying rules coming from, say $\genq$) may prove intractable, we anticipate that many policies incorporating more detailed---if still restricted---queue length information are amenable to analysis.  Moreover, we imagine that many such policies may outperform the $\cida$-driven dispatching policies studied in this paper.

Finally, there remain open problems on the theoretical front. For example, throughout our analysis asymptotic independence remains an assumption (although one that is validated by simulation) for which future work may provide a universal rigorous justification (as past work has for more restricted special cases of our framework).  There is also ample room for optimization theory to shed further light on the structure of the optimization problems presented in this work.

\newpage

\appendix

\section{Appendix: Tables of Notation}
\label{app:notation}
\begin{table}[!htp]
\caption{Table \ref{app:notation}1: Querying Rule and Policy Abbreviations}
\footnotesize
\begin{tabular}{>{\raggedright}p{0.15\textwidth}@{}>{\centering}p{0.05\textwidth}@{}p{0.68\textwidth}>{\raggedleft\arraybackslash}p{0.05\textwidth}}
    \toprule
    \multicolumn{4}{c}{\emph{Querying Rule Families}}\\
    $\detq$ &$\equiv$& \textbf{Deterministic Class Mix} \dotfill & p. \pageref{detq}\\
    $\genq$ &$\equiv$& \textbf{General Class Mix} \dotfill & p. \pageref{genq}\\
    $\indq$ &$\equiv$& \textbf{Independent Querying} \dotfill& p. \pageref{indq}\\
    $\iidq$ &$\equiv$& \textbf{Independent and Identically Distributed Querying} \dotfill& p. \pageref{iidq}\\
    $\qrf$ &$\equiv$& Generic notation for an arbitrary querying rule family \dotfill & p. \pageref{qrf}\\
    $\sfcq$ &$\equiv$& \textbf{Single Fixed Class} \dotfill & p. \pageref{sfc}\\
    $\srcq$ &$\equiv$& \textbf{Single Random Class} \dotfill & p. \pageref{srcq}\\
    \midrule
    \multicolumn{4}{c}{\emph{Individual Querying Rules}}\\
    $\brq$ &$\equiv$& \textsf{Balanced Routing} querying rule \dotfill& p. \pageref{br}\\
    $\mathsf{DQR}_{\dvec}$&$\equiv$&Individual querying rule in $\detq$ that always queries according to class mix $\dvec$\dotfill&p. \pageref{dqrsub}\\
    $\ipoptq{\DP{\qrf}{\arf}}$ & $\equiv$& Querying rule used by the ``optimal'' policy in $\DP{\qrf}{\arf}$ found by IPOPT \dotfill & p. \pageref{ipoptq}\\
    $\ipoptq{\qrf}$ &$\equiv$& Abbreviated notation for $\ipoptq{\DP{\qrf}{\cida}}$  \dotfill & p. \pageref{ipoptq}\\
    $\genseedq$ &$\equiv$& Querying rule used by the $\genseed$ dispatching policy  \dotfill & p. \pageref{genseedq}\\
    $\qr$ &$\equiv$& Generic notation for an arbitrary individual querying rule \dotfill & p. \pageref{QR}\\
    $\uniq$ &$\equiv$& \textsf{Uniform Querying} \dotfill& p. \pageref{uniq}\\
    \bottomrule
\end{tabular}
\end{table}

\begin{table}[!htp]
\caption{Table \ref{app:notation}2: Assignment Rule and Policy Abbreviations}
\footnotesize
\begin{tabular}{>{\raggedright}p{0.15\textwidth}@{}>{\centering}p{0.05\textwidth}@{}p{0.68\textwidth}>{\raggedleft\arraybackslash}p{0.05\textwidth}}
    \toprule
    \multicolumn{4}{c}{\emph{Assignment Rule Families}}\\
    $\arf$ &$\equiv$& Generic notation for an arbitrary assignment rule family \dotfill& p. \pageref{arf}\\
    $\cda$ &$\equiv$& \textbf{Class Differentiated} \dotfill& p. \pageref{cda}\\
    $\cida$ &$\equiv$& \textbf{Class and Idleness Differentiated} \dotfill& p. \pageref{cida}\\
    $\clda$ &$\equiv$& \textbf{Class and Length Differentiated} \dotfill & p. \pageref{clda}\\
    $\ida$ &$\equiv$& \textbf{Idleness Differentiated} \dotfill& p. \pageref{ida}\\
    $\lda$ &$\equiv$& \textbf{Length Differentiated} \dotfill& p. \pageref{lda}\\
    \midrule
    \multicolumn{4}{c}{\emph{Individual Assignment Rules}}\\
    $\ar$ &$\equiv$& Generic notation for an arbitrary individual assignment rule  \dotfill& p. \pageref{AR}\\
    $\ipopta{\DP{\qrf}{\arf}}$ & $\equiv$& Assignment rule used by the ``optimal'' policy in $\DP{\qrf}{\arf}$ found by IPOPT \dotfill & p. \pageref{ipopta}\\
    $\ipopta{\qrf}$ &$\equiv$& Abbreviated notation for $\ipopta{\DP{\qrf}{\cida}}$ \dotfill & p.  \pageref{ipopta}\\
    $\genseeda$ &$\equiv$& Assignment rule used by the $\genseed$ dispatching policy  \dotfill & p. \pageref{genseeda}\\
    $\jsq$ &$\equiv$& \textsf{Join the Shortest Queue} assignment rule \dotfill& p. \pageref{jsq}\\
    $\jsq^\star$ &$\equiv$& Variant of $\jsq$ where ties are broken in favor of faster classes \dotfill& p. \pageref{jsqstar}\\
    $\nda$ &$\equiv$& \textsf{Non-Differentiated} \dotfill& p. \pageref{nda}\\
    $\sed$ &$\equiv$& \textsf{Shortest Expected Delay} assignment rule \dotfill& p. \pageref{sed}\\
    $\sed^\star$ &$\equiv$& Variant of $\sed$ where ties are broken in favor of faster classes \dotfill& p. \pageref{sedstar}\\
    $\sew$ &$\equiv$& \textsf{Shortest Expected Wait} assignment rule \dotfill& p. \pageref{sew}\\
    $\sew^\star$ &$\equiv$& Variant of $\sew$ where ties are broken in favor of faster classes \dotfill& p. \pageref{sewstar}\\
    \bottomrule
\end{tabular}
\end{table}

\newpage

\begin{table}[!htp]
\caption{Table \ref{app:notation}3: Dispatching Rule and Policy Abbreviations}
\footnotesize
\begin{tabular}{>{\raggedright}p{0.15\textwidth}@{}>{\centering}p{0.05\textwidth}@{}p{0.68\textwidth}>{\raggedleft\arraybackslash}p{0.05\textwidth}}
    \toprule
    \multicolumn{4}{c}{\emph{Dispatching Policy Families}}\\
    $\DP{\qrf}{\arf}$ &$\equiv$& Dispatching policy family using $\qrf$ querying and $\arf$ assignment \dotfill & p. \pageref{qrfarf} \\
    $\DP{\qrf}{\ar}$ &$\equiv$& Disp. policy family using $\qrf$ querying with the individual $\ar$ assignment rule \dotfill & p. \pageref{qrfar} \\
    $\DP{\qr}{\arf}$ &$\equiv$& Disp. policy family using the individual $\qr$ querying rule with $\arf$ assignment \dotfill & p. \pageref{qrarf} \\
    $\mathbf{DPF}$ &$\equiv$& Generic notation for an arbitrary dispatching policy \dotfill & p. \pageref{dpf} \\
    \midrule
    \multicolumn{4}{c}{\emph{Individual Dispatching Policies}}\\
    $\DP{\qr}{\ar}$ &$\equiv$& Dispatching policy using the $\qr$ querying rule and $\ar$ assignment rule \dotfill & p. \pageref{qrar} \\
    $\mathsf{DP}$ &$\equiv$& Generic notation for an individual dispatching policy \dotfill & p. \pageref{DP} \\
    $\ipoptd{\DP{\qrf}{\arf}}$ & $\equiv$& ``Optimal'' dispatching policy in $\DP{\qrf}{\arf}$ found by IPOPT \dotfill & p. \pageref{ipoptd}\\
    $\ipoptd{\qrf}$ &$\equiv$& Abbreviated notation for $\ipoptd{\DP{\qrf}{\cida}}$ \dotfill & p. \pageref{ipoptd}\\
     $\genseed$ &$\equiv$& Variant of the $\ipoptd{\genq}$ dispatching policy, where the associated optimization problem is ``seeded'' with the parameters of the $\ipoptd{\indq}$ policy \dotfill & p. \pageref{genseedd}\\
    \bottomrule
\end{tabular}
\end{table}

\newpage

\begin{table}[!htp]
\caption{Table \ref{app:notation}4: List of Notations}
\centering
\scriptsize
\begin{tabular}{>{\raggedright}p{0.06\textwidth}@{}>{\centering}p{0.05\textwidth}@{}p{0.77\textwidth}>{\raggedleft\arraybackslash}p{0.05\textwidth}}
    \toprule
    $\alpha_i(\avec, \dvec)$ &$\equiv$& Probability that the job is assigned to a class-$i$ server when $\Avec=\avec$ and $\Dvec=\dvec$ \dotfill & p. \pageref{alpha_i(a,d)}\\
    $\alpha_i(j, \dvec)$ &$\equiv$& Notation for $\alpha_i(\avec,\dvec)$ when $\Avec=\avec$ is such that the fastest idle queried server belongs to class $j$\dotfill & p. \pageref{alphaijd}\\
    $\Avec$ &$\equiv$& $(A_1, \dots, A_s)$; Random vector representing the number of queried idle servers of each class \dotfill & p. \pageref{Avec}\\
    $A_i$ &$\equiv$& Random variable representing the number of queried idle class-$i$ servers \dotfill & p. \pageref{A_i}\\
    $\Acal$ &$\equiv$& $\{\avec\colon a_1+\cdots+a_s\le d\}$; Set of all possible values of the random vector $\Avec$ (given a fixed $d$) \dotfill & p. \pageref{Acal}\\
    $\avec$ &$\equiv$& $(a_1, \dots, a_s)$; Realization of the random vector $\Avec$ \dotfill & p. \pageref{Avec}  \\
    $a_i$ &$\equiv$& Realization of the random variable $A_i$ \dotfill & p. \pageref{Dvec}\\
    $B_i$ &$\equiv$& Busy period duration at a class-$i$ server \dotfill & p. \pageref{B_i}\\
    $\bi$ &$\equiv$& $\prod_{\ell=1}^{i-1}\rho_{\ell}^{d_{\ell}}$; Probability that all queried servers faster than those of class-$i$ in are busy \dotfill & p. \pageref{eq:b_i(d)}\\
    $d$ &$\equiv$& Total number of servers to be queried \dotfill & p. \pageref{d}\\
    $\Dvec$ &$\equiv$& $(D_1, \dots, D_s)$; Random vector representing the number of queried servers of each class \dotfill & p. \pageref{Dvec}\\
    $D_i$ &$\equiv$& Random variable representing the number of queried class-$i$ servers \dotfill & p. \pageref{Dvec}\\
    $\Dcal$ &$\equiv$& $\{\dvec\colon d_1+\cdots+d_s=d\}$; Set of all possible values of the random vector $\Dvec$ (given a fixed $d$) \dotfill & p. \pageref{Dcal}\\
    $\dvec$ &$\equiv$& $(d_1, \dots, d_s)$; Realization of $\Dvec$ representing the class mix \dotfill & p. \pageref{dvec}  \\
    $d_i$ &$\equiv$& Realization of $D_i$ representing the number of class-$i$ servers in the query \dotfill & p. \pageref{Dvec}\\
    $\gamma(j, \dvec)$ &$\equiv$& Mapping where $\alpha_i(j,\dvec)=\alpha_i(j,\gamma(j,\dvec))$ under our assignment rule pruning for all $(i,j,\dvec)$ \dotfill & p. \pageref{gamma}\\
    $h(\dvec)$ &$\equiv$& $\min\{\ell\in\Scal\colon d_\ell>0\}$; The fastest class included in a query when $\Dvec=\dvec$ \dotfill & p. \pageref{h(d)}\\
    $J$ &$\equiv$& $\min\{j\in\Scal\colon A_{j}>0\}$ (with $\min\emptyset\equiv s+1$); Class of the fastest idle queried server \dotfill & p. \pageref{J}\\
    $\Jcal_i(\dvec)$ &$\equiv$& Set of classes $j > i$ for which $(i,j,\gamma(j,\dvec)) \in \Tcal$ \dotfill & p. \pageref{eq:Jid}\\ 
    $k$ &$\equiv$& Total number of servers \dotfill & p. \pageref{k}\\
    $k_i$ &$\equiv$& Number of class $i$ servers for $i\in\Scal$\dotfill & p. \pageref{k_i}\\
    $\lambda$ & $\equiv$& Overall mean arrival rate to a server\dotfill & p. \pageref{lambda} \\
    $\lambda_i$ &$\equiv$& Mean arrival rate to a class-$i$ server\dotfill & p. \pageref{li}\\
    $\lb$ &$\equiv$& Mean arrival rate to a busy class-$i$ server \dotfill &p. \pageref{lib}\\
    $\li$ &$\equiv$& Mean arrival rate to an idle class-$i$ server \dotfill &p. \pageref{lii}\\
    $\mu_i$ &$\equiv$& Speed of a class-$i$ server \dotfill & p. \pageref{mu_i}\\
    $\Pcal$ &$\equiv$& Set of $(j,\dvec)$ pairs that can form a triple with some $i \in \Scal$ so that $(i,j,\dvec) \in \Tcal$ \dotfill  & p.\pageref{eq:pcal}\\   
    $\Pcal(\dvec)$ &$\equiv$& Set of server classes $j$ that can form a triple with $\dvec$ and some $i \in \Scal$ so that $(i,j,\dvec) \in \Tcal$ \dotfill & p.\pageref{eq:pcal-d}\\   
    $p(\dvec)$ &$\equiv$& $\pr(\Dvec = \dvec)$; The probability that $\Dvec=\dvec$; the function $p(\cdot)$ uniquely specifies the querying rule \dotfill & p. \pageref{p}\\
    $\mathcal{Q}$ &$\equiv$ & $\{1,2,\ldots,d\}$; Set of indices for each queried server in a query \dotfill & p. \pageref{Qcal}\\
    $q_i$ &$\equiv$& $k_i/k$; Proportion of servers which belong to class $i$ for $i\in\Scal$\dotfill & p. \pageref{q_i}\\
    $\rho_i$ &$\equiv$& Fraction of time a class-$i$ server is busy \dotfill & p. \pageref{rho_i}\\
    $R_i$ &$\equiv$& $\mu_i/\mu_s$; Speed of class-$i$ servers normalized by that of the slowest (i.e., class-$s$) servers \dotfill & p. \pageref{r_i}\\
    $\rb$ &$\equiv$& Probability that a busy queried tagged class-$i$ server is assigned the job when $\Dvec=\dvec$\dotfill & p. \pageref{ribd}\\
    $\ri$ &$\equiv$& Probability that an idle queried tagged class-$i$ server is assigned the job when $\Dvec=\dvec$\dotfill & p. \pageref{riid}\\
    $s$ &$\equiv$& Number of server classes\dotfill & p. \pageref{s}\\
    $\Scal$ &$\equiv$& $\{ 1,\dots, s \}$; Set of server class indices \dotfill & p. \pageref{s}\\
    $\Scalb$ &$\equiv$& $\{ 1,\dots, s+1 \}$; Set of all possible values for the random variable $J$\dotfill & p. \pageref{sbar}\\
    $\Scal(\dvec)$ &$\equiv$& $\{ i \in \Scal \colon d_i > 0 \}$; Indices of server classes included in the query \dotfill & p. \pageref{eq:scal-d}\\
    $\Tcal$ &$\equiv$& Set of triples $(i,j,\dvec)$ for which each $\alpha_i(j,\dvec)$ can take a distinct value in our pruning\dotfill & p. \pageref{eq:tcal}\\
    $\Tcal(\dvec)$ &$\equiv$& Set of pairs $(i,j)$ (given $\dvec$) for which each $\alpha_i(j,\dvec)$ can take a distinct value in our pruning\dotfill & p. \pageref{eq:tcal-d}\\
    $T$ &$\equiv$& Response time of a job (not conditioned on the class of the server on which the job runs)\dotfill & p. \pageref{T}\\
    $T_i$ &$\equiv$& Response time of a job that runs at a class-$i$ server \dotfill & p. \pageref{T_i}\\
    \bottomrule
\end{tabular}
\end{table}

\newpage

\section{Appendix: Asymptotic Independence Verification}
\label{app:asymptotic}
In this appendix we verify the validity of the asymptotic independence assumption via simulation.  We simulate a system with $k$ servers for $k\in\{25, 75, 125, \ldots, 925, 975\}$ under the $\genseed$ dispatching policy at arrival rates $\lambda\in\{0.4,0.6,0.8\}$; service times are exponentially distributed with rates $(\mu_1,\mu_2,\mu_3)=(2,4/5,2/5)$ and class proportions are given by $(q_1,q_2,q_3)=(1/3,1/6,1/2)$.  Figure~\ref{sfig:exp_service} shows the simulated mean response times in comparison to our theoretical results under the asymptotic independence assumption (where $k\to\infty$). As $k$ grows large, the simulated results appear to converge to the theoretical results, suggesting that the asymptotic independence assumption indeed holds for our system (simulations for higher values of $k$, e.g., 2000, 3000, etc. further corroborate this finding).

Figure~\ref{sfig:hyperexp_service} show the corresponding results under hyperexponential service.  Specifically, the service time of jobs on a class-$i$ server is equally likely to be distributed according to an exponential distribution with rate $5\mu_i/2$ and one with rate $5\mu_i/8$, resulting in $C^2=1.72$; note that the average service time of a job running on a class-$i$ server remains $1/\mu_i$ (as was the case in the system with exponentially distributed service times).  Further note that the policy used is still ``optimized'' under the assumption of exponential service times (i.e., for any given value of $\lambda$, the parameters defining the querying and assignment rules are the same as those used above).  Our simulations (including those at higher $k$ values) again suggest that as $k$ grows the asymptotic independence assumption holds.

\begin{figure}[h]
    \begin{subfigure}{0.45\textwidth}
        \scalebox{0.65}{
\renewcommand{\legendypos}{300}
\renewcommand{\labscal}{1.3}
\renewcommand{\lthickness}{1.2}
\begin{tikzpicture}[x=1pt,y=1pt]


\path[use as bounding box,fill=white,fill opacity=0.00] (0,0) rectangle (350, 300);
\begin{scope}
\path[genjsqstyle,line width=1.8pt,dotted,line join=round] ( 53.74,115.01) -- (301.82,115.01);

\path[gensedstyle,line width=1.8pt,dotted,line join=round] ( 53.74,152.28) -- (301.82,152.28);

\path[gensewstyle,line width=1.8pt,dotted,line join=round] ( 53.74,228.21) -- (301.82,228.21);

\path[genjsqstyle] ( 53.74,117.54) --
	( 66.80,116.07) --
	( 79.86,115.18) --
	( 92.91,116.01) --
	(105.97,115.48) --
	(119.03,115.05) --
	(132.08,115.28) --
	(145.14,115.35) --
	(158.20,115.17) --
	(171.25,115.40) --
	(184.31,115.21) --
	(197.37,115.17) --
	(210.43,115.08) --
	(223.48,115.02) --
	(236.54,115.05) --
	(249.60,115.09) --
	(262.65,115.15) --
	(275.71,115.10) --
	(288.77,115.10) --
	(301.82,115.20);

\path[gensedstyle] ( 53.74,156.43) --
	( 66.80,153.64) --
	( 79.86,151.93) --
	( 92.91,153.70) --
	(105.97,152.48) --
	(119.03,152.86) --
	(132.08,152.73) --
	(145.14,152.68) --
	(158.20,152.30) --
	(171.25,152.51) --
	(184.31,152.30) --
	(197.37,152.39) --
	(210.43,152.09) --
	(223.48,152.74) --
	(236.54,152.01) --
	(249.60,152.70) --
	(262.65,152.32) --
	(275.71,152.26) --
	(288.77,152.28) --
	(301.82,152.04);

\path[gensewstyle] ( 53.74,242.64) --
	( 66.80,232.97) --
	( 79.86,228.74) --
	( 92.91,231.48) --
	(105.97,231.45) --
	(119.03,228.32) --
	(132.08,229.74) --
	(145.14,229.46) --
	(158.20,226.94) --
	(171.25,229.93) --
	(184.31,229.52) --
	(197.37,229.08) --
	(210.43,228.50) --
	(223.48,228.78) --
	(236.54,227.85) --
	(249.60,230.66) --
	(262.65,228.94) --
	(275.71,228.53) --
	(288.77,228.17) --
	(301.82,228.06);

\end{scope}
\begin{scope}
\path[draw=gray,line width=0.8*\lthickness pt,dash pattern=] (47.21, 42.18) -- (47.21, 280);
\node[text=black,anchor=base east,inner sep=0pt, outer sep=0pt, scale= \labscal]at ( 42.21, 39.15)  {0.0};

\node[text=black,anchor=base east,inner sep=0pt, outer sep=0pt, scale= \labscal] at ( 42.21, 96.63) {0.5};

\node[text=black,anchor=base east,inner sep=0pt, outer sep=0pt, scale= \labscal] at ( 42.21,154.10) {1.0};

\node[text=black,anchor=base east,inner sep=0pt, outer sep=0pt, scale= \labscal] at ( 42.21,211.58) {1.5};

\node[text=black,anchor=base east,inner sep=0pt, outer sep=0pt, scale= \labscal] at ( 42.21,269.05) {2.0};

\path[draw=black,line width= 0.8*\lthickness pt,line join=round] ( 44, 42.18) --
	( 47.21, 42.18);

\path[draw=black,line width= 0.8*\lthickness pt,line join=round] ( 44, 99.66) --
	( 47.21, 99.66);

\path[draw=black,line width= 0.8*\lthickness pt,line join=round] ( 44,157.13) --
	( 47.21,157.13);

\path[draw=black,line width= 0.8*\lthickness pt,line join=round] ( 44,214.61) --
	( 47.21,214.61);

\path[draw=black,line width= 0.8*\lthickness pt,line join=round] ( 44,272.08) --
	( 47.21,272.08);
\path[draw=gray,line width=0.8*\lthickness pt,dash pattern=] (43.93, 42.18) -- (315, 42.18);
\path[draw=black,line width= 0.8*\lthickness pt,line join=round] ( 47.21, 39) --
	( 47.21, 42.18);

\path[draw=black,line width= 0.8*\lthickness pt,line join=round] ( 112.50, 39) --
	( 112.50, 42.18);

\path[draw=black,line width= 0.8*\lthickness pt,line join=round] (177.78, 39) --
	(177.78, 42.18);

\path[draw=black,line width= 0.8*\lthickness pt,line join=round] (243.07, 39) --
	(243.07, 42.18);

\path[draw=black,line width= 0.8*\lthickness pt,line join=round] (308.35, 39) --
	(308.35, 42.18);




\node[text=black,anchor=base,inner sep=0pt, outer sep=0pt, scale= \labscal] at ( 47.21, 29) {0};

\node[text=black,anchor=base,inner sep=0pt, outer sep=0pt, scale= \labscal] at ( 112.50, 29) {250};

\node[text=black,anchor=base,inner sep=0pt, outer sep=0pt, scale= \labscal] at (177.78, 29) {500};

\node[text=black,anchor=base,inner sep=0pt, outer sep=0pt, scale= \labscal] at (243.07, 29) {750};

\node[text=black,anchor=base,inner sep=0pt, outer sep=0pt, scale= \labscal] at (308.35, 29) {1000};
\end{scope}
\begin{scope}
\node[text=black,anchor=base,inner sep=0pt, outer sep=0pt, scale=  \labscal] at (\legendxpos/2+20,  8.22) {$k$ (number of servers)};
\node[text=black,rotate= 90.00,anchor=base,inner sep=0pt, outer sep=0pt, scale= \labscal] at ( 12,0.5*\legendypos) {$\mathbb E[T]$};
\end{scope}

\end{tikzpicture}}
        \caption{Exponential Service\label{sfig:exp_service}}
    \end{subfigure}
    \begin{subfigure}{0.45\textwidth}
        \scalebox{0.65}{
\renewcommand{\legendypos}{350}
\renewcommand{\labscal}{1.3}
\renewcommand{\lthickness}{1.2}
\begin{tikzpicture}[x=1pt,y=1pt]


\path[use as bounding box,fill=white,fill opacity=0.00] (0,0) rectangle (350, 300);
\begin{scope}
\path[genjsqstyle,line width= 1.8pt,dotted,line join=round] ( 53.74,119.38) -- (301.82,119.38);

\path[gensedstyle,line width= 1.8pt,dotted,line join=round] ( 53.74,162.23) -- (301.82,162.23);

\path[gensewstyle,line width= 1.8pt,dotted,line join=round] ( 53.74,256.64) -- (301.82,256.64);

\path[genjsqstyle] ( 53.74,123.49) --
	( 66.80,121.27) --
	( 79.86,119.72) --
	( 92.91,120.29) --
	(105.97,120.35) --
	(119.03,119.99) --
	(132.08,119.55) --
	(145.14,120.07) --
	(158.20,119.29) --
	(171.25,119.65) --
	(184.31,119.78) --
	(197.37,119.69) --
	(210.43,119.79) --
	(223.48,119.43) --
	(236.54,119.20) --
	(249.60,119.75) --
	(262.65,119.34) --
	(275.71,119.54) --
	(288.77,119.54) --
	(301.82,119.41);

\path[gensedstyle] ( 53.74,169.37) --
	( 66.80,165.71) --
	( 79.86,162.74) --
	( 92.91,163.82) --
	(105.97,163.26) --
	(119.03,162.01) --
	(132.08,162.61) --
	(145.14,162.29) --
	(158.20,161.32) --
	(171.25,162.81) --
	(184.31,162.65) --
	(197.37,162.75) --
	(210.43,162.55) --
	(223.48,162.58) --
	(236.54,162.48) --
	(249.60,163.02) --
	(262.65,162.62) --
	(275.71,162.45) --
	(288.77,162.02) --
	(301.82,162.34);

\path[gensewstyle] ( 53.74,267.15) --
	( 66.80,264.20) --
	( 79.86,256.55) --
	( 92.91,264.51) --
	(105.97,257.38) --
	(119.03,261.71) --
	(132.08,258.74) --
	(145.14,259.05) --
	(158.20,254.68) --
	(171.25,256.78) --
	(184.31,259.04) --
	(197.37,257.09) --
	(210.43,259.30) --
	(223.48,258.46) --
	(236.54,257.29) --
	(249.60,256.13) --
	(262.65,257.72) --
	(275.71,257.84) --
	(288.77,257.95) --
	(301.82,258.55);

\end{scope}
\begin{scope}
\path[draw=gray,line width=0.8*\lthickness pt,dash pattern=] (47.21, 42.18) -- (47.21, 280);
\node[text=black,anchor=base east,inner sep=0pt, outer sep=0pt, scale= \labscal]at ( 42.21, 39.15)  {0.0};

\node[text=black,anchor=base east,inner sep=0pt, outer sep=0pt, scale= \labscal] at ( 42.21, 96.63) {0.5};

\node[text=black,anchor=base east,inner sep=0pt, outer sep=0pt, scale= \labscal] at ( 42.21,154.10) {1.0};

\node[text=black,anchor=base east,inner sep=0pt, outer sep=0pt, scale= \labscal] at ( 42.21,211.58) {1.5};

\node[text=black,anchor=base east,inner sep=0pt, outer sep=0pt, scale= \labscal] at ( 42.21,269.05) {2.0};

\path[draw=black,line width= 0.8*\lthickness pt,line join=round] ( 44, 42.18) --
	( 47.21, 42.18);

\path[draw=black,line width= 0.8*\lthickness pt,line join=round] ( 44, 99.66) --
	( 47.21, 99.66);

\path[draw=black,line width= 0.8*\lthickness pt,line join=round] ( 44,157.13) --
	( 47.21,157.13);

\path[draw=black,line width= 0.8*\lthickness pt,line join=round] ( 44,214.61) --
	( 47.21,214.61);

\path[draw=black,line width= 0.8*\lthickness pt,line join=round] ( 44,272.08) --
	( 47.21,272.08);
\path[draw=gray,line width=0.8*\lthickness pt,dash pattern=] (43.93, 42.18) -- (315, 42.18);
\path[draw=black,line width= 0.8*\lthickness pt,line join=round] ( 47.21, 39) --
	( 47.21, 42.18);

\path[draw=black,line width= 0.8*\lthickness pt,line join=round] ( 112.50, 39) --
	( 112.50, 42.18);

\path[draw=black,line width= 0.8*\lthickness pt,line join=round] (177.78, 39) --
	(177.78, 42.18);

\path[draw=black,line width= 0.8*\lthickness pt,line join=round] (243.07, 39) --
	(243.07, 42.18);

\path[draw=black,line width= 0.8*\lthickness pt,line join=round] (308.35, 39) --
	(308.35, 42.18);




\node[text=black,anchor=base,inner sep=0pt, outer sep=0pt, scale= \labscal] at ( 47.21, 29) {0};

\node[text=black,anchor=base,inner sep=0pt, outer sep=0pt, scale= \labscal] at ( 112.50, 29) {250};

\node[text=black,anchor=base,inner sep=0pt, outer sep=0pt, scale= \labscal] at (177.78, 29) {500};

\node[text=black,anchor=base,inner sep=0pt, outer sep=0pt, scale= \labscal] at (243.07, 29) {750};

\node[text=black,anchor=base,inner sep=0pt, outer sep=0pt, scale= \labscal] at (308.35, 29) {1000};
\end{scope}
\begin{scope}
\node[text=black,anchor=base,inner sep=0pt, outer sep=0pt, scale=  \labscal] at (\legendxpos/2+20,  8.22) {$k$ (number of servers)};
\node[text=black,rotate= 90.00,anchor=base,inner sep=0pt, outer sep=0pt, scale= \labscal] at ( 12,0.5*\legendypos) {$\mathbb E[T]$};
\end{scope}














\end{tikzpicture}}
        \caption{Hyperexponential Service ($C^2 = 1.72$)\label{sfig:hyperexp_service}}
    \end{subfigure}
    \newline
    \centering
    \begin{subfigure}{\textwidth}
        \scalebox{0.65}{\renewcommand{\legendypos}{100}
\renewcommand{\labscal}{1.3}
\renewcommand{\lthickness}{1.2}

\begin{tikzpicture}[x=1pt,y=1pt]

\path[use as bounding box,fill=white,fill opacity=0.00] (0,0) rectangle (600, \legendypos);

\begin{scope}
\path[genjsqstyle,line width=2.0pt, dotted] (100, 0.3*\legendypos) -- (130,0.3*\legendypos);

\node[text=black,anchor=base west,inner sep=0pt, outer sep=0pt, scale= \labscal] at (140,0.3*\legendypos-2) {{$\lambda = 0.4$ (theoretical)}};

\path[genjsqstyle, line width=1.0pt] (100, 0.6*\legendypos) -- (130,0.6*\legendypos);

\node[text=black,anchor=base west,inner sep=0pt, outer sep=0pt, scale= \labscal] at (140,0.6*\legendypos-2) {{$\lambda = 0.4$ (simulated)}};
\end{scope}

\begin{scope}
\path[gensedstyle,line width=2.0pt, dotted] (285,0.3*\legendypos) -- (315,0.3*\legendypos);

\node[text=black,anchor=base west,inner sep=0pt, outer sep=0pt, scale= \labscal] at (325,0.3*\legendypos-2) {{$\lambda = 0.6$ (theoretical)}};

\path[gensedstyle, line width=1.0pt] (285,0.6*\legendypos) -- (315,0.6*\legendypos);

\node[text=black,anchor=base west,inner sep=0pt, outer sep=0pt, scale= \labscal] at (325,0.6*\legendypos-2) {{$\lambda = 0.6$ (simulated)}};
\end{scope}

\begin{scope}
\path[gensewstyle,line width=2.0pt, dotted] (470,0.3*\legendypos) -- (500,0.3*\legendypos);

\node[text=black,anchor=base west,inner sep=0pt, outer sep=0pt, scale= \labscal] at (510,0.3*\legendypos-2) {{$\lambda = 0.8$ (theoretical)}};

\path[gensewstyle, line width=1.0pt] (470,0.6*\legendypos) -- (500,0.6*\legendypos);

\node[text=black,anchor=base west,inner sep=0pt, outer sep=0pt, scale= \labscal] at (510,0.6*\legendypos-2) {{$\lambda = 0.8$ (simulated)}};
\end{scope}

\end{tikzpicture}}
    \end{subfigure}
    \caption{Simulated mean response times for varying values of $k$ (together with the theoretical values for $k\to\infty$ under the asymptotic independence assumption) for a system with $(q_1, q_2, q_3) = (1/3, 1/6, 1/2)$ and mean service times $(1/\mu_1, 1/\mu_2, 1/\mu_3) = (1/2, 5/4, 5/2)$ under the $\genseed$;  Figures \ref{sfig:exp_service} and \ref{sfig:hyperexp_service} are in the settings with exponentially ($C^2=1$) and hyperexponentially ($C^2 = 1.72$) distributed service times, respectively.\label{fig:asymp-indep}  
    } 
    \label{fig:Asymptotic-indep}
\end{figure}
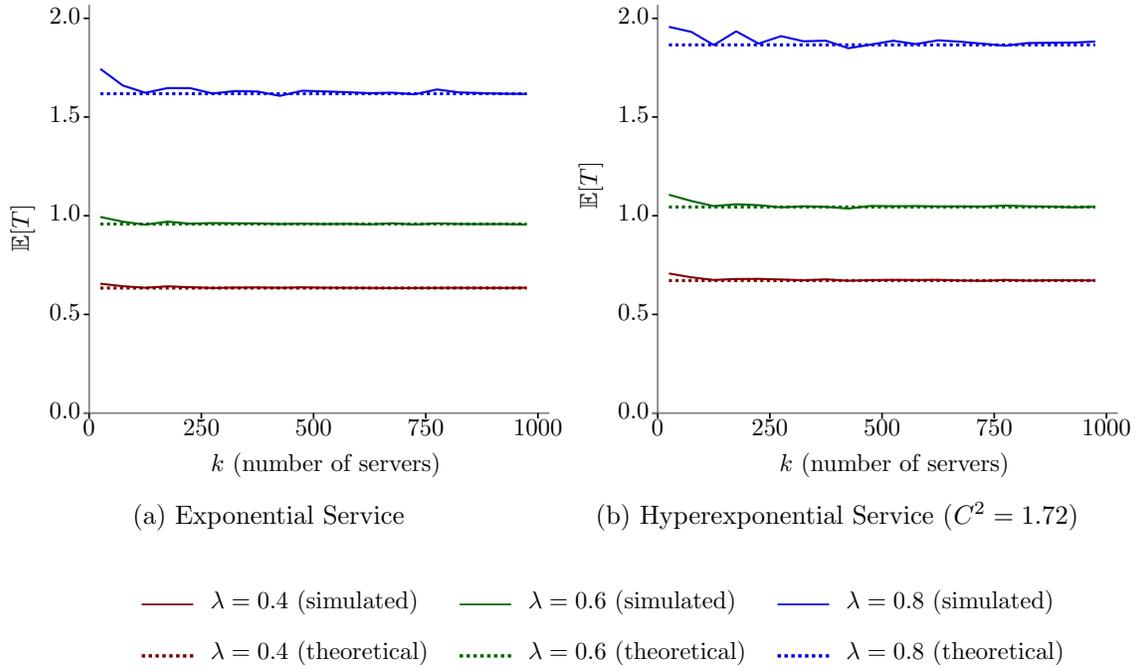

\newpage

\section{Appendix: Additional optimization problems}
\label{app:opt}

\begin{tcolorbox}

\textbf{Optimization Problem for $\DP{\iidq}{\cida}$}

Given values of $s$, $d$, $\lambda$, and $\mu_i$ and $q_i$ (both $\forall i\in\Scal$),  determine the \textbf{nonnegative} values of the decision variables $\li$, $\lb$, and $\indp(i)$ (each $\forall i\in\Scal$) and $\alpha_i(j,\dvec)$ ($\forall (i,j,\dvec)\in\Tcal$) that solve the following nonlinear program:
{\small
\begin{align*}
    \min\quad&\frac{1}{\lambda}\sum_{i=1}^s
    q_i\left(\frac{(1-\rho_i)\li+\rho_i\lb}{\mu_i-\lb}\right)\\
    \mathrm{s.t.}\quad&  \li=\frac{d!\lambda}{q_i}\sum_{\dvec\in\Dcal}\left\{\bi \left(\prod_{\ell=1}^s\frac{\indp(\ell)^{d_\ell}}{d_\ell!}\right)\alpha_i(i,\gamma(i,\dvec))\sum_{a_i=1}^{d_i}\binom{d_i-1}{a_i-1}\dfrac{(1-\rho_i)^{a_i-1}\rho_i^{d_i-a_i}}{a_i}\right\}&(\forall i\in\Scal)\\
    &\lb=\frac{d!\lambda}{q_i\rho_i}\sum_{\dvec\in\Dcal}\left\{ \left(\prod_{\ell=1}^s\frac{\indp(\ell)^{d_\ell}}{d_\ell!}\right)\ \sum_{\mathclap{j\in\Jcal_{i}(\dvec)}}\bj\left(1-\rho_j^{d_j}\right)\alpha_i(j,\gamma(i,\dvec))\right\}&(\forall i\in\Scal)\\
    &\lb <\mu_i&(\forall i\in\Scal)\\
    &\sum_{i=1}^s \indp(i)=1\\
    &\sum_{\mathclap{\substack{i\in\Scal:\\ (i,j,\dvec)\in\Tcal}}}\alpha_i(j,\dvec)=1&(\forall (j,\dvec)\in\Pcal)
\end{align*}
}%
where for all $i\in\Scal$ in writing $\rho_i$ we are denoting the expression $\li/\left(\mu_i-\lb+\li\right)$ with $\rsd\equiv0$ for all $\dvec\in\Dvec$, and for all $j\in\Scalb$ in writing $\bj$ we are denoting the expression $\prod_{\ell=1}^{j-1}\left(\lil/\left(\mu_\ell-\lbl+\lil\right)\right)^{d_\ell}$.

The $\indp(i)$ values (for all $i\in\Scal$) and the $\alpha_i(j,\dvec)$ values (for all $(i,j,\dvec)\in\Tcal$ together with $\alpha_i(j,\dvec)=0$ for all $(i,j,\dvec)\not\in\Tcal$) from an optimal solution specify the querying and assignment rules of an optimal policy, respectively.
\end{tcolorbox}


\begin{tcolorbox}
\textbf{Optimization Problem for $\DP{\indq}{\cida}$}

Given values of $s$, $d$, $\lambda$, and $\mu_i$ and $q_i$ (both given $\forall i\in\Scal$),  determine the \textbf{nonnegative} values of the decision variables $\li$ and $\lb$ ($\forall i\in\Scal$), $\indp_u(i)$ ($\forall (u,i)\in\Qcal\times\Scal$), and $\alpha_i(j,\dvec)$ ($\forall (i,j,\dvec)\in\Tcal$) that solve the following nonlinear program:
{\scriptsize
\begin{align*}
    \min\quad&\frac{1}{\lambda}\sum_{i=1}^s q_i\left(\frac{(1-\rho_i)\li+\rho_i\lb}{\mu_i-\lb}\right)\\
    \mathrm{s.t.}\quad&  \li=\frac{\lambda}{q_i}\sum_{\dvec\in\Dcal}\left\{d_i\bi \left(\sum_{\vec{\Qcal}\in\Bcal(\dvec)}\prod_{i=1}^s\prod_{u\in\Qcal_i}\indp_u(i)\right)\alpha_i(i,\gamma(i,\dvec))\sum_{a_i=1}^{d_i}\binom{d_i-1}{a_i-1}\dfrac{(1-\rho_i)^{a_i-1}\rho_i^{d_i-a_i}}{a_i}\right\}&(\forall i\in\Scal)\\
    &\lb=\frac{\lambda}{q_i\rho_i}\sum_{\dvec\in\Dcal}\left\{\left(\sum_{\vec{\Qcal}\in\Bcal(\dvec)}\prod_{i=1}^s\prod_{u\in\Qcal_i}\indp_u(i)\right)\ \  \sum_{\mathclap{j\in\Jcal_{i}(\dvec)}}\bj\left(1-\rho_j^{d_j}\right)\alpha_i(j,\gamma(i,\dvec))\right\}&(\forall i\in\Scal)\\
    &\lb <\mu_i&(\forall i\in\Scal)\\
    &\sum_{\dvec\in\Dcal}\indp_u(\dvec)=1&(\forall u\in\Qcal)\\
    &\sum_{\mathclap{\substack{i\in\Scal:\\ (i,j,\dvec)\in\Tcal}}}\alpha_i(j,\dvec)=1&(\forall (j,\dvec)\in\Pcal)
\end{align*}
}%
where for all $i\in\Scal$ in writing $\rho_i$ we are denoting the expression $\li/\left(\mu_i-\lb+\li\right)$ with $\rsd\equiv0$ for all $\dvec\in\Dvec$, and for all $j\in\Scalb$ in writing $\bj$ we are denoting the expression $\prod_{\ell=1}^{j-1}\left(\lil/\left(\mu_\ell-\lbl+\lil\right)\right)^{d_\ell}$.

The $\indp_u(i)$ values (for all $(u,i)\in\Qcal\times\Scal$) and the $\alpha_i(j,\dvec)$ values (for all $(i,j,\dvec)\in\Tcal$ together with $\alpha_i(j,\dvec)=0$ for all $(i,j,\dvec)\not\in\Tcal$) from an optimal solution specify the querying and assignment rules of an optimal policy, respectively.
\end{tcolorbox}


\begin{tcolorbox}
\textbf{Optimization Problem for $\DP{\detq}{\cida}$}

For each $\dvec\in\Dcal$ solve the following subproblem:
\begin{itemize}
    \item[]
{\scriptsize
Given values of $s$, $\dvec=(d_1,d_2,\ldots,d_s)$, $\lambda$, and $\mu_i$ and $q_i$ (both $\forall i\in\Scal(\dvec)$),  determine the \textbf{nonnegative} values of the decision variables $\li$ and $\lb$ (both $\forall i\in\Scal(\dvec)$) and $\alpha_i(j,\dvec)$ ($\forall (i,j)\in\Tcal(\dvec)$) that solve the following nonlinear program:
\begin{align*}
    \min\quad&\frac{1}{\lambda}\sum_{i=1}^s q_i\left(\frac{(1-\rho_i)\li+\rho_i\lb}{\mu_i-\lb}\right)\\
    \mathrm{s.t.}\quad&  \li=\frac{\lambda d_ib_i(\dvec)\alpha_i(i,\dvec)}{q_i}\sum_{a_i=1}^{d_i}\binom{d_i-1}{a_i-1}\dfrac{(1-\rho_i)^{a_i-1}\rho_i^{d_i-a_i}}{a_i}&(\forall i\in\Scal(\dvec))\\
    &\lb=\frac{\lambda}{q_i\rho_i}\sum_{j\in\Jcal_i(\dvec)}b_j(\dvec)\left(1-\rho_j^{d_j}\right)\alpha_i(j,\dvec)&(\forall i\in\Scal(\dvec))\\
    &\lb <\mu_i&(\forall i\in\Scal(\dvec))\\
    &\sum_{\mathclap{\substack{i\in\Scal:\\ (i,j)\in\Tcal(\dvec)}}}\alpha_i(j,\dvec)=1&(\forall j\in\Pcal(\dvec))
\end{align*}
where for all $i\in\Scal$ in writing $\rho_i$ we are denoting the expression $\li/\left(\mu_i-\lb+\li\right)$ with $\rho_{s+1}^{d_{s+1}}\equiv0$, and for all $j\in\Scalb$ in writing $b_i(\dvec)$ we are denoting the expression $\prod_{\ell=1}^{i-1}\left(\lil/\left(\mu_\ell-\lbl+\lil\right)\right)^{d_{\ell}}$.
}
\end{itemize}

Then identify $\dvec^{*}$, the query mix that yielded the lowest subproblem objective function value among all query mixes $\dvec\in\Dcal$.  This query mix $\dvec^*$ specifies the querying rule associated with $\DP{\detq}{\cida}$, while the $\alpha_i(j,\dvec^*)$ values (for all $(i,j)\in\Tcal(\dvec^*)$ together with $\alpha_i(j,\dvec^*)=0$ for all $(i,j)\not\in\Tcal(\dvec^*)$ and $\alpha_i(j,\dvec)=0$ for all $(i,j,\dvec)\in\Scal\times\Scalb\times\left(\Dcal\backslash\left\{\dvec^*\right\}\right)$ from a solution to the subproblem associated with $\dvec^*$ specifies the assignment rule associated with an optimal policy.
\end{tcolorbox}


\begin{tcolorbox}
\textbf{Optimization Problem for $\DP{\srcq}{\cida}$}

Given values of $s$, $d$, $\lambda$, and $\mu_i$ and $q_i$ (both given $\forall i\in\Scal$),  determine the \textbf{nonnegative} values of the decision variables $\li$ and $\lb$ (both $\forall i\in\Scal$), and $\srcp(i)$ ($\forall i\in\Scal$) that solve the following nonlinear program:
\begin{align*}
    \min\quad&\frac{1}{\lambda}\sum_{i=1}^s q_i\left(\frac{(1-\rho_i)\li+\rho_i\lb}{\mu_i-\lb}\right)\\
    \mathrm{s.t.}\quad&  \li=\frac{\lambda d\srcp(i)}{q_i}\sum_{a_i=1}^{d}\binom{d-1}{a_i-1}\dfrac{(1-\rho_i)^{a_i-1}\rho_i^{d-a_i}}{a_i}&(\forall i\in\Scal)\\
    &\lb=\frac{\lambda\srcp(i)\rho_i^{d-1}}{q_i}&(\forall i\in\Scal)\\
    &\lb<\mu_i&(\forall i\in\Scal)\\
    &\sum_{i=1}^s \srcp(i)=1\\
\end{align*}
where for all $i\in\Scal$ in writing $\rho_i$ we are denoting the expression $\li/\left(\mu_i-\lb+\li\right)$.

The $\srcp(i)$ values (for all $i\in\Dcal$) from an optimal solution specifies the querying rule associated with an optimal policy.
\end{tcolorbox}


\begin{tcolorbox}
\textbf{Optimization Problem for $\DP{\sfcq}{\cida}$}

For each $i\in\Scal$ solve the following subproblem:
\begin{itemize}
    \item[]
{\footnotesize
Given values of $s$, $\lambda$, $\mu_i$, and $q_i$ determine the \textbf{nonnegative} values of $\li$ and $\lb$ that solve the following nonlinear program (noting that the objective function below is irrelevant as there is at most one feasible solution):
\begin{align*}
    \min\quad&\frac{q_i}{\lambda} \left(\frac{(1-\rho_i)\li+\rho_i\lb}{\mu_i-\lb}\right)\\
    \mathrm{s.t.}\quad&  \li=\frac{\lambda d}{q_i}\sum_{a_i=1}^{d}\binom{d-1}{a_i-1}\dfrac{(1-\rho_i)^{a_i-1}\rho_i^{d-a_i}}{a_i}\\
    &\lb=\frac{\lambda\rho_i^{d-1}}{q_i}\\
    &\lb<\mu_i
\end{align*}
where for all $i\in\Scal$ in writing $\rho_i$ we are denoting the expression $\li/\left(\mu_i-\lb+\li\right)$ with $\rho_{s+1}^{d_{s+1}}\equiv0$, and for all $j\in\Scalb$ in writing $b_i(\dvec)$ we are denoting the expression $\prod_{\ell=1}^{i-1}\left(\lil/\left(\mu_\ell-\lbl+\lil\right)\right)^{d_{\ell}}$.
}
\end{itemize}
Then identify $i^{*}$, the class (or any class) that yielded the lowest subproblem objective function value among all classes $i\in\Scal$.  An optimal querying rule for the $\DP{\detq}{\cida}$ dispatching policy is that which always queries $d$ class-$i^{*}$ servers, i.e., the querying defined by $p(\dvec)=d\cdot I\{\dvec=d\evec_{i^{*}}\}$.

\end{tcolorbox}


\begin{tcolorbox}

\textbf{Optimization Problem for $\DP{\qr}{\cida}$}

Given values of $s$, $d$, $\lambda$, and $\mu_i$ and $q_i$ (both given $\forall i\in\Scal$) and some specific querying rule $\qr$ defined by some function $p\colon\Dcal\to[0,1]$ so that $p(\dvec)=\pr(\Dvec=\dvec)$  determine the \textbf{nonnegative} values of the decision variables $\li$ and $\lb$ (both $\forall i\in\Scal$) and $\alpha_i(j,\dvec)$ ($\forall (i,j,\dvec)\in\Tcal$) that solves the following nonlinear program:
\begin{align*}
    \min\quad&\frac{1}{\lambda}\sum_{i=1}^s q_i\left(\frac{(1-\rho_i)\li+\rho_i\lb}{\mu_i-\lb}\right)\\
    \mathrm{s.t.}\quad&  \li=\frac{\lambda}{q_i}\sum_{\dvec\in\Dcal}\left\{d_i\bi p(\dvec)\alpha_i(i,\gamma(i,\dvec))\sum_{a_i=1}^{d_i}\binom{d_i-1}{a_i-1}\dfrac{(1-\rho_i)^{a_i-1}\rho_i^{d_i-a_i}}{a_i}\right\}&(\forall i\in\Scal)\\
    &\lb=\frac{\lambda}{q_i\rho_i}\sum_{\dvec\in\Dcal}\left\{p(\dvec)\sum_{\mathclap{j\in\Jcal_{i}(\dvec)}}\bj\left(1-\rho_j^{d_j}\right)\alpha_i(j,\gamma(i,\dvec))\right\}&(\forall i\in\Scal)\\
    &\lb<\mu_i&(\forall i\in\Scal)\\
    &\sum_{\mathclap{\substack{i\in\Scal:\\ (i,j,\dvec)\in\Tcal}}}\alpha_i(j,\dvec)=1&(\forall (j,\dvec)\in\Pcal)
\end{align*}
where for all $i\in\Scal$ in writing $\rho_i$ we are denoting the expression $\li/\left(\mu_i-\lb+\li\right)$ with $\rsd\equiv1$ for all $\dvec\in\Dvec$, and for all $j\in\Scalb$ in writing $\bj$ we are denoting the expression $\prod_{\ell=1}^{j-1}\left(\lil/\left(\mu_\ell-\lbl+\lil\right)\right)^{d_\ell}$.

The $\alpha_i(j,\dvec)$ values (for all $(i,j,\dvec)\in\Tcal$ together with $\alpha_i(j,\dvec)=0$ for all $(i,j,\dvec)\not\in\Tcal$) from an optimal solution specifies the assignment rule of an optimal policy.
\end{tcolorbox}

\newpage

\section{Problem size}
\label{sec:problem-size}

In this section we address the sizes of the optimization problems presented in the preceding section as a theoretical proxy for problem tractability.  A more practical measure of tractability (computer runtime) will be presented in the following section, which focuses on numerical results.

\subsection{Measures of problem sizes}\label{sec:measures}

We measure problem sizes in terms of the number of variables (VAR), the number of linear equality constraints (LEC), and the number of nonlinear equality constraints (NEC).  We omit nonnegativity constraints, as there is one for each variable.  We also omit the number of upper-bound constraints (UBC) in all of our problems, as these always correspond to half the number of NECs. This is because there is always exactly one UBC for each $\lb$ variable and always exactly one NEC for each  $\li$ and each $\lb$ variable (which are equally numerous).  In the case of the problems associated with the $\detq$ and $\sfcq$ families we report the problem size associated with each \emph{subproblem} and the number of subproblems overall (SP).  Note that for $\detq$ not all subproblems are the same size, so we report all numbers in terms of the query mix, $\dvec\in\Dcal$, associated with the subproblem.  We report these problem size measures in Table~1; each entry is determined by straightforward inspection of the relevant optimization problem (as presented in Section~\ref{sec:optimization} or Appendix~\ref{app:opt}).

\begin{table}[!htp]
\label{tbl:formula}
\begin{center}
\begin{tabular}{|c|c|c|c|c|c|c|}
\cline{3-7}
\multicolumn{2}{c|}{} &
  VAR &
  \begin{tabular}[c]{@{}c@{}}LEC\end{tabular} &
  \begin{tabular}[c]{@{}c@{}}NEC\end{tabular} &
  DIM &
  \begin{tabular}[c]{@{}c@{}}SP\end{tabular} \\
  \hline \cline{1-7} 
\multicolumn{2}{|c|}{$\genq$}    & $2|\Scal|+|\Dcal|+|\Tcal|$ & $|\Pcal|+1$  & $2|\Scal|$   & $|\Dcal|+|\Tcal|-|\Pcal|-1$  & 1  \\ \hline
\multicolumn{2}{|c|}{$\indq$}                            & $(|\Qcal|+2)|\Scal|+|\Tcal|$  & $|\Pcal|+|\Qcal|$  & $2|\Scal|$  & $d|\Scal|+|\Tcal|-|\Pcal|-|\Qcal|$    & $1$  \\ \hline
\multicolumn{2}{|c|}{$\iidq$}     & $3|\Scal|+|\Tcal|$  & $|\Pcal|+1$  & $2|\Scal|$  & $|\Scal|+|\Tcal|-|\Pcal|-1$    & $1$  \\ \hline
\multicolumn{2}{|c|}{$\detq$} &
  
  \multicolumn{1}{c|}{$2|\Scal(\dvec)|+|\Tcal(\dvec)|$} &
  \multicolumn{1}{c|}{$|\Pcal(\dvec)|$} &
  \multicolumn{1}{c|}{$2|\Scal(\dvec)|$} &
  \multicolumn{1}{c|}{$|\Tcal(\dvec)|-|\Pcal(\dvec)|$} & 
  \multicolumn{1}{c|}{$|\Dcal|$} \\ \hline
\multicolumn{2}{|c|}{$\srcq$}    & $3|\Scal|$  & $1$  & $2|\Scal|$  & $1$  & $1$  \\ \hline
\multicolumn{2}{|c|}{$\sfcq$}                                & $2$ & $0$  & $2$  & $0$  & \multicolumn{1}{c|}{$|\Scal|$}  \\  \hline
\multicolumn{2}{|c|}{$\qr$} & $2|\Scal|+|\Tcal|$  &$|\Pcal|$  & $2|\Scal|$  & $|\Tcal|-|\Pcal|$  & $1$  \\ \hline
\end{tabular}
\caption{Optimization problem sizes for $\DP{\qrf}{\cida}$ (for six querying rule families $\qrf$) and for $\DP{\qr}{\cida}$ (for an individual querying rule $\qr$, e.g., $\brq$ or $\uniq$) in terms of the number of variables (VAR), the number of linear equality constraints (LEC), the number of nonlinear equality constraints (NEC), and the number of parameter dimensions of the polytope of all feasible dispatching policies (DIM). For the two querying rule families resulting in multiple subproblems (i.e., $\detq$ and $\sfcq$) the aforementioned quantities associated with each subproblem are given (as a function of $\dvec\in\Dcal$ in the case of $\detq$).  The number of subproblems (SP) in each optimization problem are also provided.}
\end{center}
\end{table}

We also introduce a measure which captures the dimension (DIM) of the polytope of parameter vectors specifying dispatching policies in Euclidean space.  We think of this dimension as corresponding to the number of ``degrees of freedom'' when specifying the querying ($p$) and assignment ($\alpha_i$) rules.  For example, when $s=3$, for any $(j, \dvec)\in\Pcal$, the three variables $\alpha_1(j,\dvec)$, $\alpha_2(j,\dvec)$, and $\alpha_3(j,\dvec)$ collectively contribute only two degrees of freedom, because the constraint $\alpha_1(j,\dvec)+\alpha_2(j,\dvec)+\alpha_3(j,\dvec)=1$ causes any one of these variables to always be uniquely determined in terms of the other two. Consequently, the number of dimensions corresponds to the number of variables minus the number of nonlinear equality constraints (because there is one such constraint for each variable that does not specify the dispatching polices, i.e., for each $\li$ and $\lb$ variable) minus the number of linear equality constraints (because these are normalization constraints, each of which reduces the ``degrees of freedom'' by one).

We note that further pruning is possible, so these problem sizes are particular to the way that we have formulated the problems.  First, we could eliminate the linear equality constraints and one variable occurring in each such constraint writing out this variable when it appears elsewhere in the optimization problem as one less the sum of the other variables appearing in the constraint; e.g., we could set $p(d\evec_1)$ equal to $\displaystyle{1-\sum_{\Dcal\backslash\{d\evec_1\}}p(\dvec)}$ in the optimization problem associated with $\DP\genq\cida$.  Second, problems allow for multiple optimal solutions in a way that could allow for further pruning opportunities that we have ignored.    For example, consider the following constraints of the optimization problem associated with $\DP{\genq}{\cida}$:
\begin{align*}
    \li&=\frac{\lambda}{q_i}\sum_{\dvec\in\Dcal}\left\{d_i\bi p(\dvec)\alpha_i(i,\gamma(i,\dvec))\sum_{a_i=1}^{d_i}\binom{d_i-1}{a_i-1}\dfrac{(1-\rho_i)^{a_i-1}\rho_i^{d_i-a_i}}{a_i}\right\}&(\forall i\in\Scal)\\
    \lb&=\frac{\lambda}{q_i\rho_i}\sum_{\dvec\in\Dcal}\left\{p(\dvec)\sum_{\mathclap{j\in\Jcal_{i}(\dvec)}}\bj\left(1-\rho_j^{d_j}\right)\alpha_i(j,\gamma(i,\dvec))\right\}&(\forall i\in\Scal)\\
    \sum_{\mathclap{\substack{i\in\Scal:\\ (i,j,\dvec)\in\Tcal}}}\alpha_i(j,\dvec)&=1&(\forall (j,\dvec)\in\Pcal).
\end{align*}
These are the only constraints in which the $\alpha_i(j,\dvec)$ variables appear for all $(i,j,\dvec)\in\Tcal$; these variables do not appear in the objective function.  Now consider an optimal solution and set the $\li$, $\lb$, and $p(\dvec)$ variables in the above constraints to values associated with that optimal solution.  Observe that any set of $\alpha(i,j,\dvec)$ values satisfying these constraints would also be optimal.  Noting that all of these constraints are linear in the $\alpha(i,j,\dvec)$ variables, the set of vectors of these variables satisfying these constraints form a polytope of at least $|\Tcal|-2|\Scal|-|\Pcal|$ dimensions, suggesting ample potential for further pruning.

\subsection{Numerical evaluation of problem sizes}

Table~1 presents the size of the problems associated with of our querying rule families in terms of the measures introduced in Section~\ref{sec:measures}.  In order to facilitate drawing insights from these problem sizes, in this section, we present numerical values for the entries of that table.  Before presenting these numerical values, we observe that most of the entries in Table~1 are presented in terms of $|\Scal|$, $|\Qcal|$, $|\Dcal|$, $|\Tcal|$, and $|\Pcal|$, which depend on $s$ and/or $d$, and are the cardinalities of sets defined respectively in Equations~\eqref{eq:scal}, \eqref{eq:dcal}, \eqref{eq:tcal}, \eqref{eq:pcal}, and~\eqref{eq:qcal}.  Consequently, computing the desired problem size measures as numerical values would require a method for computing the cardinalities of the aforementioned sets of $s$ and $d$.  Moreover, in the case of $\detq$, we also need to compute $|\Scal(\dvec)|$, $|\Tcal(\dvec)|$, and $|\Pcal(\dvec)|$, which are the cardinalities of sets defined respectively in Equations~\eqref{eq:scal-d}, \eqref{eq:tcal-d}, and~\eqref{eq:pcal-d}; since it would be inconvenient to examine these values for each $\dvec\in\Dcal$, we instead examine the maximum and average values of these sets over all $\dvec\in\Dcal$.  To this end, we present expressions for all these quantities of interest in Proposition~\ref{prop:combinatorics}.

\begin{proposition}
\label{prop:combinatorics}
Let $\Psi_{m}^n\equiv|\{\mathcal U\subseteq\{1,2,\ldots,n\}\colon |\mathcal U|\le m\}|=\sum_{\ell=0}^{m}\binom{n}{\ell}$ denote the number of subsets of cardinality at most $m$ (including the empty set) taken from a set of cardinality $n$.  The cardinalities of $\Scal$, $\Qcal$, $\Dcal$, $\Tcal$, and $\Pcal$---along with the maximum value and summation over all query mixes $\dvec\in\Dcal$ of the cardinalities of $\Scal(\dvec)$, $\Tcal(\dvec)$, and $\Pcal(\dvec)$---are as follows:
\begin{align*}
\begin{split}
    |\Scal|&=s,\quad |\Qcal|=d\\
    |\Dcal|&=
    \binom{s+d-1}{d}\\
    |\Tcal|&=s\Psi_{d-1}^{s-1}+\sum_{i=1}^s\left\{\Psi_{d-1}^{i-1}+\sum_{j=i+1}^s\Psi_{d-2}^{j-2}\right\}\\
    |\Pcal|&=\Psi_{d}^{s}+\sum_{j=2}^s\Psi_{d-1}^{j-1}\\
    \max_{\dvec\in\Dcal}|\Scal(\dvec)|&=\min(s,d)
\end{split}
\begin{split}
   \sum_{\dvec\in\Dcal}\frac{|\Scal(\dvec)|}{|\Dcal|}&=\frac{sd}{s+d-1}\\
    \max_{\dvec\in\Dcal}|\Tcal(\dvec)|&=\frac{\min\left(s^2+3s,d^2+3d\right)}{2}\\
   \sum_{\dvec\in\Dcal}\frac{|\Tcal(\dvec)|}{|\Dcal|}&=\frac{sd(sd+3s+3d-7)}{2(s+d-1)(s+d-2)}\\ 
    \max_{\dvec\in\Dcal}|\Pcal(\dvec)|&=\min(s,d)+1\\
    \sum_{\dvec\in\Dcal}\frac{|\Pcal(\dvec)|}{|\Dcal|}&=\frac{sd+s+d-1}{s+d-1}.
\end{split}
\end{align*}
Moreover, for all $\dvec\in\Dcal$, we have \[|\Tcal(\dvec)|=\frac{|\Scal(\dvec)|^2+3|\Scal(\dvec)|}{2}\quad\mbox{and}\quad\mathcal |\Pcal(\dvec)|=|\Scal(\dvec)|+1.\]
\end{proposition}

\begin{proof}
Clearly, $|\Scal|=|\{1,2,\ldots,s\}|=s$ and $|\Qcal|=|\{1,2,\ldots,d\}|=d$.  Meanwhile, $|\Dcal|$ is the number of query mixes consisting of $d$ servers drawn from $s$ classes, which corresponds exactly to the number of \emph{multisets} of cardinality $d$ taken from a set of cardinality $s$, so standard combinatorial reasoning (e.g., the so-called ``stars and bars'' technique) yields $|\Dcal|=\binom{s+d-1}{d}$.  We proceed to prove the remaining six formulas separately.

\noindent \textbf{Counting $|\Tcal|$.} Recalling that  \[\Tcal\equiv\left\{(i,j,\dvec)\in\Scal\times\Scalb\times\Dcal\colon i\le j,\, d_i>0,\, (j\le s)\implies d_j>0,\, \gamma(j,\dvec)=\dvec\right\},\] we count $|\Tcal|$ by counting the number of elements in each set in the partition described by the equation \[\Tcal=\bigcup_{i=1}^s\{(i',j,\dvec)\in\Tcal\colon i'=i\}\]  and then summing these counts:
\begin{align*}
|\Tcal|&=\sum_{i=1}^s|\{(i',j,\dvec)\in\Tcal\colon i'=i\}|\\
&=\sum_{i=1}^s\sum_{j=i}^{s+1}|\{(i',j',\dvec)\in\Tcal\colon (i',j')=(i,j)\}|\\
&=\sum_{i=1}^s\sum_{j=i}^{s+1}|\{\dvec\in\Dcal\colon (i,j,\dvec)\in\Tcal\}|\\
&=\sum_{i=1}^s\left\{|\{\dvec\in\Dcal\colon (i,i,\dvec)\in\Tcal\}|+\left(\sum_{j=i+1}^{s}|\{\dvec\in\Dcal\colon (i,j,\dvec)\in\Tcal\}|\right)+|\{\dvec\in\Dcal\colon (i,s+1,\dvec)\in\Tcal\}|\right\}\\
&=\sum_{i=1}^s\left\{\Psi_{d-1}^{i-1}+\left(\sum_{j=i+1}^s\Psi_{d-2}^{j-2}\right)+\Psi_{d-1}^{s-1}\right\}\\
&=s\Psi_{d-1}^{s-1} +\sum_{i=1}^s\left\{\Psi_{d-1}^{i-1}+\sum_{j=i+1}^s\Psi_{d-2}^{j-2}\right\},
\end{align*}
with the penultimate step requiring some further justification, which we now provide.  By the definition of $\Tcal$ (and specifically, the definition of the map $\gamma$), for any fixed $i\in\Scal$ and $j\in\Scalb$ such that $j\ge i$, the triple $(i,j,\dvec)\in\Tcal$ for some $\dvec\in\Dcal$ if and only if the query mix $\dvec$ satisfies the following conditions:
\begin{itemize}
    \item The query mix $\dvec$ consists of at least one class-$i$ server.
    \item The query mix $\dvec$ consists of at least one class-$j$ server if $j\in\Scal$ i.e., $j\neq s+1$.
    \item The query mix $\dvec$ consists of no servers that are slower than class $j$ (if $j=s+1$ all such $\dvec\in\Dcal$ satisfy this requirement vacuously).
    \item No more than one server of each class except for the fastest class is queried under the query mix $\dvec$.
\end{itemize}
Hence, there exists an isomorphism between such mixes $\dvec$ and the set of subsets of $\Scal$ consisting of at most $d$ classes (elements), which include both $i$ and (if $j\neq s+1$) $j$, and possibly some additional classes slower than class $i$ (selected from any of the classes $\{i+1,i+2,\ldots,j-1,j+1,\ldots,s\}$).
\begin{align*}
|\{d\in\Dcal\}\colon (i,j,\dvec)\in\Tcal|&=|\{\mathcal U\subseteq\{1,2,\ldots,\min(j,s)\}\colon |\mathcal U|\le d,\  i\in\mathcal U,\ j\in\mathcal U\cup\{s+1\}\}|\\
&=
\begin{cases}
|\{\mathcal U\subseteq\{1,2,\ldots,i-1\}\colon |\mathcal U|\le d-1\}|=\Psi_{d-1}^{i-1}&\mbox{if }i=j\\
|\{\mathcal U\subseteq\{1,2,\ldots,j-2\}\colon |\mathcal U|\le d-2\}|=\Psi_{d-2}^{j-2}&\mbox{if }i<j<s+1\\
|\{\mathcal U\subseteq\{1,2,\ldots,s-1\}\colon |\mathcal U|\le d-1\}|=\Psi_{d-1}^{s-1}&\mbox{if }j=s+1
\end{cases},
\end{align*}
as $\Psi_m^n\equiv|\{\mathcal U\subseteq\{1,2,\ldots,n\}\colon |\mathcal U|\le m\}|=\sum_{\ell=0}^{m}\binom{n}{\ell}$.  Hence, we have justified the equations above, and proved the claimed count of $|\Tcal|$.

\noindent\textbf{Counting $|\Pcal|$.} Recalling that \[\Pcal\equiv\left\{(j,\dvec)\in\Scalb\times\Dcal\colon(\exists i\in\Scal\colon (i,j,\dvec)\in\Tcal)\right\},\] we count $|\Pcal|$ by counting the number of elements in each set in the partition described by the equation \[\Pcal=\bigcup_{j=1}^s\{(j',\dvec)\in\Pcal\colon j'=j\}\] and then summing these counts:
\begin{align*}
    |\Pcal|&=\sum_{j=1}^{s+1}|\{(j',\dvec)\in\Pcal\colon j'=j\}|\\
    &=\sum_{j=1}^{s+1}|\{\dvec\in\Dcal\colon (j,\dvec)\in\Pcal\}|\\
    &=\sum_{j=1}^{s+1}|\{\dvec\in\Dcal\colon (\exists i\in\Scal\colon (i,j,\dvec)\in\Tcal)\}|\\
    &=|\{\mathcal U\subseteq\mathcal S\colon 1\le|\mathcal U|\le d\}|+\sum_{j=1}^s|\{\mathcal U\subseteq\{1,2,\ldots,j-1\}\colon |\mathcal U|\le d\}|\\
    &=\left(|\{\mathcal U\subseteq\mathcal S\colon |\mathcal U|\le d\}|-|\{\mathcal U\subseteq\mathcal S\colon |\mathcal U|=0\}|\right)+\left(|\{\mathcal U\subseteq\emptyset\colon|\mathcal U|\le d\}|+\sum_{j=2}^s|\{\mathcal U\subseteq\{1,2,\ldots,j-1\}\colon |\mathcal U|\le d\}|\right)\\
    &=\left(\Psi_d^s-1\right)+\left(1+\sum_{j=2}^s\Psi_{d-1}^{j-1}\right),\\
    &=\Psi_d^s+\sum_{j=2}^s\Psi_{d-1}^{j-1}.
\end{align*}

\noindent\textbf{Evaluating $\displaystyle{\max_{\dvec\in\Dcal}|\Scal(\dvec)|}$ and $\displaystyle{\sum_{\dvec\in\Dcal}|\Scal(\dvec)|}$.} Recalling that $\Scal(\dvec)\equiv\{i\in\Scal\colon d_i>0\}$, observe that we must clearly have $|\Scal(\dvec)|\le\min(s,d)$; in fact, this inequality holds tightly, e.g., when  $\dvec=\displaystyle{\max(d-s,0)\evec_s+\sum_{i=1}^{\mathclap{\min(s,d)}}\evec_i}$, we have $|\Scal(\dvec)| = \min(s,d)$.
Therefore, $\displaystyle{\max_{\dvec\in\Dcal}|\Scal(\dvec)|} = \min(s,d)$ as claimed.

Next, we evaluate the mean value of $|\Scal(\dvec)|$ over all $d\in\Dcal$.  To this end we count $|\{\Scal(\dvec)\colon\dvec\in\Dcal,\,|\Scal(\dvec)|=\ell\}|=|\{\dvec\in\Dcal\colon|\Scal(\dvec)|=\ell\}|$ for each $\ell\in\{1,2,\ldots,d\}$.  This count must correspond to the number of ways that $d$ indistinguishable balls can be put into $s$ distinguishable bins such that exactly $\ell$ bins are nonempty.  It follows from this observation that $|\{\dvec\in\Dcal\colon|\Scal(\dvec)|=\ell\}|=\binom{s}{\ell}\binom{d-1}{d-\ell}$ because there are $\binom{s}{\ell}$ ways of choosing $\ell$ nonempty bins from among $s$ possible bins, and---since each of the $\ell$ nonempty bins necessarily includes at least one ball---there are (by the ``stars and bars'' technique)  $\binom{\ell+(d-\ell)-1}{d-1}=\binom{d-1}{d-\ell}$ ways of distributing the remaining $d-\ell$ balls among these $\ell$ bins.  We then use these counts to compute the desired mean value as follows:
\[ \footnotesize
    \frac1{|\Dcal|}\sum_{\dvec\in\Dcal}|\Scal(\dvec)|=\frac1{|\Dcal|}\sum_{\ell=1}^d \ell\left(|\{\Scal(\dvec)\colon\dvec\in\Dcal,\,|\Scal(\dvec)|=\ell\}|\right)=\frac1{|\Dcal|}\sum_{\ell=1}^d \ell\binom{s}{\ell}\binom{d-1}{d-\ell}=\frac{sd}{(s+d-1)|\Dcal|}\binom{s+d-1}{d}=\frac{sd}{s+d-1},
\]
where the penultimate inequality can be verified using a computer algebra system and the last equality follows from the previously established fact that $|\Dcal|=\binom{s+d-1}{d}$.

\noindent\textbf{Expressing $|\Tcal(\dvec)|$ in terms of $|\Scal(\dvec)|$ and evaluating $\displaystyle{\max_{\dvec\in\Dcal}|\Tcal(\dvec)|}$ and $\displaystyle{\sum_{\dvec\in\Dcal}|\Tcal(\dvec)|}$.} Recalling that \[\Tcal(\dvec)\equiv\left\{(i,j)\in\Scal\times\Scalb\colon i\le j,\, d_i>0,\, (j\le s)\implies d_j>0\right\},\] and $\Scal(\dvec)\equiv\{i\in\Scal\colon d_i>0\}$, observe that we can partition $\mathcal T(\dvec)$ into the three sets $\{(i,i)\colon i\in\Scal(\dvec)\}$, $\left\{(i,j)\in\Scal(\dvec)^2\colon i<j\right\}$, and $\{(i,s+1)\colon i\in\Scal(\dvec)\}$.  Using this partition, we compute $|\Tcal(\dvec)|$ in terms of $\Scal(\dvec)$:
\begin{align*}
    |\mathcal T(\dvec)|&=|\{(i,i)\colon i\in\Scal(\dvec)\}|+\left|\left\{(i,j)\in\Scal(\dvec)^2\colon i<j\right\}\right|+|\{(i,s+1)\colon i\in\Scal(\dvec)\}|\\
    &=|\Scal(\dvec)|+\binom{|\Scal(\dvec)|}{2}+|\Scal(\dvec)|=\frac{|\Scal(\dvec)|^2+3|\Scal(\dvec)|}{2},
\end{align*}
as claimed.  

Since $|\mathcal T(\dvec)|$ is monotonically increasing in $|\Scal(\dvec)|$, the former is maximized at precisely those values of $\dvec$ where the latter is maximized, allowing us to obtain the claimed result on the maximum value of $|\Tcal(\dvec)|$ over all $\dvec\in\Dcal$: \[\max_{\dvec\in\Dcal}|\Tcal(\dvec)|=\max_{\dvec\in\Dcal}\left\{\frac{|\Scal(\dvec)|^2+3|\Scal(\dvec)|}{2}\right\}=\frac{\min(s,d)^2+3\min(s,d)}2=\frac{\min\left(s^2+3s,d^2+3d\right)}{2}.\]

Our expression of $|\Tcal(\dvec)|$ in terms of $|\Scal(\dvec)|$ also allows us to prove the claim regarding the mean value of $|\Tcal(\dvec)|$ over all $\dvec\in\Dcal$ in a fashion similar to that used to prove the claim regarding the analogous mean value associated with $|\Scal(\dvec)|$:
\begin{align*}
    \footnotesize
    \frac1{|\Dcal|}\sum_{\dvec\in\Dcal}|\Tcal(\dvec)|=\frac1{|\Dcal|}\sum_{\ell=1}^d \left(\frac{\ell^2+3\ell}{2}\right)\left(|\{\Scal(\dvec)\colon\dvec\in\Dcal,\,|\Scal(\dvec)|=\ell\}|\right)=\frac1{|\Dcal|}\sum_{\ell=1}^d \left(\frac{\ell^2+3\ell}{2}\right)\binom{s}{\ell}\binom{d-1}{d-\ell}=\frac{sd(sd+3s+3d-7)}{2(s+d-1)(s+d-2)},
\end{align*}
where the last equality can be verified using a computer algebra system.

\noindent\textbf{Expressing $|\Pcal(\dvec)|$ in terms of $|\Scal(\dvec)|$ and evaluating $\displaystyle{\max_{\dvec\in\Dcal}|\Pcal(\dvec)|}$ and $\displaystyle{\sum_{\dvec\in\Dcal}|\Pcal(\dvec)|}$.} Recalling that \[\Pcal(\dvec)\equiv\left\{j\in\Scalb\colon(\exists i\in\Scal\colon(i,j)\in\Tcal(\dvec))\right\}=\{j\in\Scal\colon d_j>0\}\cup\{s+1\},\] we must clearly have $|\Pcal(\dvec)|=|\Scal(\dvec)|+1$, as $\Scal(\dvec)\equiv\{\ell\in\Scal\colon d_\ell>0\}$.  We readily obtain the remaining claimed results:
\begin{align*}
    \max_{\dvec\in\Dcal}|\Pcal(\dvec)|&=\max_{\dvec\in\Dcal}\{|\Scal(\dvec)|+1\}=\min(s,d)+1\\
    \frac1{|\Dcal|}\sum_{\dvec\in\Dcal}|\Pcal(\dvec)|&=\footnotesize{\frac1{|\Dcal|}\sum_{\dvec\in\Dcal}\{|\Scal(\dvec)|+1\}=\frac1{|\Dcal|}\left(|\Dcal|+\sum_{\dvec\in\Dcal}|\Scal(\dvec)|\right)=1+\frac1{|\Dcal|}\sum_{\dvec\in\Dcal}|\Scal(\dvec)|=1+\frac{sd}{s+d-1}=\frac{sd+s+d-1}{sd}.}
\end{align*}
\end{proof}

\begin{table}[!htp]
\label{tab:numeric}
\centering
\begin{adjustbox}{addcode={\begin{minipage}{\width}}{\caption{%
Problem sizes for the optimization problems associated with $\DP{\qrf}{\cida}$ for five querying rule families $\qrf$ and for $\DP{\qr}{\cida}$ (for an individual querying rule $\qr$)---in terms of VAR, LEC, NEC, and DIM---for each combination of $s,d\in\{2,3,4,5\}$.  For the case where $\qrf=\detq$ we present both the maximum (labeled ``max'') and average (labeled ``avg'') problem size measures, rounding averages to the nearest integer; we also present the number of subproblems (above the letters ``SP'' on the far right of the box associated with each $(s,d)$ pair in the rows labeled ``$\detq$'').
      }\end{minipage}},rotate=90,center, scale=0.9}

\begin{tabular}{ccc !{\vrule width 2pt}c|c|c|c|c!{\vrule width 2pt}c|c|c|c|c!{\vrule width 2pt}c|c|c|c|c!{\vrule width 2pt}c|c|c|c|c!{\vrule width 2pt}}
\cline{4-23}
 &  &  & \multicolumn{5}{c!{\vrule width 2pt}}{$d=2$} & \multicolumn{5}{c!{\vrule width 2pt}}{$d=3$} & \multicolumn{5}{c!{\vrule width 2pt}}{$d=4$} & \multicolumn{5}{c!{\vrule width 2pt}}{$d=5$} \\ \cline{2-23} 
\multicolumn{1}{c|}{} & \multicolumn{2}{c!{\vrule width 2pt}}{Policy} & VAR & LEC & NEC & \multicolumn{2}{c!{\vrule width 2pt}}{DIM} & VAR & LEC & NEC & \multicolumn{2}{c!{\vrule width 2pt}}{DIM} & VAR & LEC & NEC & \multicolumn{2}{c!{\vrule width 2pt}}{DIM} & VAR & LEC & NEC & \multicolumn{2}{c!{\vrule width 2pt}}{DIM} \\ 
\specialrule{.2em}{.1em}{.1em} 
\rowcolor[HTML]{EFEFEF}
\multicolumn{1}{|c|}{\cellcolor{white}} & \multicolumn{2}{c!{\vrule width 2pt}}{$\genq$} & 15 & 7 & 4 & \multicolumn{2}{c!{\vrule width 2pt}}{4} & 16 & 7 & 4 & \multicolumn{2}{c!{\vrule width 2pt}}{5} & 17 & 7 & 4 & \multicolumn{2}{c!{\vrule width 2pt}}{6} & 18 & 7 & 4 & \multicolumn{2}{c!{\vrule width 2pt}}{7} \\ \cline{2-23} 
\multicolumn{1}{|c|}{} & \multicolumn{2}{c!{\vrule width 2pt}}{$\indq$} & 16 & 8 & 4 & \multicolumn{2}{c!{\vrule width 2pt}}{4} & 18 & 9 & 4 & \multicolumn{2}{c!{\vrule width 2pt}}{5} & 20 & 10 & 4 & \multicolumn{2}{c!{\vrule width 2pt}}{6} & 22 & 11 & 4 & \multicolumn{2}{c!{\vrule width 2pt}}{7}\\ \cline{2-23} 
\rowcolor[HTML]{EFEFEF}
\multicolumn{1}{|c|}{\cellcolor{white}} & \multicolumn{2}{c!{\vrule width 2pt}}{$\iidq$} & 14 & 7 & 4 & \multicolumn{2}{c!{\vrule width 2pt}}{3} & 14 & 7 & 4 & \multicolumn{2}{c!{\vrule width 2pt}}{3} & 14 & 7 & 4 & \multicolumn{2}{c!{\vrule width 2pt}}{3} & 14 & 7 & 4 & \multicolumn{2}{c!{\vrule width 2pt}}{3} \\ \cline{2-23} 
\multicolumn{1}{|c|}{} & \multicolumn{1}{c|}{} & max & 9 & 3 & 4 & 2 & \multicolumn{1}{c!{\vrule width 2pt}}{{\small 3}} & 9 & 3 & 4 & 2 & \multicolumn{1}{c!{\vrule width 2pt}}{{\small 4}} & 9 & 3 & 4 & 2 & \multicolumn{1}{c!{\vrule width 2pt}}{{\small 5}} & 9 & 3 & 4 & 2 & \multicolumn{1}{c!{\vrule width 2pt}}{{\small 6}} \\ \cline{3-7} \cline{9-12} \cline{14-17} \cline{19-22}
\multicolumn{1}{|c|}{} & \multicolumn{1}{c|}{\multirow{-2}{*}{$\detq$}} & avg & 6 & 2 & 3 & 1 & \multicolumn{1}{c!{\vrule width 2pt}}{{\scriptsize SP}} & 6 & 2 & 3 & 1 & \multicolumn{1}{c!{\vrule width 2pt}}{{\scriptsize SP}} & 7 & 3 & 3 & 1 & \multicolumn{1}{c!{\vrule width 2pt}}{{\scriptsize SP}} & 7 & 3 & 3 & 1 & 
\multicolumn{1}{c!{\vrule width 2pt}}{{\scriptsize SP}}\\ \cline{2-23} 
\rowcolor[HTML]{EFEFEF}
\multicolumn{1}{|c|}{\cellcolor{white}} & \multicolumn{2}{c!{\vrule width 2pt}}{$\srcq$} &  6 & 1 & 4 & \multicolumn{2}{c!{\vrule width 2pt}}{1} & 6 & 1 & 4 & \multicolumn{2}{c!{\vrule width 2pt}}{1} & 6 & 1 & 4 & \multicolumn{2}{c!{\vrule width 2pt}}{1} & 6 & 1 & 4 & \multicolumn{2}{c!{\vrule width 2pt}}{1} \\ \cline{2-23} 
\multicolumn{1}{|c|}{\cellcolor{white} \multirow{-7}{*}{s = 2}} & \multicolumn{2}{c!{\vrule width 2pt}}{$\qr$} & 12 & 6 & 4 & \multicolumn{2}{c!{\vrule width 2pt}}{2} & 12 & 6 & 4 & \multicolumn{2}{c!{\vrule width 2pt}}{2} & 12 & 6 & 4 & \multicolumn{2}{c!{\vrule width 2pt}}{2} & 12 & 6 & 4 & \multicolumn{2}{c!{\vrule width 2pt}}{2} \\ 
\specialrule{.2em}{.1em}{.1em} 
\rowcolor[HTML]{EFEFEF}
\multicolumn{1}{|c|}{\cellcolor{white}} & \multicolumn{2}{c!{\vrule width 2pt}}{$\genq$} & 30 & 13 & 6 & \multicolumn{2}{c!{\vrule width 2pt}}{11} & 40 & 15 & 6 & \multicolumn{2}{c!{\vrule width 2pt}}{19} & 45 & 15 & 6 & \multicolumn{2}{c!{\vrule width 2pt}}{24} & 51 & 15 & 6 & \multicolumn{2}{c!{\vrule width 2pt}}{30} \\ \cline{2-23} 
\multicolumn{1}{|c|}{} & \multicolumn{2}{c!{\vrule width 2pt}}{$\indq$} & 30 & 14 & 6 & \multicolumn{2}{c!{\vrule width 2pt}}{10} & 39 & 17 & 6 & \multicolumn{2}{c!{\vrule width 2pt}}{16} & 42 & 18 & 6 & \multicolumn{2}{c!{\vrule width 2pt}}{18} & 45 & 19 & 6 & \multicolumn{2}{c!{\vrule width 2pt}}{20} \\ \cline{2-23} 
\rowcolor[HTML]{EFEFEF}
\multicolumn{1}{|c|}{\cellcolor{white}} & \multicolumn{2}{c!{\vrule width 2pt}}{$\iidq$} & 27 & 13 & 6 & \multicolumn{2}{c!{\vrule width 2pt}}{8} & 33 & 15 & 6 & \multicolumn{2}{c!{\vrule width 2pt}}{12} & 33 & 15 & 6 & \multicolumn{2}{c!{\vrule width 2pt}}{12} & 33 & 15 & 6 & \multicolumn{2}{c!{\vrule width 2pt}}{12} \\ \cline{2-23} 
\multicolumn{1}{|c|}{} & \multicolumn{1}{c|}{} & max & 9 & 3 & 4 & 2 & \multicolumn{1}{c!{\vrule width 2pt}}{6} & 15 & 4 & 6 & 5 & \multicolumn{1}{c!{\vrule width 2pt}}{10} & 15 & 4 & 6 & 5 & \multicolumn{1}{c!{\vrule width 2pt}}{15} & 15 & 4 & 6 & 5 & \multicolumn{1}{c!{\vrule width 2pt}}{21} \\ \cline{3-7} \cline{9-12} \cline{14-17} \cline{19-22} 
\multicolumn{1}{|c|}{} & \multicolumn{1}{c|}{\multirow{-2}{*}{$\detq$}} & avg & 6 & 2 & 3 & 1 & \multicolumn{1}{c!{\vrule width 2pt}}{{\scriptsize SP}} & 8 & 3 & 4 & 2 & \multicolumn{1}{c!{\vrule width 2pt}}{{\scriptsize SP}} & 9 & 3 & 4 & 2 & \multicolumn{1}{c!{\vrule width 2pt}}{{\scriptsize SP}} & 10 & 3 & 4 & 3 & \multicolumn{1}{c!{\vrule width 2pt}}{{\scriptsize SP}} \\ \cline{2-23} 
\rowcolor[HTML]{EFEFEF}
\multicolumn{1}{|c|}{\cellcolor{white}} & \multicolumn{2}{c!{\vrule width 2pt}}{$\srcq$} & 9 & 1 & 6 & \multicolumn{2}{c!{\vrule width 2pt}}{2} & 9 & 1 & 6 & \multicolumn{2}{c!{\vrule width 2pt}}{2} & 9 & 1 & 6 & \multicolumn{2}{c!{\vrule width 2pt}}{2} & 9 & 1 & 6 & \multicolumn{2}{c!{\vrule width 2pt}}{2} \\ \cline{2-23} 
\multicolumn{1}{|c|}{\cellcolor{white} \multirow{-7}{*}{s = 3}} & \multicolumn{2}{c!{\vrule width 2pt}}{$\qr$} & 24 & 12 & 6 & \multicolumn{2}{c!{\vrule width 2pt}}{6} & 30 & 14 & 6 & \multicolumn{2}{c!{\vrule width 2pt}}{10} & 30 & 14 & 6 & \multicolumn{2}{c!{\vrule width 2pt}}{10} & 30 & 14 & 6 & \multicolumn{2}{c!{\vrule width 2pt}}{10} \\ 
\specialrule{.2em}{.1em}{.1em} 
\rowcolor[HTML]{EFEFEF}
\multicolumn{1}{|c|}{\cellcolor{white}} & \multicolumn{2}{c!{\vrule width 2pt}}{$\genq$} & 50 & 21 & 8 & \multicolumn{2}{c!{\vrule width 2pt}}{21} & 84 & 29 & 8 & \multicolumn{2}{c!{\vrule width 2pt}}{47} & 107 & 31 & 8 & \multicolumn{2}{c!{\vrule width 2pt}}{68} & 128 & 31 & 8 & \multicolumn{2}{c!{\vrule width 2pt}}{89} \\ \cline{2-23} 
\multicolumn{1}{|c|}{} & \multicolumn{2}{c!{\vrule width 2pt}}{$\indq$} & 48 & 22 & 8 & \multicolumn{2}{c!{\vrule width 2pt}}{18} & 76 & 31 & 8 & \multicolumn{2}{c!{\vrule width 2pt}}{37} & 88 & 34 & 8 & \multicolumn{2}{c!{\vrule width 2pt}}{46} & 92 & 35 & 8 & \multicolumn{2}{c!{\vrule width 2pt}}{49} \\ \cline{2-23} 
\rowcolor[HTML]{EFEFEF}
\multicolumn{1}{|c|}{\cellcolor{white}} & \multicolumn{2}{c!{\vrule width 2pt}}{$\iidq$} & 44 & 21 & 8 & \multicolumn{2}{c!{\vrule width 2pt}}{15} & 68 & 29 & 8 & \multicolumn{2}{c!{\vrule width 2pt}}{31} & 76 & 31 & 8 & \multicolumn{2}{c!{\vrule width 2pt}}{37} & 76 & 31 & 8 & \multicolumn{2}{c!{\vrule width 2pt}}{37} \\ \cline{2-18} 
\multicolumn{1}{|c|}{} & \multicolumn{1}{c|}{} & max & 9 & 3 & 4 & 2 & \multicolumn{1}{c!{\vrule width 2pt}}{10} & 15 & 4 & 6 & 5 & \multicolumn{1}{c!{\vrule width 2pt}}{20} & 22 & 5 & 8 & 9 & \multicolumn{1}{c!{\vrule width 2pt}}{35} & 22 & 5 & 8 & 9 & \multicolumn{1}{c!{\vrule width 2pt}}{56} \\ \cline{3-7} \cline{9-12} \cline{14-17} \cline{19-22}  
\multicolumn{1}{|c|}{} & \multicolumn{1}{c|}{\multirow{-2}{*}{$\detq$}} & avg & 7 & 3 & 3 & 1 & \multicolumn{1}{c!{\vrule width 2pt}}{\scriptsize SP} & 9 & 3 & 4 & 2 & \multicolumn{1}{c!{\vrule width 2pt}}{\scriptsize SP} & 11 & 3 & 5 & 3 & \multicolumn{1}{c!{\vrule width 2pt}}{\scriptsize SP} & 12 & 4 & 5 & 4 & \multicolumn{1}{c!{\vrule width 2pt}}{\scriptsize SP} \\ \cline{2-23} 
\rowcolor[HTML]{EFEFEF}
\multicolumn{1}{|c|}{\cellcolor{white}} & \multicolumn{2}{c!{\vrule width 2pt}}{$\srcq$} & 12 & 1 & 8 & \multicolumn{2}{c!{\vrule width 2pt}}{3} & 12 & 1 & 8 & \multicolumn{2}{c!{\vrule width 2pt}}{3} & 12 & 1 & 8 & \multicolumn{2}{c!{\vrule width 2pt}}{3} & 12 & 1 & 8 & \multicolumn{2}{c!{\vrule width 2pt}}{3} \\ \cline{2-23} 
\multicolumn{1}{|c|}{\cellcolor{white} \multirow{-7}{*}{s = 4}} & \multicolumn{2}{c!{\vrule width 2pt}}{$\qr$} & 40 & 20 & 8 & \multicolumn{2}{c!{\vrule width 2pt}}{12} & 64 & 28 & 8 & \multicolumn{2}{c!{\vrule width 2pt}}{28} & 72 & 30 & 8 & \multicolumn{2}{c!{\vrule width 2pt}}{34} & 72 & 30 & 8 & \multicolumn{2}{c!{\vrule width 2pt}}{34} \\ 
\specialrule{.2em}{.1em}{.1em} 

\rowcolor[HTML]{EFEFEF}
\multicolumn{1}{|c|}{\cellcolor{white}} & \multicolumn{2}{c!{\vrule width 2pt}}{$\genq$} & 75 & 31 & 10 & \multicolumn{2}{c!{\vrule width 2pt}}{34} & 155 & 51 & 10 & \multicolumn{2}{c!{\vrule width 2pt}}{94} & 230 & 61 & 10 & \multicolumn{2}{c!{\vrule width 2pt}}{159} & 296 & 63 & 10 & \multicolumn{2}{c!{\vrule width 2pt}}{223} \\ \cline{2-23} 
\multicolumn{1}{|c|}{} & \multicolumn{2}{c!{\vrule width 2pt}}{$\indq$} & 70 & 32 & 10 & \multicolumn{2}{c!{\vrule width 2pt}}{28} & 135 & 53 & 10 & \multicolumn{2}{c!{\vrule width 2pt}}{72} & 180 & 64 & 10 & \multicolumn{2}{c!{\vrule width 2pt}}{106} & 195 & 67 & 10 & \multicolumn{2}{c!{\vrule width 2pt}}{118} \\ \cline{2-23} 
\rowcolor[HTML]{EFEFEF}
\multicolumn{1}{|c|}{\cellcolor{white}} & \multicolumn{2}{c!{\vrule width 2pt}}{$\iidq$} & 65 & 31 & 10 & \multicolumn{2}{c!{\vrule width 2pt}}{24} & 125 & 51 & 10 & \multicolumn{2}{c!{\vrule width 2pt}}{64} & 165 & 61 & 10 & \multicolumn{2}{c!{\vrule width 2pt}}{94} & 175 & 63 & 10 & \multicolumn{2}{c!{\vrule width 2pt}}{102} \\ \cline{2-18} 
\multicolumn{1}{|c|}{} & \multicolumn{1}{c|}{} & max & 9 & 3 & 4 & 2 & \multicolumn{1}{c!{\vrule width 2pt}}{15} & 15 & 4 & 6 & 5 & \multicolumn{1}{c!{\vrule width 2pt}}{35} & 22 & 5 & 8 & 9 & \multicolumn{1}{c!{\vrule width 2pt}}{70} & 30 & 6 & 10 & 14 & \multicolumn{1}{c!{\vrule width 2pt}}{126} \\ \cline{3-7} \cline{9-12} \cline{14-17} \cline{19-22}  
\multicolumn{1}{|c|}{} & \multicolumn{1}{c|}{\multirow{-2}{*}{$\detq$}} & avg & 7 & 3 & 3 & 1 & \multicolumn{1}{c!{\vrule width 2pt}}{\scriptsize SP} & 10 & 3 & 4 & 3 & \multicolumn{1}{c!{\vrule width 2pt}}{\scriptsize SP} & 12 & 4 & 5 & 4 & \multicolumn{1}{c!{\vrule width 2pt}}{\scriptsize SP} & 14 & 4 & 6 & 5 & \multicolumn{1}{c!{\vrule width 2pt}}{\scriptsize SP} \\ \cline{2-23} 
\rowcolor[HTML]{EFEFEF}
\multicolumn{1}{|c|}{\cellcolor{white}} & \multicolumn{2}{c!{\vrule width 2pt}}{$\srcq$} & 15 & 1 & 10 & \multicolumn{2}{c!{\vrule width 2pt}}{4} & 15 & 1 & 10 & \multicolumn{2}{c!{\vrule width 2pt}}{4} & 15 & 1 & 10 & \multicolumn{2}{c!{\vrule width 2pt}}{4} & 15 & 1 & 10 & \multicolumn{2}{c!{\vrule width 2pt}}{4} \\ \cline{2-23} 
\multicolumn{1}{|c|}{\cellcolor{white} \multirow{-7}{*}{s = 5}} & \multicolumn{2}{c!{\vrule width 2pt}}{$\qr$} & 60 & 30 & 10 & \multicolumn{2}{c!{\vrule width 2pt}}{20} & 120 & 50 & 10 & \multicolumn{2}{c!{\vrule width 2pt}}{60} & 160 & 60 & 10 & \multicolumn{2}{c!{\vrule width 2pt}}{90} & 170 & 62 & 10 & \multicolumn{2}{c!{\vrule width 2pt}}{98} \\ 
\specialrule{.2em}{.1em}{.1em}
\end{tabular}
\end{adjustbox}
\end{table}

Using Proposition~\ref{prop:combinatorics}, we present numerical values of the problem sizes associated with the optimization problems of interest for each combination of $s,d\in\{2,3,4,5\}$ in Table~7; we omit the optimization problem associated with $\DP{\sfcq}{\cida}$, as it features $s$ problems and the sizes of all of the subproblems (insofar as our measures are concerned) are constant in both $s$ and $d$.  For $\DP{\detq}{\cida}$ we report the maximum and average measures across the subproblems (of which there is one for each $\dvec\in\Dcal$; we also report the number of subproblems).  As average measures may not be integers, we round them to the nearest integer to conserve space.

Across most of our measures and combinations of $s$ and $d$, the problem associated with $\genq$ is larger than that associated with $\indq$ (although there are exceptions, particularly with LEC), which is in turn larger than that associated with $\iidq$, with those associated with $\qr$ and $\srcq$ following in that order (with $\srcq$ problem sizes being much smaller than those associated with the other families/rules at higher $s$ and $d$).  The proportional disparity in problem sizes with respect to most of our measures becomes substantially more pronounced as both $s$ and $d$ grow.  Meanwhile, while each subproblem  associated with $\detq$ is small, the number of subproblems grows sharply in $s$ and $d$, suggesting that the optimization problem associated with $\DP\detq\cida$ is generally the most computationally intensive among the problems studied in this paper. As the $\detq$-driven dispatching policies are the only ones studied in~\cite{gardner2020scalable} (to our knowledge, the only work that jointly optimizes querying and assignment rules), one important contribution of this paper is the generalization of the $\detq$ querying rule family to families (i.e., $\indq$ and $\genq$) that more readily lend themselves to efficient optimization. 


\newpage 

\section{Appendix: Notes on implementation and numerical results}
\label{app:implementation}
In this section, we present some notes on the implementation of the numerical computations and results presented in Sections~\ref{sec:numerical}~and~\ref{sec:construction}.

\subsection{System specifications and implementation}

The numerical results in Section~\ref{sec:numerical} were obtained on a system with an Intel i7-7500U CPU @ 2.70GHz with 16GB of 2133MHz RAM running the Windows 10 Home 64-bit operating system.  These numerical results were obtained using code written in the programming language Julia.  We used the JuMP package \cite{dunning2017jump} in Julia to define our optimization models and we solved these problems using the Interior Point Optimizer (IPOPT) optimization algorithm \cite{Lubin2015}.  

The JuMP package facilitates formulating the optimization problems as presented in the paper.  For our optimization problem, we used an alternate $\gamma$ mapping defined as follows:
\[\gamma(j,\dvec)\mapsto\sum_{i=1}^sI\{d_i>0\mbox{ and } i\le j\}\evec_i.\]
This alternate definition should not influence the solutions or the size of the optimization problem in any way.  The optimization problem resulting from this alternate mapping is equivalent to that obtained using the mapping defined in Section~\ref{sec:optimization}.

\subsection{Failure to find solutions}

We note that while IPOPT returned a solution for more than $99.99\%$ of these optimization problems, there were 5 optimization problems where IPOPT failed to find a solution (arising from 2 and 3 parameter settings for the $\DP\indq\cida$ and $\DP\iidq\cida$ families, respectively).  These cases were omitted from the average and median $\mathbb E[T]$ calculations for the corresponding policies, but were included in the policy runtimes.  We took the runtime in these cases to be the time it took for IPOPT to report that it failed to find a solution.

\subsection{Computing $\genseed$}
\label{app:genseed}

The $\genseed$ policy is found using the following methodology: we seed the optimization problem associated with $\DP\genq\cida$ with an initial value corresponding to the parameters of $\ipoptd\indq$ (where possible) and augment this optimization problem by imposing an additional constraint that the objective function value be no higher than $\mathbb E[T]^{\ipoptd\indq}+10^{-3}$; this additional slack of $10^{-3}$ helped obtain better solutions, presumably by avoiding complications that would arise due to numerical imprecision.  We take the $\genseed$ policy to be specified by the solution IPOPT finds to this optimization problem, unless IPOPT fails to find a solution (which occurs across $2.44\%$ of our parameter settings), in which case we set $\genseed=\ipoptd\indq$.  Finally, whenever IPOPT fails to find any solution to the optimization problem associated with $\DP\indq\cida$ (which occurred in only $2$ of our $12\,825$ parameter settings), we set $\genseed=\ipoptd\genq$.

\newpage

\section{General service times and scheduling rules}
\label{sec:general}
In this section we briefly discuss how the consideration of general (rather than exponential) service distributions impacts the framework and analysis presented thus far in this paper.  This generalization also necessitates giving explicit consideration to the ramifications of the choice of scheduling rule.

\subsection{The general service distribution model}

We now relax the assumption that job sizes (i.e., service requirements in terms of time) are distributed according to an exponential distribution.  Instead, we now assume that the size of a job running on a class-$i$ server is drawn from a general distribution with cdf $G_i$ and mean $1/\mu_i$, where all such distributions have the same ``shape,'' i.e., $G_i(x)=G_\ell(\mu_i x/\mu_\ell)$. For example, job sizes could be drawn from Pareto distributions, so that for each $i\in\mathcal S$, a job's size at a class-$i$ server follows the cdf  $G_i(x) = 1 - \left( \frac{\kappa - 1}{\kappa\mu_i x} \right)^{\kappa}$, for $\kappa > 2$.

\begin{remark}
We provide one interpretation of the assumption that job size distributions have the same shape at all server classes: each job's size $X$ is drawn---independently of the server on which it ultimately runs---from some distribution $G$ with mean 1; the job would then require $X/\mu_i$ time units of service to run on a server of speed $\mu_i$ (see \cite{gardner2017better} for further discussion on this assumption).
\end{remark}

\subsection{Analysis under general service distributions and various scheduling rules}

With the new assumption described above, the analysis of $\mathbb E[T]$ now depends on the choice of scheduling rule, $\mathsf{SR}$; stability results remain unchanged. The only changes to the analysis presented in Section~\ref{sec:analysis} are the formula for $\mathbb E[T_i]$ and the resulting formula for $\mathbb E[T]$.  Rather than using the M/M/1 formula for $\mathbb E[T_i]$, we now use the M/G/1/$\mathsf{SR}$ formula.  Recall that we have restricted attention to size-blind working-conserving scheduling rules.  Common examples of such rules include \textsf{First Come First Served} ($\fcfs$), \textsf{Processor Sharing} ($\ps$) and \textsf{Preemptive Last Come First Served} ($\plcfs$). 
We also note that the analysis in Section~\ref{sec:analysis} relies on the assumption that, as the number of servers approaches infinity, all queue states become independent. This assumption has, in general, been more difficult to prove in settings with general service time distributions.
The result has been shown for the Least-Loaded($d$) dispatching policy, for the JSQ($d$) dispatching policy under FCFS scheduling when the service time distribution has decreasing hazard rate, and for any dispatching policy provided that the arrival rate is sufficiently small and that the service time distribution has finite first and second moments~\cite{Bramson2012}.  See Appendix~\ref{app:asymptotic} for simulation studies supporting the asymptotic independence assumption under our model.

Under $\fcfs$ scheduling, using the Pollaczek-Khinchine formula (for the mean response time in an M/G/1/$\fcfs$ system) we find that
\begin{align}\label{eq:generalization-1}
\mathbb E[T_i]=\frac{\left(1+C^2\right)\lb}{2\mu_i(\mu_i-\lb)}+\frac1{\mu_i},\end{align}  where the squared coefficient of variation \label{Csquared} $C^2\equiv\ep[X^2]/\ep[X]^2-1$ for $X\sim G$ (equivalently, for $X\sim G_i$ for any $i\in\Scal$).  These changes result in a modified objective function for all of the optimization problems presented in this paper, but they leave the variables and constraints of these problems unchanged.
Specifically, the new objective function under the $\mathsf{FCFS}$ scheduling rule is as follows:
\begin{align}\label{eq:generalization-2}
    \ep[T] = \sum_{i=1}^s \left(\frac{q_i\lambda_i}{\lambda}\right)\ep[T_i]
    =\frac{1}{\lambda} \sum_{i=1}^{s} q_i\lambda_i \left( \frac{\left(1+C^2\right)\lb}{2\mu_i(\mu_i-\lb)}+\frac1{\mu_i} \right).
\end{align}

Meanwhile, under both $\ps$ and $\plcfs$, $\mathbb E[T_i]=1/(\mu_i-\lambda_i)$, just as it was under the assumption of exponentially distributed job sizes for all work-conserving size-blind scheduling rules.  Therefore, under both $\ps$ and $\plcfs$, the analysis presented in Section~\ref{sec:analysis} holds without modification (assuming that asymptotic independence holds).

In principle one could generalize this model by allowing for ``greater heterogeneity'' across server classes.  Specifically, one could allow an entirely different cdf $G_i$ for the service distribution (rather than merely considering scalings of the same distributions), and one could also implement a different scheduling rule at each server class (perhaps to compensate for the different job size distributions).  While analyzing policies under such generalizations may at first appear to be straightforward, these generalizations can introduce complications that are best explored in future work. Most notably, there is no longer a natural ``ordering'' on server classes as two server classes may have the same ``speed,'' (so that jobs will have the same mean service time at each) but with different service time distribution ``shapes.''  Moreover, even when one server class is faster than another, differences in the service time distribution ``shapes'' and scheduling rules employed at each class may confound the use of a clear-cut ordering, and possibly necessitate more care in the pruning of $\cida$ assignment rules than that presented earlier in this paper.

\subsection{Age-based assignment and scheduling rules}

The allowance of general service distributions introduces new state information that can be used to further differentiate busy servers (even those of the same class) from one another for the purpose of assignment.  Specifically, we define the \emph{attained service level} (ASL) of a job (sometimes referred to as a job's \emph{age}) at any given time to be the number of units of processing time received by the job thus far. For example, under $\mathsf{FCFS}$ a job that has been in service for 2 minutes has an ASL of 2 minutes.  Now consider a slightly more complicated example: under $\mathsf{PS}$ a job that has been present at a server for 2 minutes, where it spent the first minute alone and the second minute sharing the server with one other job has an ASL of 1.5 minutes.

Scheduling rules informed by the ASLs of the jobs present at a server have been explored extensively in the literature.  One common size-blind work-conserving ASL-based policy is the \textsf{Least Attained Service} ($\las$)---also known as \textsf{Foreground-Background} ($\fb$)---scheduling rule~\cite{nuijens2004foreground,rai2003analysis}, which can be generalized to the \textsf{Gittins Index Policy}~\cite{gittins2011multi,scully2018soap}.  In principle, for any such policy with a known closed-form mean response time formula, our performance analysis will once again hold with a properly modified objective function incorporating said formula.

The consideration of ASLs also significantly complicates the construction of assignment rules.  In general, the aggregate state associated with a query would need to include each queried server's class together with the ASL of each job at that server.  Moreover, a job's expected remaining size can be non-monotonic in that job's ASL, which may obstruct straightforward attempts at pruning the space of assignment rules.  Note, however, that some scheduling rules can greatly simplify the aggregate query states that one can encounter: for example, under any non-preemptive scheduling rule (e.g., $\fcfs$) at most one job at each server will have a non-zero ASL.

\end{document}